\documentclass[aps,pra,groupedaddress,amsmath,amssymb,tightenlines,notitlepage,nofootinbib,onecolumn,superscriptaddress,11pt]{revtex4-1}
\pdfoutput=1

\PassOptionsToPackage{pdftitle={On the optimal error exponents for classical and quantum antidistinguishability}}{hyperref}
\usepackage{bm,bbm}%
\usepackage{braket}
\usepackage{amsthm,amsmath,amssymb}

\usepackage[utf8]{inputenc}

\usepackage[dvipsnames,svgnames]{xcolor}
\usepackage{amsmath,amssymb}
\usepackage{changepage,afterpage}
\usepackage{etruscan}
\usepackage{tikz}
\usepackage{enumerate}
\usepackage{mathtools}
\usepackage[colorlinks=true, allcolors=MidnightBlue!70!black!70!TealBlue]{hyperref}

\usepackage{array,booktabs,multirow}
\usepackage{siunitx}

\usepackage{newpxtext,newpxmath}

\newtheorem{theorem}{Theorem}

\newtheorem{proposition}[theorem]{Proposition}
\newtheorem{corollary}[theorem]{Corollary}
\newtheorem{definition}[theorem]{Definition}
\newtheorem{lemma}[theorem]{Lemma}

\theoremstyle{definition}
\newtheorem{remark}{Remark}

\newtheoremstyle{red}{}{}{\normalfont}{}{\color{red!80!black}\bfseries}{.}{ }{}
\theoremstyle{red}

\let\nc\newcommand

\nc{\NegW}{W_\tau}

\nc{\RW}{\Omega_\FF}

\DeclareMathOperator{\Tr}{Tr}

\nc{\id}{\mathbbm{1}}

\renewcommand{\O}{\mathcal{O}}

\newcommand{\DD}{\mathbb{D}}

\newcommand{\FF}{\mathcal{F}}

\nc{\Dmax}{D_{\max}}

\let\O\OO

\nc{\SEP}{\mathrm{SEP}}
\nc{\STAB}{\mathrm{STAB}}
\nc{\PPT}{\mathrm{PPT}}
\nc{\PPTP}{\mathrm{PPTP}}
\nc{\SEPP}{\mathrm{SEPP}}
\nc{\SP}{\mathrm{KP}}
\nc{\FP}{\mathrm{FP}}
\nc{\CPTP}{{\mathrm{CPTP}}}

\nc{\Op}{\mathcal{O}}

\nc{\idc}{\mathrm{id}}
\nc{\ve}{\varepsilon}

\nc{\Omax}{\O_{\mathrm{max}}}

\usepackage{stackengine,moresize}
\stackMath

\nc{\sminfty}{{\infty,\bullet}}

\usepackage[most,breakable]{tcolorbox}
\renewenvironment{boxed}[1]%
  {\expandafter\ifstrequal\expandafter{#1}{orange}{\begin{tcolorbox}[colback=orange!5,colframe=orange!15,breakable,enhanced]}{\begin{tcolorbox}[colback=white,colframe=gray!10,breakable,enhanced]}}%
  {\end{tcolorbox}}

\nc{\regrob}{\DD_{s,\FF}^{\infty}}
\nc{\regrobs}{\DD_{s,\FF}^{\sminfty}}

\makeatletter 
\renewcommand\onecolumngrid{%
\do@columngrid{one}{\@ne}%
\def\set@footnotewidth{\onecolumngrid}
\def\footnoterule{\kern-6pt\hrule width 1.5in\kern6pt}%
}
\makeatother

\usepackage{fmtcount,refcount}

\let\oldproofname\proofname
\renewcommand{\proofname}{\rm\bf{\oldproofname}}

\usepackage{pifont}

\usepackage{newunicodechar}

\newtheorem{example}[theorem]{Example}

\def\Tr{\operatorname{Tr}}

\def\id{\operatorname{id}}

\def\1{\openone}

\def\d{\textbf{d}}
\def\keywords{\xdef\@thefnmark{}\@footnotetext}

\def\DD{D}

\allowdisplaybreaks

\begin{document}

\title{On the optimal error exponents for classical and quantum antidistinguishability}

 \author{Hemant K. Mishra}
\email{hemant.mishra@cornell.edu}
  \affiliation{School of Electrical and Computer Engineering, Cornell University, Ithaca, New York 14850, USA}

\author{Michael Nussbaum}
\email{nussbaum@math.cornell.edu}
\affiliation{Department of Mathematics, Cornell University, Ithaca, New York 14850, USA}

\author{Mark M. Wilde}
\email{wilde@cornell.edu}
 \affiliation{School of Electrical and Computer Engineering, Cornell University, Ithaca, New York 14850, USA}
 

\begin{abstract}
            The concept of antidistinguishability of quantum states has been studied to investigate foundational questions in quantum mechanics. It is also called quantum state elimination, because the goal of such a protocol is to guess which state, among finitely many chosen at random, the system is not prepared in (that is, it can be thought of as the first step in a process of elimination). Antidistinguishability has been used to investigate the reality of quantum states, ruling out $\psi$-epistemic ontological models of quantum mechanics [Pusey~\textit{et~al.}, \textit{Nat.~Phys.}, 8(6):475-478, 2012]. Thus, due to the established importance of antidistinguishability in quantum mechanics, exploring it further is warranted. 
            
            In this paper, we provide a comprehensive study of the optimal error exponent---the rate at which the optimal error probability vanishes to zero asymptotically---for classical and quantum antidistinguishability. 
            We derive an exact expression for the optimal error exponent in the classical case and show that it is given by the {\it multivariate classical Chernoff divergence}. Our work thus provides this  divergence with a meaningful operational interpretation as the optimal error exponent for antidistinguishing a set of probability measures.
            For the quantum case, we provide several bounds on the optimal error exponent: a lower bound given by the best pairwise Chernoff divergence of the states, a single-letter semi-definite programming upper bound, and
            lower and upper bounds in terms of minimal and maximal {\it multivariate quantum Chernoff divergences}. It remains an open problem to obtain an explicit expression for the optimal error exponent for quantum antidistinguishability.
\end{abstract}

\maketitle

\noindent \textit{Dedicated to the memory of Mary Beth Ruskai. She was an important foundational figure in the field of quantum information, and her numerous seminal research contributions and reviews, including \cite{LR73,ruskai2002inequalities,horodecki2003entanglement}, have inspired many quantum information scientists.}

\tableofcontents

\section{Introduction}

Quantum state discrimination is a fundamental component of quantum information science, which plays a key role in quantum computing \cite{bacon2006}, quantum communication \cite{helstrom1969quantum, Barnett09}, and  quantum key distribution \cite{bae2015quantum}.
The state discrimination or distinguishability task is to infer the actual state of a quantum system by applying a quantum measurement to the system. More formally, consider a quantum system prepared in one of the quantum states $\rho_1,\ldots,\rho_r$.
A quantum measurement is specified by a positive operator-valued measure  $\{M_1,\ldots, M_r\}$ with output~$i$ indicating $\rho_i$ as the true state of the system with success probability $\Tr[M_i\rho_i]$, as given by the Born rule \cite{born1926quantum}.

The task that we consider here is in a sense opposite to the aforementioned task of distinguishability, and it is thus called \textit{antidistinguishability} of quantum states or \textit{quantum state elimination} \cite{caves2002conditions, pusey2012reality, barrett2014no, leifer2014quantum, bandyopadhyay2014conclusive, havlivcek2020simple, leifer2020noncontextuality, russo2023}. 
In particular, for the task of antidistinguishability, we are interested in designing a measurement whose outcome corresponds to a state that is not the actual state of the quantum system. 
In the classical version of the antidistinguishability problem, quantum states are replaced by probability measures on a measurable space, and the task is to rule out one of the probability measures upon observing i.i.d.~(independent and identically distributed) data that is not produced by the probability measure.

As an illustrative example in the classical case, suppose that one of three possible dice is tossed, a red one with probability distribution $p_R$, a green one with probability distribution $p_G$, or a blue one with probability distribution $p_B$. The task is then, after observing a sample, to output ``not red'' if the green or blue die is tossed, ``not green'' if the red or blue die is tossed, and ``not blue'' if the red or green die is tossed. It is also of interest to consider the antidistinguishability task when the same colored die is tossed multiple times, leading to several samples that one can use to arrive at a conclusion.

To the best of our knowledge, an analysis of the asymptotics of the error probability of antidistinguishability is missing in the literature for both cases, classical as well as quantum,  and it is this scenario that we consider in our paper.

\subsection{Contributions}

In this paper, we provide a comprehensive study of the optimal error exponent---the rate at which the optimal error probability vanishes to zero asymptotically---for classical and quantum antidistinguishability. 
\begin{itemize}
\item 
We derive an exact expression for the optimal error exponent in the classical case and show that it is given by the {\it multivariate classical Chernoff divergence} (Theorem~\ref{thm:optimal-classical-general}). 
Our work thus provides this multivariate divergence with a meaningful operational interpretation as the optimal error exponent for antidistinguishing a set of probability measures.

\item 
We provide several bounds on the optimal error exponent in the quantum case:
\begin{itemize}
\item[$\circ$] lower bound given by the best pairwise Chernoff divergence of the states (Theorem~\ref{newthm:lower_bound_optimal_error}),
\item [$\circ$] single-letter semi-definite programming upper bound (Theorem~\ref{thm:improved-one-letter-upper-bound}), and
\item[$\circ$]
lower and upper bounds in terms of minimal and maximal {\it multivariate quantum Chernoff divergences} (Theorem~\ref{thm:min-max-bounds-optimal-error-exponent}).
\end{itemize}

\item We also provide an upper bound on the optimal error probability of antidistinguishing an ensemble of         quantum states in terms of the pairwise optimal error probabilities of the
    states, and consequently, we deduce that the given quantum states are perfectly antidistinguishable if at least two of them are orthogonal to each other (Theorem~\ref{thm:upper_bound_optimal_error}).

    \item As a contribution of independent interest and auxiliary to Theorem~\ref{thm:improved-one-letter-upper-bound}, we establish several fundamental properties of the extended max-relative entropy, a quantity of interest originally defined in \cite{ww2020}.
\end{itemize}
\noindent It remains an intriguing open problem to determine an explicit expression for the optimal error exponent in the quantum case.

\subsection{Literature review}

Let us briefly review some prior contributions to the topic of antidistinguishability. We note here that quantum state discrimination is equivalent to finding a size-$(r-1)$ subset of $\{\rho_1,\ldots,\rho_r\}$ such that none of the states in the subset is the true state of the system; thus, the task is equivalent to what is called \textit{quantum $(r-1)$-state exclusion}. 
A generalization of this task is \textit{quantum $m$-state exclusion} for $1\leq m \leq r-1,$ which aims at detecting a size-$m$ subset of $\{\rho_1,\ldots,\rho_r\}$ such that none of the states in the subset is the true state of the system \cite{russo2023}.
Quantum $1$-state exclusion is therefore the same as antidistinguishability of quantum states.

The concept of antidistinguishability has been studied to investigate foundational questions in quantum mechanics \cite{caves2002conditions, pusey2012reality, barrett2014no, leifer2014quantum}. For example, it was used in \cite{pusey2012reality} to investigate the reality of quantum states, ruling out $\psi$-epistemic ontological models of quantum mechanics. It was also used in studying quantum communication complexity \cite{havlivcek2020simple}, in deriving noncontextuality inequalities \cite{leifer2020noncontextuality}, and has applications in quantum cryptography \cite{Collins2014}.
Thus, due to the established importance of antidistinguishability in quantum mechanics, exploring it further is warranted.
There have been a number of works that determine algebraic conditions on a set of quantum states such that perfect antidistinguishability is possible. A sufficient condition for perfect antidistinguishability of pure states \cite{heinosaari2018antidistinguishability} is that if some positive linear combination of the pure states is a projection with a ``special'' kernel, then the states are \textit{antidistinguishable}. 
In the same paper, a necessary and sufficient condition for antidistinguishability of pure states was given, which demands the existence of projections satisfying three non-trivial conditions. 
 Very recently, a necessary and sufficient condition for non-antidistinguishability of general quantum states was given in \cite{russo2023}, which also demands the existence of a Hermitian matrix with positive trace satisfying a set of non-trivial inequalities.
 Even though the conditions given in the aforementioned works are interesting and insightful, verifying them is not straightforward. 
One of the consequences of our work is that we provide a simple sufficient condition for perfect antidistinguishability of quantum states (Theorem~\ref{thm:upper_bound_optimal_error}).

\subsection{Paper organization}

The organization of our paper is as follows. In Section~\ref{sec:mathback}, we state some definitions and provide a brief mathematical background of relevant topics covered in our paper. 
We start Section~\ref{sec:classical-optimal-error-exponent} by building a theory of classical antidistinguishability, where we introduce the notions of optimal error probability and optimal error exponent. We then derive an explicit expression for the optimal error exponent in the classical case, and we show that it is given by the {\it multivariate classical Chernoff divergence} (Theorem~\ref{thm:optimal-classical-general}). The following sections deal with the optimal error exponent in the quantum case.
We begin Section~\ref{sec:quantum-achievable-error-exponent} by providing an upper bound on the optimal error probability of antidistinguishing an ensemble of quantum states in terms of the pairwise optimal error probabilities of the states (Theorem~\ref{thm:upper_bound_optimal_error}),
and we then use this result to derive a lower bound on the optimal error exponent (Theorem~\ref{newthm:lower_bound_optimal_error}). Next, we provide both lower and upper bounds on the optimal error exponent in terms of minimal and maximal {\it multivariate quantum Chernoff divergences} in Section~\ref{sec:bounds_optimal_error_exponent_xi_max_min} (Theorem~\ref{thm:min-max-bounds-optimal-error-exponent}). 
Lastly, in Section~\ref{sec:one-letter-upper-bound-optimal-error-exponent}, we derive a single-letter semi-definite programming upper bound on the optimal error exponent (Theorem~\ref{thm:improved-one-letter-upper-bound}).
Appendices~\ref{app:explicit-form-gamma-i-non-corner-points} through~\ref{app:max-rel-entropy-bounds-similarity} contain mathematical proofs of various claims made throughout the paper.


\section{Mathematical background}
\label{sec:mathback}


\subsection{Antidistinguishability of probability measures}

Let $P_1,\ldots, P_{r}$ be probability measures on a measurable space $(\Omega, \mathcal{A})$, where $\mathcal{A}$ is a $\sigma$-algebra on the set $\Omega$. 
Set $[r]\coloneqq \{1,\ldots,r\}$.
Let $\eta_1,\ldots, \eta_r$ be positive real numbers such that $\sum_{i\in[r]}\eta_i=1$.  Throughout the paper, we call
\begin{equation}
\mathcal{E}_{\operatorname{cl}} \coloneqq \{(\eta_i, P_i): i \in [r]\}  
\end{equation}
an ensemble of probability measures on the measurable space $(\Omega, \mathcal{A})$.
Let $\mu$ be the dominating measure
\begin{align}
    \mu \coloneqq \sum_{i\in [r]}\eta_i P_i, \label{eq:dominating-measure}
\end{align}
and $p_1,\ldots, p_r$ the induced densities
\begin{align}
    p_i \coloneqq \dfrac{\d P_i}{\d \mu}, \qquad i\in[r]\label{eq:induced-densities},
\end{align}
which are given by the Radon--Nikod\'ym theorem \cite{billingsley1995probability}.

The problem of distinguishability, i.e., identifying the correct probability density $p_i$ based on i.i.d.~(independent and identically distributed) data, has been well studied. This problem is as follows:
Suppose that $i$ is sampled with probability $\eta_i$, and then $n$ i.i.d.~samples are selected according to the product measure $P_i^{\otimes n}$. The task is to identify the correct value of $i$ based on the $n$ i.i.d.~samples observed.
It is known that the maximum likelihood method for the identification task is optimal, and the optimal success probability, in the case that $n=1$, is given by 
\begin{align}\label{eq:maximum_likelihood_principle}
    \int\!\!\d\mu\ \left( \eta_1 p_1 \vee \cdots \vee \eta_r p_r \right) \coloneqq \int\!\!\d\mu(\omega)\ \max \{\eta_1 p_1(\omega),\ldots, \eta_r p_r(\omega)\}.
\end{align}
Asymptotically, the optimal error vanish to zero exponentially and the error exponent is known to be equal to the Chernoff divergence for the least favorable pair $(p_i,p_j)$, for $i \neq j$ \cite{salikhov1973asymptotic,  torgersen1981measures, Leang1997, salikhov1999one, salikhov2003optimal}.

For the antidistinguishability problem in the classical case, the task is to guess a probability density that is not represented by the observed data. For this problem, no literature is available to the best of our knowledge. 
A reasonable first idea for selecting a density that is unlikely to be the true one is to choose the one such that $\eta_i p_i(\omega)$ is minimum if $\omega$ is observed.
This corresponds to a \textit{minimum likelihood principle}. In what follows, we discuss this idea more formally.

A deterministic decision rule for the antidistinguishability problem is a function 
\begin{align}
    \delta: \Omega \to \{\mathbf{e}_{i}: i\in[r]\},
\end{align}
where $\mathbf{e}_i$ is the $i$th standard unit vector in $\mathbb{R}^r$, such that $\delta(\omega)=\mathbf{e}_{i}$ means that we indicate $p_{i}$ to be our guess for the density that is not the true one. More generally, we can admit a randomized decision rule, along the following lines:
\begin{align}
    \delta: \Omega \to [0,1]^r,\qquad \sum_{i \in [r]} \delta_i(\omega)=1.
\end{align}
If $p_{i}$ is the true density, then the antidistinguishability error probability is given by 
\begin{align}
    \int\!\!\d\mu(\omega)\ \delta_{i}(\omega)p_{i}(\omega), 
\end{align}
and the total error probability is
\begin{align}  \operatorname{Err}_{\operatorname{cl}}(\delta; \mathcal{E}_{\operatorname{cl}}) &\coloneqq  \sum_{i \in [r]} \eta_i \int\!\!\d\mu(\omega)\ \delta_{i}(\omega)p_{i}(\omega)=  \int\!\!\d\mu(\omega)\ \sum_{i \in [r]} \delta_{i}(\omega) \eta_i p_{i}(\omega).
\end{align}
To minimize the above expression, we can minimize the integrand for every $\omega$. Since $\delta_i(\omega)$ is a weight, we should place maximum weight on the smallest of $\eta_i p_i(\omega)$. So, the optimal decision for given $\omega$ corresponds to the \textit{minimum likelihood rule}:   $\delta^{*}(\omega) = \mathbf{e}_i$, if $i \in [r]$ is the minimum index such that $\eta_ip_i(\omega)=\min\{\eta_1p_1(\omega),\ldots, \eta_rp_{r}(\omega)\}$. The total error probability when using the decision rule $\delta^*$ is the optimal error probability, given by
\begin{align}
\operatorname{Err}_{\operatorname{cl}}(\mathcal{E}_{\operatorname{cl}}) \coloneqq \operatorname{Err}_{\operatorname{cl}}(\delta^{*}; \mathcal{E}_{\operatorname{cl}}) 
&=  \int\!\!\d\mu(\omega)\ \min\{\eta_1 p_{1}(\omega),\ldots, \eta_r p_{r}(\omega)\} = \int\!\!\d\mu\ \left(\eta_1 p_{1} \land \cdots \land \eta_r p_{r}\right).\label{eq:optimal-AD-error-probability}
\end{align}

In the asymptotic treatment of the problem, we consider the \textit{$n$-fold ensemble} $\mathcal{E}_{\operatorname{cl}}^{n}\coloneqq \{(\eta_i, P_i^{\otimes n}): i \in [r]\}$ on the $n$-fold measurable space $(\Omega^n, \mathcal{A}^{(n)})$, where $\Omega^n$ is the $n$-fold Cartesian product of $\Omega$ and $\mathcal{A}^{(n)}$ is the $\sigma$-algebra on $\Omega^n$ generated by the $n$-fold Cartesian product of  $\mathcal{A}$. It then follows that the optimal antidistinguishability error probability, in this case, is
\begin{align}
\operatorname{Err}_{\operatorname{cl}}(\mathcal{E}_{\operatorname{cl}}^{n}) =  \int\!\!\d\mu^{\otimes n}\ \left(\eta_1 p_{1}^{\otimes n}\land \cdots \land \eta_r p_{r}^{\otimes n}\right).\label{eq:n-fold-ensemble-optimal-error-probability}
\end{align}
Set $\eta_{\operatorname{min}}\coloneqq \min\{\eta_1,\ldots, \eta_r\}$ and $\eta_{\operatorname{max}}\coloneqq \max\{\eta_1,\ldots, \eta_r\}$. It is easy to see that for all $n \in \mathbb{N}$, we have
\begin{align}
    \eta_{\operatorname{min}} \int\!\!\d\mu^{\otimes n}\ \left( p_{1}^{\otimes n}\land \cdots \land  p_{r}^{\otimes n}\right) \leq \operatorname{Err}_{\operatorname{cl}}(\mathcal{E}_{\operatorname{cl}}^{n}) \leq \eta_{\operatorname{max}} \int\!\!\d\mu^{\otimes n}\ \left( p_{1}^{\otimes n}\land \cdots \land  p_{r}^{\otimes n}\right).
\end{align}
This implies
\begin{align}
    \liminf_{n \to \infty} -\dfrac{1}{n}\ln \operatorname{Err}_{\operatorname{cl}}(\mathcal{E}_{\operatorname{cl}}^{n}) = \liminf_{n \to \infty} -\dfrac{1}{n}\ln \int\!\!\d\mu^{\otimes n}\ \left( p_{1}^{\otimes n}\land \cdots \land  p_{r}^{\otimes n}\right),\label{eq:optimal-error-exponent-ind-of-eta}
\end{align}
which is independent of $\eta_1,\ldots, \eta_r$.
\begin{boxed}{white}
    \begin{definition}
        The optimal error exponent for antidistinguishing the probability measures of a given ensemble $\mathcal{E}_{\operatorname{cl}}=\{(\eta_i, P_i): i \in [r]\}$ is defined by
        \begin{align}\label{eq:asy-error-rate-def}
            \operatorname{E}_{\operatorname{cl}}(P_1,\ldots, P_r) \coloneqq \liminf_{n \to \infty} -\dfrac{1}{n}\ln \operatorname{Err}_{\operatorname{cl}}(\mathcal{E}_{\operatorname{cl}}^{n}).
        \end{align}
    \end{definition}
\end{boxed}
\begin{remark}\label{rem:ind-of-dom-measure}
    We note that the above definition of the optimal error exponent is independent of the choice of the dominating measure $\mu$.
    This is because the development in \eqref{eq:maximum_likelihood_principle}--\eqref{eq:optimal-error-exponent-ind-of-eta} is independent of the choice of the probability measure $\mu$ dominating $P_1,\ldots, P_r$. Indeed, if $\mu^{\prime}$ is an arbitrary probability measure dominating $P_1,\ldots, P_r$, then $\mu^{\prime}$ also dominates $\mu$. Let $\nu \coloneqq \frac{\d \mu}{\d \mu^{\prime}}$. We then have 
    \begin{align}
        p_i^{\prime} 
            \coloneqq \frac{\d P_i}{\d \mu^{\prime}} 
            = \frac{\d P_i}{\d \mu} \cdot \frac{\d \mu}{\d \mu^{\prime}} 
            =  p_i \nu, \quad \text{for all } i \in [r].
    \end{align}
    Consequently, the quantity in \eqref{eq:maximum_likelihood_principle} is given by
    \begin{align}
        \int\!\!\d\mu\ \left( \eta_1 p_1 \vee \cdots \vee \eta_r p_r \right) 
            &= \int\!\!\d\mu^{\prime}\  \left( \eta_1  p_1 \vee \cdots \vee \eta_r  p_r \right)\nu \\
            &= \int\!\!\d\mu^{\prime}\  \left( \eta_1  p_1 \nu \vee \cdots \vee \eta_r  p_r \nu \right) \\
            &= \int\!\!\d\mu^{\prime}\ \left( \eta_1  p_1^{\prime} \vee \cdots \vee \eta_r  p_r^{\prime} \right).
    \end{align}
      Similarly, the remaining quantities in \eqref{eq:maximum_likelihood_principle}--\eqref{eq:optimal-error-exponent-ind-of-eta} can be shown to be independent of the choice of $\mu$. See \cite[p.~233]{shiryaev2016probability}.
\end{remark}

\subsection{Multivariate classical Chernoff divergence}

In 1909, Hellinger introduced a function of multiple probability distributions \cite{hellinger1909neue}. Known as the {\it Hellinger transform} in the literature on probability and statistics, this quantity plays an important role in our work. We recall its definition below.
Let $\mathcal{E}_{\operatorname{cl}}=\{(\eta_i, P_i): i \in [r]\}$ be an ensemble of probability measures on a measurable space $(\Omega, \mathcal{A})$. Let $\mu$ be the dominating measure defined by \eqref{eq:dominating-measure}, and let $p_1,\ldots,p_r$ be the induced probability densities given by \eqref{eq:induced-densities}.
Denote by $\mathbb{S}_r$ the unit simplex in $\mathbb{R}^{r}$:
\begin{align}\label{eq:unit_simplex}
    \mathbb{S}_r &\coloneqq \bigg\{ \mathbf{s}\in [0,1]^r: \mathbf{s}=(s_1,\ldots, s_r), \sum_{i\in [r]} s_i = 1 \bigg\}.
\end{align}
\begin{boxed}{white}
    \begin{definition}\label{def:hellinger-transform}
    The Hellinger transform of the probability measures $P_1,\ldots, P_r$ is a function on the unit simplex, defined as
        \begin{align}
    \mathbb{H}_\mathbf{s}(P_1,\ldots,P_r) \coloneqq \int\!\!\d\mu\  p_1^{s_1}\cdots p_r^{s_r}, \qquad \text{for } \mathbf{s}\coloneqq(s_1,\ldots, s_r) \in \mathbb{S}_r.
    \label{eq:hellinger-trans-def}
        \end{align}
    Here we use the convention $0^0=1$ of \cite[Definition~5.10]{strasser2011mathematical}.
\end{definition}
\end{boxed}
\noindent The term ``Hellinger transform'' was perhaps first used in \cite{lecam1970}, followed by several works in the area of probability and statistics. See \cite{toussaint1974some, torgersen1981measures, JACOD19893, torgersen1991comparison, Grigelionis1993, fazekas1996some, liese2010statistical, strasser2011mathematical} and references therein.
\begin{remark}
    We emphasize that the Hellinger transform given in Definition~\ref{def:hellinger-transform}, and hence our further analysis of the classical antidistinguishability error probability, is independent of the choice of the dominating measure $\mu$. This easily follows by similar arguments as in Remark~\ref{rem:ind-of-dom-measure}. See \cite[Chapter~3, Section~9, Lemma~3]{shiryaev2016probability} for a proof in the case of $r=2$.
\end{remark}

If $P_1,\ldots, P_r$ are mutually absolutely continuous, i.e., for all $A \subseteq \Omega$, $P_1(A)=0$ if and only if $P_i(A)=0$ for all $i \in [r]$, then the Hellinger transform is continuous on $\mathbb{S}_r$. Indeed, we have $P_i \leq \eta_i^{-1} \mu$, implying that $p_i \leq \eta_i^{-1}$ for all $i \in [r]$. This gives $\prod_{i \in [r]} (p_i+1)$ as an integrable upper bound on $\prod_{i \in [r]} p_i^{s_i}$ for all $(s_1,\ldots, s_r) \in \mathbb{S}_r$. 
If the probability measures are mutually absolutely continuous, then 
\begin{align}
    \mu(\{\omega \in \Omega: p_i(\omega)>0 \ \forall i \in [r]\})=1.
\end{align}
 Without loss of generality, we can then assume that $p_i(\omega)>0$ for all $\omega \in \Omega$ and $i \in [r]$. Thus, for every $\mathbf{s}\coloneqq (s_1,\ldots, s_r) \in \mathbb{S}_r$ and for all $n \in \mathbb{N}$ and every sequence $\mathbf{s}^{(n)}\coloneqq (s_1^{(n)},\ldots, s_r^{(n)})$  in the interior of $\mathbb{S}_r$ such that $\lim_{n \to \infty} \mathbf{s}^{(n)}=\mathbf{s}$, we have
\begin{align}
    \lim_{n \to \infty} \prod_{i \in [r]} \left(p_i(\omega)\right)^{s_i^{(n)}} = \prod_{i \in [r]}\left(p_i(\omega)\right)^{s_i}, \qquad \text{for } \omega \in \Omega.
\end{align}
By the Lebesgue dominated convergence theorem, we then have $\lim_{n \to \infty} \mathbb{H}_{\mathbf{s}^{(n)}}(P_1,\ldots, P_r) = \mathbb{H}_{\mathbf{s}}(P_1,\ldots, P_r)$, thereby proving continuity of the Hellinger transform on $\mathbb{S}_r$.

In general, the Hellinger transform is a measure of ``closeness'' among several probability distributions.
 It is easy to see that
\begin{align}
    0 \leq \mathbb{H}_\mathbf{s}(P_1,\ldots,P_r) \leq 1,
\end{align}
which follows from H\"older's inequality \cite[Lemma~53.3]{strasser2011mathematical}.
As the value of $\mathbb{H}_\mathbf{s}(P_1,\ldots,P_r)$ gets close to $0$, the distance among the measures increases in some sense \cite{fazekas1996some}. 

The following quantity plays an important role in our paper.
\begin{boxed}{white}
    \begin{definition}
        We define the multivariate Chernoff divergence of the probability measures $P_1,\ldots, P_r$ by
        \begin{align}
            \xi_{\operatorname{cl}}(P_1,\ldots,P_r)\coloneqq-\ln \inf_{\mathbf{s}\in \mathbb{S}_r} \mathbb{H}_\mathbf{s}(P_1,\ldots,P_r),    \label{eq:classical-CH-divergence}
        \end{align}
        where $\mathbb{H}_\mathbf{s}$ is defined in \eqref{eq:hellinger-trans-def}.
    \end{definition}
\end{boxed}
 The divergence can be viewed as a generalization of the classical Chernoff divergence, the latter being a special case of the former for $r=2$ \cite{chernoff1952}. One of the main results of our paper is that the optimal error exponent for antidistinguishing an ensemble of probability measures is equal to their multivariate Chernoff divergence (Theorem~\ref{thm:optimal-classical-general}).

\subsection{Quantum states, channels, and measurements}

A quantum system is associated with a complex Hilbert space $\mathcal{H}$. We focus exclusively on systems with finite-dimensional Hilbert spaces in this paper. Let $\dim(\mathcal{H})$ denote the dimension of~$\mathcal{H}$.
We denote every element of $\mathcal{H}$ using the \textit{ket} notation as $|\psi \rangle, |\phi \rangle$, etc., and every element of its dual using the \textit{bra} notation as $\langle\psi|, \langle \phi|$, etc. The notations go well with the natural action of a dual element $\langle\psi|$ on a vector $|\phi \rangle $ in terms of the inner product of the two vectors:  $\langle\psi| (|\phi \rangle)=\langle\psi|\phi \rangle$.

A quantum state of a system is identified by a density operator $\rho$, which is a self-adjoint, positive semi-definite operator of unit trace acting on $\mathcal{H}$. A pure state is given by a state vector $|\psi \rangle \in \mathcal{H}$ whose corresponding density operator is $|\psi\rangle\!\langle\psi|$.
The set of density operators forms a convex set with pure states as the extreme points. Let $\mathcal{D}(\mathcal{H})$ denote the set of density operators and $\mathcal{L}(\mathcal{H})$ the space of linear operators acting on $\mathcal{H}$. We shall use the notation $\mathcal{D}$ for the set of density operators whenever the underlying Hilbert space is clear from the context. A quantum channel~$\mathcal{N}$, between two quantum systems represented by Hilbert spaces $\mathcal{H}$ and $\mathcal{K}$,  is a completely positive, trace-preserving linear map from $\mathcal{L}(\mathcal{H})$ to $\mathcal{L}(\mathcal{K})$. In particular, for all $\rho \in \mathcal{D}(\mathcal{H})$, we have that $\mathcal{N}(\rho) \in \mathcal{D}(\mathcal{K})$.

A quantum measurement is described by a positive operator-valued measure (POVM) $\mathscr{M}=\{M_1, \ldots, M_r\}$, which is a finite set of positive semi-definite operators whose sum is the identity operator, i.e.,
\begin{align}
    M_i \geq 0 \text{ for all } i \in [r], \qquad \qquad \sum_{i\in [r]} M_i = \mathbb{I},
\end{align}
where $\mathbb{I}$ is the identity operator acting on $\mathcal{H}$.

The projection onto the support of an operator $A$ is denoted by $\operatorname{supp}(A)$, its absolute value is denoted by $|A|\coloneqq \sqrt{A^\dag A}$, and its positive part by $A_+ \coloneqq \frac{1}{2} (A + |A|)$. For two Hermitian operators~$A$ and $B$, we use the notation
 \begin{equation}
     A \land B \coloneqq \frac{1}{2}(A+B - |A-B|),
     \label{eq:operator-min}
 \end{equation}
 in analogy with $\min (a,b)=\frac{1}{2}(a+b - |a-b|) \equiv a \land b$ for $a,b \in \mathbb{R}$.



\subsection{Antidistinguishability of quantum states}

Suppose that a quantum system is prepared in one of the quantum states $\rho_1,\ldots, \rho_r$ with \textit{a priori} probability distribution $\eta_1,\ldots, \eta_r$ such that $\eta_i>0$ for all $i \in [r]$. Throughout the paper, we call $\{(\eta_i, \rho_i): i \in [r]\}$ an ensemble of quantum states over a Hilbert space $\mathcal{H}$ and denote it by $\mathcal{E}$. Antidistinguishability of the states performed by a POVM $\mathscr{M}=\{M_1,\ldots, M_r\}$ can be described as follows: ``the measurement outcome $i$ occurring corresponds to a guess that the true state of the system is not $\rho_i$.'' Thus, if $\rho_i$ is the true state of the system, then $\Tr[M_i \rho_i]$ is the error probability for making an incorrect guess. The average error probability of antidistinguishability, for a fixed POVM $\mathscr{M}$, is then given by
\begin{align}
    \operatorname{Err}(\mathscr{M}; \mathcal{E}) \coloneqq     \sum_{i\in [r]} \eta_i \Tr[M_i \rho_i].
\end{align}
We are interested in determining the optimal antidistinguishability error probability, which is optimized over all possible measurements:
\begin{align}\label{eq:optimal-error-ad}
            \operatorname{Err}(\mathcal{E}) \coloneqq \inf_{\mathscr{M}} \operatorname{Err}(\mathscr{M}; \mathcal{E}),
\end{align}
where the infimum is taken over all POVMs of the form $\mathscr{M}=\{M_1,\ldots, M_r\}$ acting on $\mathcal{H}$.

The quantum states are said to be perfectly antidistinguishable if there exists a quantum measurement whose outputs always correspond to a false state of the system; i.e., there exists a POVM~$\mathscr{M}$ such that $\operatorname{Err}(\mathscr{M}; \mathcal{E})=0$.  In general, an ensemble of quantum states may not be antidistinguishable, which means, for such an ensemble $\mathcal{E}$, that $\operatorname{Err}(\mathscr{M}; \mathcal{E})>0$ for every POVM~$\mathscr{M}$. For instance, two non-orthogonal quantum states are not perfectly antidistinguishable~\cite{leifer2014quantum}.

In the asymptotic treatment of the antidistinguishability problem for a given ensemble $\mathcal{E}=\{(\eta_i, \rho_i): i \in [r]\}$, we consider the {\it $n$-fold ensemble} $\mathcal{E}^n \coloneqq \{(\eta_i, \rho_i^{\otimes n}):i \in [r]\}$. The optimal error probability of antidistinguishability for $\mathcal{E}^n$ is by definition given as
\begin{align}
    \operatorname{Err}(\mathcal{E}^n) = \inf_{ \mathscr{M}^{(n)}} \sum_{i\in[r]} \Tr\!\left[\eta_i M_i^{(n)} \rho_i^{\otimes n} \right],\label{eq:optimal-error-n-fold}
\end{align}
where the infimum is taken over the set of POVMs $\mathscr{M}^{(n)}=\{M_1^{(n)},\ldots, M^{(n)}_r\}$ acting on the $n$-fold tensor product Hilbert space $\mathcal{H}^{\otimes n}$. 
This gives for all $n\in \mathbb{N}$,
\begin{align}
    \eta_{\operatorname{min}}  \inf_{ \mathscr{M}^{(n)}} \sum_{i\in[r]} \Tr\!\left[ M_i^{(n)} \rho_i^{\otimes n} \right] \leq \operatorname{Err}(\mathcal{E}^n) \leq \eta_{\operatorname{max}} \inf_{ \mathscr{M}^{(n)}} \sum_{i\in[r]} \Tr\!\left[ M_i^{(n)} \rho_i^{\otimes n} \right],
\end{align}
which implies that
\begin{align}\label{eq:optimal-error-exponent-independent-prior}
    \liminf_{n\to \infty} -\dfrac{1}{n} \ln \operatorname{Err}(\mathcal{E}^n) =  \liminf_{n\to \infty} -\dfrac{1}{n} \ln \inf_{ \mathscr{M}^{(n)}} \sum_{i\in[r]} \Tr\!\left[ M_i^{(n)} \rho_i^{\otimes n} \right],
\end{align}
the latter being independent of $\eta_1,\ldots, \eta_r$.
\begin{boxed}{white}
    \begin{definition}
        The optimal error exponent for antidistinguishing the states of a quantum ensemble $\mathcal{E}=\{(\eta_i, \rho_i): i \in [r]\}$ is defined by
        \begin{align}\label{eq:qu-asy-error-rate-def}
            \operatorname{E}(\rho_1,\ldots, \rho_r) \coloneqq \liminf_{n \to \infty} -\dfrac{1}{n}\ln \operatorname{Err}(\mathcal{E}^{n}).
        \end{align}
    \end{definition}
\end{boxed}

\subsection{Quantum Chernoff divergence}

Here we briefly recall known results for distinguishability of two or more states;  a key quantity for this purpose is as follows:

\begin{boxed}{white}
    \begin{definition}
        The quantum Chernoff divergence between two states $\rho_1$ and $\rho_2$ is defined as
    \begin{align}\label{eq:quantum_chernoff_divergence}
        \xi(\rho_1,\rho_2) \coloneqq -\ln \inf_{s \in [0,1]} \Tr [\rho_1^s \rho_2^{1-s}].
    \end{align}
    \end{definition}
\end{boxed}
If $\rho_1=|\psi\rangle\!\langle \psi|$ and $\rho_2=|\phi\rangle\!\langle \phi|$ are pure states, then
\begin{align}\label{eqn:pure_qcd}
    \xi(\rho_1,\rho_2)&= -\ln  |\langle \psi|\phi\rangle|^2.
\end{align}
The quantum Chernoff divergence between two states is known to be the {\it optimal error exponent} in distinguishing them \cite{audenaert2007discriminating, ns_chernoff}; i.e.,
\begin{equation}
    \lim_{n\to\infty} - \frac{1}{n} \ln \left(\operatorname{Tr}[ \eta_1 \rho_1^{\otimes n} \land \eta_2 \rho_2^{\otimes n}]\right) = \lim_{n\to\infty} - \frac{1}{n} \ln \left(\frac{1}{2}\left[1-\left\Vert \eta_1 \rho_1^{\otimes n} - \eta_2 \rho_2^{\otimes n}\right\Vert_1\right]\right) = \xi(\rho_1,\rho_2),
    \label{eq:q-chernoff-thm}
\end{equation}
where we have used the well known fact that the optimal error probability in distinguishing $\rho_1^{\otimes n} $ from $\rho_2^{\otimes n}$ is equal to 
\begin{equation}
    \operatorname{Tr}[ \eta_1 \rho_1^{\otimes n} \land \eta_2 \rho_2^{\otimes n}] = \frac{1}{2}\left[1-\left\Vert \eta_1 \rho_1^{\otimes n} - \eta_2 \rho_2^{\otimes n}\right\Vert_1\right],
\end{equation}
with  $\rho_1^{\otimes n} $  prepared with probability $\eta_1$ and  $\rho_2^{\otimes n} $   with probability $\eta_2$ \cite{helstrom1969quantum,Hol72}.

It is known more generally that the optimal error exponent in distinguishing the ensemble $\{(\eta_i, \rho_i^{\otimes n}):i\in [r]\}$ is equal to the minimum pairwise Chernoff divergence \cite{keli}.


\section{Multivariate Chernoff divergence as the optimal error exponent for classical antidistinguishability}
\label{sec:classical-optimal-error-exponent}

In this section, we prove that the optimal error exponent for antidistinguishing an ensemble of probability measures \eqref{eq:asy-error-rate-def} is equal to the multivariate Chernoff divergence of the probability measures.
We first prove this result for mutually absolutely continuous probability measures because the proof in this case is more readable. The proof in the general case is similar, albeit more technical, and so we present it separately in Appendix~\ref{app:general-classical-AD-error-exponent}. Some aspects of 
the proof presented here follow the development in the appendix of \cite{ns_chernoff}.
\begin{boxed}{white}
\begin{theorem}\mbox{}\label{thm:optimal-classical}
        Consider an ensemble $\mathcal{E}_{\operatorname{cl}}=\{(\eta_i, P_i): i \in [r]\}$ of mutually absolutely continuous probability measures on a measurable space $(\Omega, \mathcal{A})$. The optimal error exponent for antidistinguishing the probability measures is given by their multivariate Chernoff divergence, i.e.,
        \begin{align}
            \operatorname{E}_{\operatorname{cl}}(P_1,\ldots, P_r)=\lim_{n \to \infty} -\dfrac{1}{n}\ln \operatorname{Err}_{\operatorname{cl}}(\mathcal{E}_{\operatorname{cl}}^{n}) &= \xi_{\operatorname{cl}}(P_1,\ldots,P_r),\label{eq:classical-asymptotic-error-rate}
        \end{align}
        where, recalling \eqref{eq:unit_simplex}, \eqref{eq:hellinger-trans-def}, and \eqref{eq:classical-CH-divergence}, the multivariate classical Chernoff divergence $\xi_{\operatorname{cl}}$ is defined as
        \begin{equation}
            \xi_{\operatorname{cl}}(P_1,\ldots,P_r) \coloneqq -\ln \inf_{\mathbf{s}\in \mathbb{S}_r} \int\!\!\d\mu\  p_1^{s_1}\cdots p_r^{s_r} .
        \end{equation}
    \end{theorem}
\end{boxed}
\begin{proof}
Let $\mu$ be the dominating measure given by \eqref{eq:dominating-measure} and $p_1,\ldots, p_r$ the induced densities defined in \eqref{eq:induced-densities}. 
The assumption that the probability measures are mutually absolutely continuous implies that
\begin{align}
    \mu(\{\omega \in \Omega: p_i(\omega)>0 \ \forall i \in [r]\})=1.
\end{align} 
So, without loss of generality, we assume that $p_i(\omega)>0$ for all $\omega \in \Omega$ and $i \in [r]$.

We have from \eqref{eq:n-fold-ensemble-optimal-error-probability} that 
\begin{align}
    \operatorname{Err}_{\operatorname{cl}}(\mathcal{E}_{\operatorname{cl}}^n) 
        &= \int\!\!\d\mu^{\otimes n}\ \left(\eta_1 p_1^{\otimes n} \land \cdots \land \eta_r  p_{r}^{\otimes n}\right) \leq \int\!\!\d\mu^{\otimes n}\ \left( p_1^{\otimes n} \land \cdots \land   p_{r}^{\otimes n}\right). \label{eq:upper-bound-optimal-error-n-case}
\end{align} 
Let $\mathbf{s}\in \mathbb{S}_r$ be arbitrary. We can write the right-hand side of \eqref{eq:upper-bound-optimal-error-n-case} as
\begin{align}
    \int\!\!\d\mu^{\otimes n}\ \left( p_1^{\otimes n} \land \cdots \land   p_{r}^{\otimes n}\right)
        &= \int\!\!\d\mu^{\otimes n}\ (p_1^{\otimes n} \land \cdots \land p_{r}^{\otimes n})^{s_1} \cdots (p_1^{\otimes n} \land \cdots \land p_{r}^{\otimes n})^{s_r} \\
        &\leq \int\!\!\d\mu^{\otimes n}\ (p_1^{\otimes n})^{s_1} \cdots (p_{r}^{\otimes n})^{s_r}. \label{neweqn5}
\end{align}
The expression on the right-hand side of \eqref{neweqn5} has a product structure. Indeed, for every generic $\omega^n\coloneqq(\omega_1,\ldots, \omega_n)\in \Omega^n$, we have that
\begin{align}
     \int\!\!\d\mu^{\otimes n}(\omega^n)\ (p_1^{\otimes n}(\omega^n))^{s_1} \cdots (p_r^{\otimes n}(\omega^n))^{s_r} 
        &= \int\!\! \prod_{k\in[n]} \d\mu(\omega_k)\ \left(\prod_{k\in[n]} p_1(\omega_k)\right)^{s_1} \cdots \left(\prod_{k\in[n]} p_r(\omega_k)\right)^{s_r} \\
        &= \int\!\! \prod_{k\in[n]}  \d\mu(\omega_k)\ \prod_{k\in[n]} p_1^{s_1}(\omega_k) \cdots p_r^{s_r}(\omega_k) \\
        &= \int\!\! \prod_{k\in[n]} \left( \d\mu(\omega_k)\ p_1^{s_1}(\omega_k) \cdots p_r^{s_r}(\omega_k)\right) \\
        &= \prod_{k\in[n]} \int\!\!\d\mu(\omega_k)\ (p_1^{s_1} \cdots p_r^{s_r})(\omega_k) \\
        &=\left(\int\!\!\d\mu\ p_1^{s_1} \cdots p_r^{s_r} \right)^n \\
        &= \mathbb{H}_{\mathbf{s}}(P_1,\ldots,P_r)^n.\label{neweqn6}
\end{align}
From \eqref{eq:upper-bound-optimal-error-n-case}, \eqref{neweqn5}, and \eqref{neweqn6}, we thus get
\begin{align}
    \operatorname{Err}_{\operatorname{cl}}(\mathcal{E}_{\operatorname{cl}}^n)
        &\leq \mathbb{H}_{\mathbf{s}}(P_1,\ldots,P_r)^n, \qquad \text{for all } \mathbf{s} \in \mathbb{S}_r.\label{eq:upper-bound-optimal-error-n-case-last}
\end{align}
This implies, for all $n\in \mathbb{N}$, that
\begin{align}
    -\dfrac{1}{n} \ln \operatorname{Err}_{\operatorname{cl}}(\mathcal{E}_{\operatorname{cl}}^{n}) \geq -\ln \inf_{\mathbf{s}\in\mathbb{S}_r}  \mathbb{H}_{\mathbf{s}}(P_1,\ldots,P_r)=\xi_{\operatorname{cl}}(P_1,\ldots,P_r).
\end{align}
Therefore, we get
\begin{align}
    \liminf_{n\to \infty} -\dfrac{1}{n} \ln \operatorname{Err}_{\operatorname{cl}}(\mathcal{E}_{\operatorname{cl}}^n) \geq \xi_{\operatorname{cl}}(P_1,\ldots,P_r).
\end{align}
This proves the achievability part of the optimal error exponent. 

To prove the optimality part, we apply multivariable calculus and the law of large numbers. For this purpose,
let us parameterize the unit simplex of $\mathbb{R}^r$ by the corner of the standard unit cube of $\mathbb{R}^{r-1}$, defined as
\begin{align}
    \mathbb{T}_r &\coloneqq \bigg \{ \mathbf{t} \in [0,1]^{r-1}: \mathbf{t}\coloneqq (t_1,\ldots, t_{r-1}), \sum_{i\in [r-1]} t_i \leq 1 \bigg \}.
\end{align}
The unit simplex \eqref{eq:unit_simplex} can be expressed as
\begin{align}
    \mathbb{S}_r &= \bigg\{(t_1,\ldots, t_{r-1}, 1-\sum_{i\in[r-1]}t_i): (t_1,\ldots, t_{r-1})\in \mathbb{T}_r \bigg\}.\label{eq:parametrization-unit-simplex}
\end{align}
Using the new parameterization, let us denote the elements of $\mathbb{S}_r$ by $\mathbf{s}_{\mathbf{t}} \coloneqq (t_1,\ldots, t_{r-1}, 1-\sum_{i\in[r-1]}t_i)$ for $\mathbf{t}\coloneqq (t_1,\ldots, t_{r-1})\in \mathbb{T}_r$. The Hellinger transform of $P_1,\ldots, P_r$ can then be expressed as the following function on $\mathbb{T}_r$:
\begin{align}
    \operatorname{H}(\mathbf{t}) \coloneqq \mathbb{H}_{\mathbf{s}_{\mathbf{t}}}(P_1,\ldots,P_r), \qquad \text{for } \mathbf{t}\in \mathbb{T}_r.\label{eq:definition-A(t)}
\end{align}
Thus, the multivariate Chernoff divergence of $P_1,\ldots, P_r$ has the form 
\begin{align}
    \xi_{\operatorname{cl}}(P_1,\ldots,P_r)= \sup_{\mathbf{t}\in \mathbb{T}_r} - \ln \operatorname{H}(\mathbf{t}).\label{eq:hellinger-divergence-a(t)}
\end{align}

In what follows, using the reparametrized Hellinger transform \eqref{eq:definition-A(t)}, we define an {\it exponential family} of densities $p_{\mathbf{t}}$, as given in \eqref{eq:density-pt-expression} and with $\mathbf{t} \in \mathbb{T}_r$, which enables us to express each $p_i^{\otimes n}$ in terms of $p_{\mathbf{t}}^{\otimes n}$ for all $n \in \mathbb{N}$. This then allows for the use of the law of large numbers to deduce a family of upper bounds on the asymptotic error exponent, given by $-\min_{1\leq i \leq r} \gamma_i (\mathbf{t})$ for the {\it non-corner} points in $\mathbb{T}_r$, as defined later on in \eqref{eq:def_gamma_i}. Lastly, we use multivariable calculus to prove that there exists a non-corner point $\mathbf{t}^*$ such that $\ln \operatorname{H}(\mathbf{t}^*)=\min_{1\leq i \leq r} \gamma_i (\mathbf{t}^*)$. This implies that the multivariate Chernoff divergence is the optimal bound for the asymptotic error rate.

For every $\mathbf{t} \in \mathbb{T}_r$, let us express $\operatorname{H}(\mathbf{t})$ in an \textit{exponential-integral} form as follows:
\begin{align}
    \operatorname{H}(\mathbf{t})
        &=  \int\!\!\d\mu\ p_1^{t_1}\cdots p_{r-1}^{t_{r-1}} p_r^{1-\sum_{i\in[r-1]}t_i} \\
        &=  \int\!\!\d\mu\ (p_1/p_r)^{t_1}\cdots (p_{r-1}/p_{r})^{t_{r-1}} p_r \\
        &=  \int\!\!\d\mu\ p_r \exp\! \left(\sum_{i\in[r-1]} t_i \ln\left( p_i/p_r\right)\right) \\
        &=  \int\!\!\d\mu\ p_r \exp\! \left(\sum_{i\in[r-1]} t_i q_i \right), \label{eq:exponential-integral-form-hellinger}
\end{align}
where $q_i \coloneqq \ln (p_i/p_r)$. 
Consider the following {\it exponential family} of densities with respect to $\mu$ given by 
\begin{equation}
    p_{\mathbf{t}} 
    \coloneqq \dfrac{1}{\operatorname{H}(\mathbf{t})} p_r \exp\! \left(\sum_{j\in[r-1]} t_j q_j \right) , \qquad \mathbf{t}\in \mathbb{T}_r.\label{eq:density-pt-expression}
\end{equation}
Also, define a function $\operatorname{K}:\mathbb{T}_r \to \mathbb{R}$ by 
\begin{align}
    \operatorname{K}(\mathbf{t}) \coloneqq \ln \operatorname{H}(\mathbf{t}).\label{eq:definition-K(t)}
\end{align}
Let $\mathbf{e}_1,\ldots, \mathbf{e}_{r-1}$ denote the standard unit vectors in $\mathbb{R}^{r-1}$, and let $\mathbb{T}_r^{\circ}$ denote the interior of $\mathbb{T}_r$ which is given by
\begin{align}\label{eq:def-T_r-interior}
    \mathbb{T}_r^{\circ} \coloneqq \bigg \{ (t_1,\ldots, t_{r-1}) \in (0,1)^{r-1}: \sum_{i\in [r-1]} t_i < 1 \bigg \}.
\end{align}
By Theorem~2.64 of \cite{schervish2012theory} we know that $\operatorname{H}$ is a smooth function on $\mathbb{T}_r^\circ$; also its partial derivatives are given for $\mathbf{t} \in \mathbb{T}_r^{\circ}$ and $i \in [r-1]$ by
\begin{align}
    \partial_i\!\operatorname{H}(\mathbf{t}) 
    &\coloneqq \lim_{h \to 0} \dfrac{\operatorname{H}(\mathbf{t}+h\mathbf{e}_i)-\operatorname{H}(\mathbf{t})}{h} \\
    &= \int\!\!\d\mu\ q_i p_r \exp\! \left(\sum_{j\in[r-1]} t_j q_j \right) \\
    &= \operatorname{H}(\mathbf{t}) \mathbb{E}_\mathbf{t}\left[q_i\right],\label{eq:partial-hellinger-modified-function}
\end{align}
where $\mathbb{E}_\mathbf{t}$ is the expectation under the density $p_{\mathbf{t}}.$
We know that the Hellinger transform is a continuous function taking only positive values on $\mathbb{S}_r$. This implies that $\operatorname{K}$ is a real-valued continuous function on $\mathbb{T}_r$. Additionally, the smoothness of $\operatorname{H}$ on $\mathbb{T}_r^{\circ}$ implies the smoothness of $\operatorname{K}$  on $\mathbb{T}_r^{\circ}$. From \eqref{eq:partial-hellinger-modified-function} we have that 
\begin{align}\label{neweqn3}
    \partial_i\!\operatorname{K}(\mathbf{t})
    =\dfrac{1}{\operatorname{H}(\mathbf{t})} \partial_i\!\operatorname{H}(\mathbf{t}) =  \mathbb{E}_\mathbf{t}\left[q_i\right], \qquad i \in [r-1], \ \mathbf{t} \in \mathbb{T}_r^\circ.
\end{align}
By Theorem~18.7 of \cite{dasgupta2011probability}, the Hessian matrix of $\operatorname{K}$ at $\mathbf{t} \in \mathbb{T}_r^{\circ}$ is the covariance matrix of the random vector $X \coloneqq (q_1,\ldots, q_{r-1})$ with the probability density $p_{\mathbf{t}}$, whence the Hessian matrix of $\operatorname{K}$ is positive semi-definite at each point in $\mathbb{T}_r^{\circ}$. This implies that $\operatorname{K}$ is a convex function on $\mathbb{T}_r^{\circ}$. By continuity, $\operatorname{K}$ is convex on $\mathbb{T}_r$. 
Let $\mathbb{T}_r^1$ denote the set
\begin{align}\label{eq:corner-point-def}
    \mathbb{T}_{r}^{1}\coloneqq \left\{  (t_1,\ldots, t_{r-1}) \in\mathbb{T}_{r}: \sum_{i\in\left[
r-1\right]  }t_{i}<1\right\}.
\end{align}
We call $\mathbb{T}_r^1$ the set of {\it non-corner} points of $\mathbb{T}_r$. It is easy to see that $\mathbb{T}_r^{\circ} \subset \mathbb{T}_r^1$. 
For any $\mathbf{t} \in \mathbb{T}_r^1$ and $i \in [r-1]$, the limit
\begin{align}
    \partial_i^+\!\operatorname{K}(\mathbf{t}) \coloneqq \lim_{h \searrow 0} \dfrac{\operatorname{K}(\mathbf{t}+h\mathbf{e}_i)-\operatorname{K}(\mathbf{t})}{h}\label{eq:one-sided-dir-der-K}
\end{align}
exists in $\mathbb{R} \cup \{-\infty\}$ (Lemma~\ref{lem-right-deriv}).
Observe that for $\mathbf{t} \in \mathbf{T}_r^\circ$, we have $\partial_i\!\operatorname{K}(\mathbf{t})=\partial_i^+\!\operatorname{K}(\mathbf{t})$ for all $i \in [r-1]$.
It is shown in Appendix~\ref{app:explicit-form-gamma-i-non-corner-points} that for $\mathbf{t} \in \mathbb{T}_r^1$ and $i \in [r-1]$, the expectation value $\mathbb{E}_{\mathbf{t}}[q_i]$ exists in $\mathbb{R}\cup \{-\infty\}$ and satisfies 
\begin{align}
    \partial_i^+\!\operatorname{K}(\mathbf{t})=\mathbb{E}_{\mathbf{t}}[q_i].\label{eq:expectation-qi-one-sided-derivative}
\end{align}
Define a set
\begin{equation}
\mathbb{T}_{r,f}^{1}\coloneqq \left\{  \mathbf{t}\in\mathbb{T}_{r}^{1}:\partial
_{i}^{+}\!\operatorname{K}(\mathbf{t})\neq-\infty, \ \forall i\in\left[  r-1\right]  \right\}
.\label{Trf-definition}%
\end{equation}
Note that $\mathbb{T}_{r}^{\circ}\subset\mathbb{T}_{r,f}^{1}$.

Using the definition \eqref{eq:density-pt-expression} of the density $p_{\mathbf{t}}$ for $\mathbf{t}\coloneqq (t_1,\ldots, t_{r-1}) \in \mathbb{T}_{r,f}^{1}$ and $i \in [r]$, we have
\begin{align}
    \ln \dfrac{p_i}{p_{\mathbf{t}}} 
        &=  \ln p_i - \ln p_{\mathbf{t}} \\
        &=  \ln p_i - \sum_{j\in[r-1]} t_j q_j - \ln p_r + \ln \operatorname{H}(\mathbf{t}) \\
        &=  \ln \dfrac{p_i}{p_r} - \sum_{j\in[r-1]} t_j q_j + \operatorname{K}(\mathbf{t})\\
        &= q_i - \sum_{j\in[r-1]} t_j q_j + \operatorname{K}(\mathbf{t}) \label{neweqn1},
\end{align}
where $q_r$ is the zero function on $\Omega$. 
We write \eqref{neweqn1} in a more compact form as
\begin{align}
    \ln \dfrac{p_i}{p_{\mathbf{t}}} &= \sum_{j\in[r-1]} (\delta_{ij}-t_j) q_j + \operatorname{K}(\mathbf{t}), \qquad \text{for } i\in[r], \quad \mathbf{t}\in \mathbb{T}_{r,f}^{1}. \label{neweqn2}
\end{align}
Here $\delta_{ij}$ is the Kronecker delta (taking the value $1$ if $i=j$, and $0$ otherwise).
By taking the expectation on both sides of \eqref{neweqn2} under the density $p_{\mathbf{t}}$, and then using \eqref{eq:expectation-qi-one-sided-derivative}, we get
\begin{align}
    \gamma_i(\mathbf{t}) \coloneqq \mathbb{E}_{\mathbf{t}} \!\left[\ln \dfrac{p_i}{p_{\mathbf{t}}}\right] 
        &= \sum_{j\in[r-1]} (\delta_{ij}-t_j) \mathbb{E}_{\mathbf{t}}[q_j] + \operatorname{K}(\mathbf{t}) \label{eq:def_gamma_i-expectation} \\
        &= \sum_{j\in[r-1]} (\delta_{ij}-t_j) \partial_i^+\!\operatorname{K}(\mathbf{t}) + \operatorname{K}(\mathbf{t})\label{eq:def_gamma_i}
\end{align}
 for all $i\in [r]$ and $\mathbf{t}\in \mathbb{T}_{r,f}^{1}$.
We can write \eqref{eq:def_gamma_i} in a more compact form as
\begin{align}\label{eq:gamma-i-compact-form}
    \gamma_i (\mathbf{t}) =
    \begin{cases}
        \partial_i^+\!\operatorname{K}(\mathbf{t}) - \mathbf{t}^T \nabla^+\! \operatorname{K}(\mathbf{t}) + \operatorname{K}(\mathbf{t}), & i\in [r-1], \\
        - \mathbf{t}^T \nabla^+\!\operatorname{K}(\mathbf{t}) + \operatorname{K}(\mathbf{t}), & i=r,
    \end{cases} \qquad \text{for }  \mathbf{t} \in \mathbb{T}_{r,f}^{1},
\end{align}
where $\nabla^+\!\operatorname{K}(\mathbf{t}) \coloneqq (\partial_1^+\!\operatorname{K}(\mathbf{t}), \ldots, \partial_{r-1}^+\!\operatorname{K}(\mathbf{t}))^T$.

Let $\omega^{n}\coloneqq (\omega_1,\ldots,\omega_n) \in \Omega^{n}$ and $\mathbf{t} \in \mathbb{T}_{r,f}^{1}$ be arbitrary. We have that
\begin{align}
    p_i^{\otimes n} (\omega^n)&= \left(\prod_{j\in [n]} \dfrac{p_i}{p_{\mathbf{t}}} (\omega_j)\right) p_{\mathbf{t}}^{\otimes n} (\omega^n)= \exp\! \left(n G^{(i)}_{\mathbf{t},n}(\omega^n)\right)p_{\mathbf{t}}^{\otimes n} (\omega^n),\label{eq:tensored-density-expression}
\end{align}
where 
\begin{align}
G^{(i)}_{\mathbf{t},n}(\omega^n) \coloneqq \dfrac{1}{n}  \sum_{j\in[n]} \ln \dfrac{p_i}{p_{\mathbf{t}}} (\omega_j), \qquad \text{for } i\in[r].\label{eq:def-G-i}
\end{align}
Let $P_{\mathbf{t}}^{\otimes n}$ be the product measure corresponding to the density $p_{\mathbf{t}}^{\otimes n}$, and let $\mathbb{E}_{\mathbf{t}}^n$ be the pertaining expectation. By the definition in \eqref{eq:def_gamma_i-expectation}, we then have that
\begin{align}
    \mathbb{E}_{\mathbf{t}}^n \!\left[G^{(i)}_{\mathbf{t},n}\right] &=  \gamma_{i}(\mathbf{t}), \qquad \text{for }i \in [r].
\end{align}
Since $G^{(i)}_{\mathbf{t},n}$ is an i.i.d.~average, the law of large numbers \cite{billingsley1995probability} implies that for arbitrary $\delta >0,$ there exists $n_\delta \in \mathbb{N}$ such that the probability of the event
\begin{align}
    U_{n,\delta} &\coloneqq \{\omega^n \in \Omega^n: \forall i \in [r], \  G^{(i)}_{\mathbf{t},n}(\omega^n) \geq \gamma_{i}(\mathbf{t})-\delta\}
\end{align}
satisfies
\begin{align}
    P_{\mathbf{t}}^{\otimes n}(U_{n,\delta}) &\geq 1-\delta, \qquad \text{ for }n \geq n_\delta.\label{eq:end-tensored-density-expression}
\end{align}
The development in \eqref{eq:tensored-density-expression}--\eqref{eq:end-tensored-density-expression} implies that, for all $n \geq n_\delta$, 
\begin{align}
    \operatorname{Err}_{\operatorname{cl}}(\mathcal{E}_{\operatorname{cl}}^n) 
        &= \int\!\!\d\mu^{\otimes n}\ \left(\eta_1 p_1^{\otimes n}  \wedge \cdots \wedge \eta_r p_{r}^{\otimes n} \right)\\
        &\geq \eta_{\operatorname{min}} \int\!\!\d\mu^{\otimes n}\ \left( p_1^{\otimes n}  \wedge \cdots \wedge  p_{r}^{\otimes n} \right)\\
        &= \eta_{\operatorname{min}}  \int\!\!\d\mu^{\otimes n}\ \left(\exp\! \left(n G_{\mathbf{t},n}^{(1)}\right)  \wedge \cdots \wedge \exp\! \left(n G_{\mathbf{t},n}^{(r)}\right)\right) p_{\mathbf{t}}^{\otimes n} \\
        &= \eta_{\operatorname{min}} \, \mathbb{E}_{\mathbf{t}}^n \!\left[\exp\!\left(n G_{\mathbf{t},n}^{(1)}\right)  \wedge \cdots \wedge \exp\! \left(n G_{\mathbf{t},n}^{(r)}\right)\right] \\
        &\geq \eta_{\operatorname{min}}\,  \mathbb{E}_{\mathbf{t}}^n \!\left[\mathbf{1}_{U_{n,\delta}} \!\left(\exp\!\left(n G_{\mathbf{t},n}^{(1)}\right)  \wedge \cdots \wedge \exp\! \left(n G_{\mathbf{t},n}^{(r)}\right)\right) \right] \\
        &\geq \eta_{\operatorname{min}}\,  P_{\mathbf{t}}^{\otimes n}(U_{n,\delta}) \exp\!\left(n\min_{1\leq i \leq r}(\gamma_{i}(\mathbf{t})-\delta) \right)\\
        &\geq \eta_{\operatorname{min}} (1-\delta)   \exp\!\left( n \min_{1\leq i \leq r}\gamma_{i}(\mathbf{t})-n\delta \right).
\end{align}
Here $\mathbf{1}_{U_{n,\delta}}$ denotes the indicator function of the set $U_{n,\delta}$.
Therefore, we have that
\begin{align}
    -\dfrac{1}{n} \ln \operatorname{Err}_{\operatorname{cl}}(\mathcal{E}_{\operatorname{cl}}^n) \leq -\dfrac{\ln(\eta_{\operatorname{min}}(1-\delta))}{n}-\left( \min_{1\leq i \leq r}\gamma_{i}(\mathbf{t})-\delta \right),\qquad \text{for } n \geq n_{\delta}.\label{eq:optimal-error-exponent-n-shot-upper-bound}
\end{align}
By taking the limit superior as $n \to \infty$ on both sides of \eqref{eq:optimal-error-exponent-n-shot-upper-bound} and then the limit $\delta \to 0,$ we thus get
\begin{align}
    \limsup_{n \to \infty} - \dfrac{1}{n} \ln \operatorname{Err}_{\operatorname{cl}}(\mathcal{E}_{\operatorname{cl}}^n) \leq  -\min_{1\leq i \leq r}\gamma_{i}(\mathbf{t}),\qquad \text{for } \mathbf{t} \in \mathbb{T}_{r,f}^{1}. \label{eq:classical-error-exp-upp-bound-gamma-i}
\end{align}
Recall from \eqref{eq:hellinger-divergence-a(t)} and the fact $\operatorname{K}(\mathbf{t})=\ln \operatorname{H}(\mathbf{t})$, our goal is to show that 
\begin{align}
    \limsup_{n\to \infty}- \dfrac{1}{n} \ln \operatorname{Err}_{\operatorname{cl}}(\mathcal{E}_{\operatorname{cl}}^n) \leq  \sup_{\mathbf{t}\in \mathbb{T}_r} -\operatorname{K}(\mathbf{t}).
\end{align}
In view of \eqref{eq:classical-error-exp-upp-bound-gamma-i}, it suffices to show that for some $\mathbf{t}^* \in \mathbb{T}_{r,f}^{1}$, the following holds
\begin{align}\label{eqn4}
    \min_{1\leq i \leq r} \gamma_i (\mathbf{t}^*) \geq \operatorname{K}(\mathbf{t}^*).
\end{align}
We now argue that such a $\mathbf{t}^*$ exists. 
Since $\operatorname{K}$ is a continuous function on the compact set $\mathbb{T}_r$, there exists $\mathbf{t}^*\coloneqq(t^*_1, \ldots, t^*_{r-1}) \in \mathbb{T}_r$ that minimizes $\operatorname{K}$ over $\mathbb{T}_r$, i.e.,
\begin{align}
    \operatorname{K}(\mathbf{t}^*) = \min_{\mathbf{t} \in \mathbb{T}_r} \operatorname{K}(\mathbf{t}).
\end{align}
Consider the following two cases.

\bigskip
\noindent \textbf{Case A:} Suppose $\mathbf{t}^* \in \mathbb{T}_r^1.$ Choose arbitrary $i \in [r-1]$.
If $t^*_i=0$, then by convexity, continuity of $\operatorname{K}$, and the fact that $\mathbf{t}^*$ is a minimizer, we have $\partial_i^+\!\operatorname{K}(\mathbf{t}^*)  \geq 0.$ (see Lemma~\ref{lem-right-deriv}). Else we have $0< t^*_i < 1$ and the first order necessary condition for a minimizer implies $\partial_i^+\!\operatorname{K}(\mathbf{t}^*)=0.$ Combining these, we get $\partial_i^+\!\operatorname{K}(\mathbf{t}^*) \geq 0$  for all $i\in[r-1]$ and hence $\mathbf{t}^* \in \mathbb{T}^1_{r,f}$, and $\mathbf{t}^{*T} \nabla^+\!\operatorname{K}(\mathbf{t}^*)=0.$ From \eqref{eq:gamma-i-compact-form}, we thus get
\begin{align}
    \gamma_i (\mathbf{t}^*) =
    \begin{cases}
        \partial_i^+\!\operatorname{K}(\mathbf{t}^*)  + \operatorname{K}(\mathbf{t}^*), & i\in [r-1], \\
         \operatorname{K}(\mathbf{t}^*), & i=r.
    \end{cases}
\end{align}
This implies that the inequality \eqref{eqn4} holds for the minimizer $\mathbf{t}^*$.

\bigskip
\noindent \textbf{Case B:} Suppose $\mathbf{t}^* \in \mathbb{T}_r \backslash \mathbb{T}_r^1$, i.e., $t^*_1 + \cdots + t^*_{r-1} = 1.$ 
For some $i \in [r-1]$, we have $t^*_i > 0.$ 
According to the current parameterization of the unit simplex given in \eqref{eq:parametrization-unit-simplex},  $\mathbf{t}^*$ corresponds to the vector
$(t_1^*,\ldots, t_{r-1}^*,0)$ in $\mathbb{S}_r$.
We reparameterize the unit simplex $\mathbb{S}_r$ as
\begin{align}
    \mathbf{s}_{\mathbf{u}}=\left(u_1, \ldots, u_i, 1-\sum_{j\in[r-1]}u_j, u_{i+1},\ldots, u_{r-1} \right),\qquad \mathbf{u} \in \mathbb{T}_r.
\end{align}
In the reparameterized problem, the corresponding minimizer $\mathbf{u}^*$ of $\operatorname{K}$ satisfies $\mathbf{s}_{\mathbf{u}^*}=(t_1^*,\ldots, t_{r-1}^*,0)$, which implies
\begin{align}\label{eq:reparametrization-minimizer}
    1- \sum_{j\in[r-1]} u_j^* = t^*_i > 0.
\end{align}
This reduces the problem to \textbf{Case A}, which implies that \eqref{eqn4} holds.
\bigskip

Combining the above two cases, we conclude that \eqref{eqn4}  holds for the minimizer $\mathbf{t}^*$ and this completes the proof. 
\end{proof}

\begin{boxed}{white}
\begin{theorem}\mbox{}\label{thm:optimal-classical-general}
        Consider an ensemble $\mathcal{E}_{\operatorname{cl}}=\{(\eta_i, P_i): i \in [r]\}$ of probability measures, with possibly unequal supports, on a measurable space $(\Omega, \mathcal{A})$. The optimal error exponent for antidistinguishing the probability measures is given by their multivariate Chernoff divergence, i.e.,
        \begin{align}
            \operatorname{E}_{\operatorname{cl}}(P_1,\ldots, P_r)=\lim_{n \to \infty} -\dfrac{1}{n}\ln \operatorname{Err}_{\operatorname{cl}}(\mathcal{E}_{\operatorname{cl}}^n) &= \xi_{\operatorname{cl}}(P_1,\ldots,P_r), \label{eq:classical-asymptotic-error-rate-general}
        \end{align}
        where, recalling \eqref{eq:unit_simplex}, \eqref{eq:hellinger-trans-def}, and \eqref{eq:classical-CH-divergence}, the multivariate classical Chernoff divergence $\xi_{\operatorname{cl}}$ is defined as
        \begin{equation}
            \xi_{\operatorname{cl}}(P_1,\ldots,P_r) \coloneqq -\ln \inf_{\mathbf{s}\in \mathbb{S}_r} \int\!\!\d\mu\  p_1^{s_1}\cdots p_r^{s_r} .
        \end{equation}
    \end{theorem}
\end{boxed}

\begin{proof}
 See Appendix~\ref{app:general-classical-AD-error-exponent} for a detailed proof.
\end{proof}

\subsection{Multivariate Chernoff divergence versus pairwise Chernoff divergences}

Identifying the true probability measure out of the given $r$ probability measures is the same as eliminating all the remaining $r-1$ false probability measures.
As such, general intuition says that, upon observing i.i.d.~data, it is easier to eliminate a false probability measure than to identify the true probability measure. This also means that the optimal error exponent of classical antidistinguishability should be greater than that of multiple classical hypothesis testing, the former being the   multivariate classical Chernoff divergence and the latter being the minimum of the pairwise Chernoff divergences of the probability measures \cite{salikhov1973asymptotic} (see also \cite[Theorem~4.2]{torgersen1981measures} and   \cite{Leang1997,salikhov1999one,salikhov2003optimal}). Indeed, for any two indices $i, j \in [r]$ define a subset of $\mathbb{S}_r$:
 \begin{align}
     \mathbb{S}_r^{(i,j)} = \{\mathbf{s}\in \mathbb{S}_r: \mathbf{s}\coloneqq(s_1,\ldots, s_r), s_i+s_j=1\}.
 \end{align}
 By definition, we have
 \begin{align}
     \xi_{\operatorname{cl}}(P_1,\ldots,P_r) 
     &\geq -\ln \inf_{\mathbf{s}\in \mathbb{S}_r^{(i,j)}} \mathbb{H}_\mathbf{s}(P_1,\ldots,P_r) \\
     &= -\ln \inf_{s \in [0,1]} \int\!\!\d\mu\  p_i^{s} p_j^{(1-s)} \\
     &= \xi_{\operatorname{cl}}(P_i,P_j),
 \end{align}
 where $\xi_{\operatorname{cl}}(P_i,P_j)$ is the Chernoff divergence of the probability measures $P_i$ and $P_j$.
 This gives
 \begin{align}
     \xi_{\operatorname{cl}}(P_1,\ldots,P_r) \geq \max_{i<j} \xi_{\operatorname{cl}}(P_i,P_j)\geq \min_{i<j} \xi_{\operatorname{cl}}(P_i,P_j). \label{eq:classical-ch-div-lower-bounds}
 \end{align}
The following example illustrates an instance for which the first inequality in \eqref{eq:classical-ch-div-lower-bounds} is strict.
\begin{boxed}{white}
    \begin{example}\label{example:mult-chernoff-strict-ineq-max-pairwise-chernoff}
        Consider a uniform ensemble $\mathcal{E}_{\operatorname{cl}}=\{(1/3, P_1),(1/3, P_2),(1/3, P_3)\}$ of probability measures on a discrete space $\Omega=\{x,y,z\}$ whose densities with respect to the counting measure $\mu$ are given by
        \begin{align}
            p_1 = \dfrac{1}{2}\mathbf{1}_{\{x,y\}}, \qquad
            p_2 = \dfrac{1}{2}\mathbf{1}_{\{x,z\}},\qquad
            p_3 = \dfrac{1}{3}\mathbf{1}_{\Omega}.
        \end{align}
        We have for $\omega^n \in \Omega^n$,
        \begin{align}
           (p_1^{\otimes n} \land p_2^{\otimes n} \land p_3^{\otimes n})(\omega^n) =\begin{cases}
                \dfrac{1}{3^n}, \qquad \text{if } \omega^n=(\underbrace{x,\ldots,x}_{n \text{ times}}), \\
                0, \qquad \text{otherwise}.
            \end{cases}
        \end{align}
        By the minimum likelihood principle, we thus get 
        \begin{align}
            \operatorname{Err}_{\operatorname{cl}}(\mathcal{E}_{\operatorname{cl}}^{n})
            &=\dfrac{1}{3}  \int\!\!\d\mu^{\otimes n}\ \left(p_{1}^{\otimes n}\land   p_{2}^{\otimes n} \land p_3^{\otimes n}\right)\\
            &= \dfrac{1}{3} \cdot \mu^{\otimes n}(\{(\underbrace{x,\ldots,x}_{n \text{ times}})\}) \cdot \dfrac{1}{3^{n}}\\
            &=\dfrac{1}{3} \cdot \mu(\{x\})^n \cdot \dfrac{1}{3^{n}}\\
            &= \dfrac{1}{3} \cdot 1 \cdot \dfrac{1}{3^{n}}\\
            &= \dfrac{1}{3^{n+1}}.
        \end{align}
        This gives the optimal error exponent
        \begin{align}
            \operatorname{E}_{\operatorname{cl}}(P_1,P_2,P_3)=\liminf_{n\to \infty}-\dfrac{1}{n}\ln \operatorname{Err}_{\operatorname{cl}}(\mathcal{E}_{\operatorname{cl}}^{n})=\ln 3. \label{eq:example-optimal-error-exponent}
        \end{align}
        We now compute the pairwise Chernoff divergences of the probability measures as follows.
        \begin{align}
            \xi_{\operatorname{cl}}(P_1, P_2) 
            &= -\ln \inf_{s\in [0,1]} \int\!\!\d\mu\ p_{1}^s   p_{2}^{(1-s)} \\
            &= -\ln \inf_{s\in [0,1]} \int_{\{x\}}\!\!\d\mu\ \dfrac{1}{2^s}   \dfrac{1}{2^{(1-s)}} \\
            &= -\ln \inf_{s\in [0,1]} \dfrac{1}{2^s}   \dfrac{1}{2^{(1-s)}} \\
            &= -\ln \left(\dfrac{1}{2}\right)  \\
            &= \ln 2.
        \end{align}
        Also,
        \begin{align}
            \xi_{\operatorname{cl}}(P_1, P_3) 
            &= -\ln \inf_{s\in [0,1]} \int\!\!\d\mu\ p_{1}^s   p_{3}^{(1-s)} \\
            &= -\ln \inf_{s\in [0,1]} \int_{\{x,y\}}\!\!\d\mu\ \dfrac{1}{2^s}   \dfrac{1}{3^{(1-s)}} \\
            &= -\ln \left[\mu(\{x,y\}) \cdot \dfrac{1}{3} \cdot\inf_{s\in [0,1]} \left(\dfrac{3}{2}\right)^s\right] \\
            &=-\ln \left[2 \cdot \dfrac{1}{3} \cdot 1\right] \\
            &= \ln(3/2).
        \end{align}
        By similar arguments, we get $\xi_{\operatorname{cl}}(P_2, P_3)=\ln (3/2)$.
        This implies
        \begin{align}
            \max\{\xi_{\operatorname{cl}}(P_1, P_2), \xi_{\operatorname{cl}}(P_1, P_3), \xi_{\operatorname{cl}}(P_2, P_3)\} & = \ln 2.\label{eq:example-pairwise-max-chernoff} \\       
            \min\{\xi_{\operatorname{cl}}(P_1, P_2), \xi_{\operatorname{cl}}(P_1, P_3), \xi_{\operatorname{cl}}(P_2, P_3)\} & = \ln (3/2).\label{eq:example-pairwise-min-chernoff} 
        \end{align}
        From \eqref{eq:example-optimal-error-exponent}, \eqref{eq:example-pairwise-max-chernoff}, and \eqref{eq:example-pairwise-min-chernoff}, we have
        \begin{align}
            \operatorname{E}_{\operatorname{cl}}(P_1,P_2,P_3) & > \max\{\xi_{\operatorname{cl}}(P_1, P_2), \xi_{\operatorname{cl}}(P_1, P_3), \xi_{\operatorname{cl}}(P_2, P_3)\} \\
            & > \min\{\xi_{\operatorname{cl}}(P_1, P_2), \xi_{\operatorname{cl}}(P_1, P_3), \xi_{\operatorname{cl}}(P_2, P_3)\}  .
        \end{align}
    \end{example}
\end{boxed}

\section{Achievable error exponent for quantum antidistinguishability}

\label{sec:quantum-achievable-error-exponent}

\subsection{One-shot case}

Observe that the ``antidistinguishability problem'' between any two states $\rho_1$ and $\rho_2$ is the same as the state discrimination problem. Indeed, if we say that ``$\rho_1$ is not the true state,'' then we are saying ``$\rho_2$ is the true state.''
Using this observation, we obtain an upper bound on the optimal error probability of antidistinguishing the states of a given quantum ensemble by considering ``special'' POVMs that focus on pairs of states, as  expounded upon in the proof of the following theorem:   
\begin{boxed}{white}
    \begin{theorem}\mbox{}\label{thm:upper_bound_optimal_error}
         Consider a quantum ensemble $\mathcal{E}=\{(\eta_i, \rho_i): i \in [r]\}$.
        An upper bound on the optimal error probability of antidistinguishing the  states of the ensemble is given by
        \begin{align}
        \label{eq:err-prob-pairwise}
            \operatorname{Err}(\mathcal{E}) & \leq \min_{1 \leq i < j \leq r} \Tr[\eta_i \rho_i \land \eta_j \rho_j]\\
            & = \min_{1 \leq i < j \leq r}  \frac{1}{2}\left(\eta_i + \eta_j - \left \Vert 
            \eta_i \rho_i - \eta_j \rho_j
            \right \Vert_1\right)
            .   
        \label{eq:err-prob-pairwise-2}
        \end{align}
        In particular, if at least two states in $\rho_1,\ldots, \rho_r$ are mutually orthogonal then $\operatorname{Err}(\mathcal{E})=0$. 
    \end{theorem}
\end{boxed}

\begin{proof}
Given two fixed indices $i,j\in [r]$, let $\Xi_r^{(i,j)}$ denote the set of POVMs $\mathscr{M}=\{M_1,\ldots, M_r\}$ such that $M_k=0$ if $k \notin \{i,j\}$.  
For such POVMs, we have $M_i+M_j=\mathbb{I}$ and 
\begin{align}
    \operatorname{Err}(\mathscr{M}; \mathcal{E})
    &=\eta_i \Tr \!\left[M_i \rho_i\right] + \eta_j\Tr \!\left[(\mathbb{I}-M_i) \rho_j\right] \\
    &=\eta_i \Tr \!\left[M_i \rho_i\right] + \eta_j  - \eta_j\Tr \!\left[M_i\rho_j\right] \\
    &= \eta_j  - \Tr \!\left[M_i \left(\eta_j\rho_j-\eta_i \rho_i \right)\right].\label{eqn14}
\end{align}
By taking the infimum over $\Xi_r^{(i,j)}$ on both sides of \eqref{eqn14}, we get
\begin{align}
    \inf_{\mathscr{M} \in \Xi_r^{(i,j)}} \operatorname{Err}(\mathscr{M}; \mathcal{E}) = \eta_j  - \sup_{0 \leq M_i \leq \mathbb{I}} \Tr \!\left[M_i \left(\eta_j\rho_j-\eta_i \rho_i \right)\right],\label{eq:inf_special_povm}
\end{align}
 where the supremum on the right-hand side of \eqref{eq:inf_special_povm} is taken over every positive semi-definite operator $M_i$ such that $0 \leq M_i \leq \mathbb{I}$.
 The supremum is attained by the \textit{Helstrom--Holevo measurement} \cite{helstrom1969quantum, Hol72} given by $M_i=\operatorname{supp} (\eta_j \rho_j-\eta_i\rho_i)_+$ (note the order of $i$ and $j$). We thus get
 \begin{align}
    \inf_{\mathscr{M} \in \Xi_r^{(i,j)}} \operatorname{Err}(\mathscr{M}; \mathcal{E})
    &= \eta_j  - \Tr \!\left[ \left(\eta_j\rho_j-\eta_i \rho_i \right)_+\right] \\
    &=\Tr \!\left[\eta_j\rho_j \right]  - \Tr \!\left[ \left(\eta_j\rho_j-\eta_i \rho_i + |\eta_j\rho_j-\eta_i \rho_i|\right)/2\right] \\ 
    &=\Tr \!\left[\left( \eta_i \rho_i+\eta_j \rho_j - |\eta_i\rho_i-\eta_j \rho_j|\right)/2\right] \\
    & = \Tr \!\left[\eta_i \rho_i \land \eta_j \rho_j \right]
 \label{eq2:inf_special_povm}\\
 & = \frac{1}{2}\left(\eta_i + \eta_j - \left \Vert 
            \eta_i \rho_i - \eta_j \rho_j
            \right \Vert_1\right).
\end{align} 
 It is clear that the optimal antidistinguishability error probability satisfies
 \begin{align}
     \operatorname{Err}(\mathcal{E}) \leq \inf_{\mathscr{M} \in \Xi_r^{(i,j)}} \operatorname{Err}(\mathscr{M}; \mathcal{E})=\Tr \!\left[\eta_i \rho_i \land \eta_j \rho_j \right], \qquad \text{for all } 1 \leq i < j \leq r.\label{eq:optimal_error_special_povm_error}
 \end{align}
By combining \eqref{eq2:inf_special_povm}--\eqref{eq:optimal_error_special_povm_error}, we thus get the upper bound on the optimal antidistinguishability error probability stated in the  theorem. 
\end{proof}

The expression on the right-hand side of \eqref{eq:err-prob-pairwise} can be further simplified for pure states. This is a consequence of the following identity (see Proposition~\ref{prop:pure-state-identity-TD} in Appendix~\ref{app:pure-state-identity-TD}):
\begin{equation}
\left\Vert |\varphi\rangle\!\langle\varphi|-|\zeta\rangle\!\langle\zeta
|\right\Vert _{1}^{2}=\left(  \langle\varphi|\varphi\rangle+\langle\zeta
|\zeta\rangle\right)  ^{2}-4\left\vert \langle\zeta|\varphi\rangle\right\vert
^{2},
\label{eq:pure-state-identity-TD}
\end{equation}
which holds for vectors $|\varphi\rangle$ and $|\zeta\rangle$,
as well as
Theorem~1 of \cite{audenaert2007discriminating} which states that for all positive semi-definite operators $A,B$ and all $0 \leq s \leq 1$, we have
\begin{align}\label{op_min_inequality}
    \operatorname{Tr} [A \wedge B] \leq \operatorname{Tr} A^s B^{1-s}.
\end{align}
\begin{boxed}{white}
    \begin{corollary}\mbox{} \label{ad_pure_states}
        If the quantum states in Theorem~\ref{thm:upper_bound_optimal_error} are pure, i.e., given by $\rho_i=|\psi_i \rangle\!\langle \psi_i|$, then we have
        \begin{align}
             \operatorname{Err}(\mathcal{E}) & \leq \min_{1\leq i<j \leq r} \frac{\eta_i + \eta_j}{2} \left(1 - \sqrt{1- \frac{4 \eta_i \eta_j |\langle \psi_i|\psi_j \rangle|^2}{(\eta_i + \eta_j)^2}}\right) \label{eqn6}\\
             &\leq \dfrac{1}{2} \min_{1\leq i<j \leq r} |\langle \psi_i|\psi_j \rangle|^2.\label{eq:pure-state-inequality-ad-error}
        \end{align}
    \end{corollary}
\end{boxed}
\begin{proof}
Applying \eqref{eq:err-prob-pairwise-2} and \eqref{eq:pure-state-identity-TD}, we find that for all $1 \leq i < j \leq r$,
\begin{align}
    \operatorname{Err}(\mathcal{E}) &\leq \frac{1}{2}\left(\eta_i + \eta_j - \left \Vert 
            \eta_i \vert \psi_{i}\rangle \!\langle \psi_{i}\vert - \eta_j \vert \psi_{j}\rangle \!\langle \psi_{j}\vert
            \right \Vert_1\right)\\
            & = 
            \frac{1}{2}\left(\eta_i + \eta_j - \sqrt{(\eta_i + \eta_j)^2 - 4 \eta_i \eta_j |\langle \psi_i|\psi_j \rangle|^2}\right) \\
            & = \frac{\eta_i + \eta_j}{2} \left(1 - \sqrt{1- \frac{4 \eta_i \eta_j |\langle \psi_i|\psi_j \rangle|^2}{(\eta_i + \eta_j)^2}}\right).
\end{align}
This proves the inequality \eqref{eqn6}. By \eqref{op_min_inequality}, we get that for all $1 \leq i < j \leq r$ and $0 \leq s \leq 1$,
\begin{align}\label{pure_min_inequality}
    \mathrm{Tr}\!\left[  \eta_i\vert \psi_{i}\rangle \!\langle \psi_{i}\vert \wedge\eta_j\vert \psi_{j}\rangle \!\langle \psi_{j}\vert \right]
    \leq \eta_i^s \eta_j^{1-s} 
    \operatorname{Tr}\!\left[\vert \psi_{i}\rangle \!\langle \psi_{i} \vert \psi_{j}\rangle \!\langle \psi_{j}\vert \right]=\eta_i^s \eta_j^{1-s} \vert \langle \psi_{i}|\psi_{j}\rangle \vert ^{2}.
\end{align}
Since \eqref{pure_min_inequality} holds for all $s \in [0,1]$, we get
\begin{align}
    \mathrm{Tr}\!\left[  \eta_i\vert \psi_{i}\rangle \!\langle \psi_{i}\vert \wedge\eta_j\vert \psi_{j}\rangle \!\langle \psi_{j}\vert \right] & \leq (\eta_i \wedge \eta_j) \vert \langle \psi_{i}|\psi_{j}\rangle \vert ^{2}  \leq \dfrac{1}{2} \vert \langle \psi_{i}|\psi_{j}\rangle \vert ^{2}. \label{pure_min_inequality_final}
\end{align}
The desired inequality \eqref{eq:pure-state-inequality-ad-error} thus follows by using the inequality \eqref{pure_min_inequality_final} in \eqref{eq:err-prob-pairwise}.
\end{proof}

The sufficient condition for perfect antidistinguishability given in Theorem~\ref{thm:upper_bound_optimal_error} is not a necessary condition, even in the simple case of commuting states. This is illustrated in the following example.
\begin{boxed}{white}
    \begin{example}
        Consider states $\rho_1$, $\rho_2$, and $\rho_3$ diagonalizable in a common eigenbasis $\{|1 \rangle\!\langle 1|, |2 \rangle\!\langle 2 |, |3 \rangle\!\langle 3 |\}$, given by
        \begin{align}
            \rho_1 &=\dfrac{1}{2} \left(|1 \rangle\!\langle 1|+|2 \rangle\!\langle 2 | \right),\\
            \rho_2 &= \dfrac{1}{2} \left(|1 \rangle\!\langle 1|+|3 \rangle\!\langle 3 | \right),\\
            \rho_3 &= \dfrac{1}{2} \left(|2 \rangle\!\langle 2|+|3 \rangle\!\langle 3 | \right).
        \end{align}
        Consider a POVM $\mathscr{M}=\{M_1,M_2,M_3\}$ given by
        \begin{align}
            M_1 &=|3 \rangle\!\langle 3 |,\\
            M_2 &= |2 \rangle\!\langle 2|,\\
            M_3 &= |1 \rangle\!\langle 1 |.
        \end{align}
        The POVM $\mathscr{M}$ antidistinguishes the states perfectly because $\Tr[M_i \rho_i]=0$ for $i\in [3]$. 
        However, no pair of states are mutually orthogonal to each other.
    \end{example}
\end{boxed}


\subsection{Asymptotic case}

As a consequence of Theorem~\ref{thm:upper_bound_optimal_error}, we arrive at a lower bound on the optimal error exponent, as stated in the following theorem.
\begin{boxed}{white}
    \begin{theorem}\mbox{}\label{newthm:lower_bound_optimal_error}
    Consider a quantum ensemble $\mathcal{E}=\{(\eta_i, \rho_i): i \in [r]\}$.
    A lower bound on the optimal error exponent for antidistinguishing the states of the ensemble is given by the maximum of the pairwise Chernoff divergence of the states; i.e., we have
        \begin{align}\label{eqn18}
            \operatorname{E}(\rho_1,\ldots, \rho_r)  \geq \max_{1\leq i<j\leq r} \xi(\rho_i, \rho_j).
        \end{align}
    \end{theorem}
\end{boxed}
\begin{proof}
By Theorem~\ref{thm:upper_bound_optimal_error}, we have
\begin{align} \label{eqn17}
    \operatorname{Err}(\mathcal{E}^{n})   \leq \min_{1 \leq i < j \leq r} \Tr[\eta_i \rho_i^{\otimes n} \land \eta_j \rho_j^{\otimes n}].
\end{align}
By combining \eqref{eqn17} with \eqref{eq:q-chernoff-thm}, 
we get the desired inequality in \eqref{eqn18}.
\end{proof}
\medskip
\noindent Let us recall from Example~\ref{example:mult-chernoff-strict-ineq-max-pairwise-chernoff} that the inequality in \eqref{eqn18} can be strict in some cases.
\begin{boxed}{white}
    \begin{corollary}\label{cor1}
    If the quantum states in Theorem~\ref{newthm:lower_bound_optimal_error} are pure, given by $\rho_i=|\psi_i \rangle\!\langle \psi_i|$, then we have
        \begin{align}\label{eqn7}
             \operatorname{E}(|\psi_1 \rangle\!\langle \psi_1|,\ldots, |\psi_r\rangle\!\langle \psi_r|)  \geq \max_{1\leq i<j \leq r}  - \ln  |\langle \psi_i|\psi_j \rangle|^2.
        \end{align}
    \end{corollary}
\end{boxed}
\begin{proof}
    It follows directly from \eqref{eqn18} and the representation \eqref{eqn:pure_qcd} of the quantum Chernoff divergence for pure states.
\end{proof}


\section{Bounds on the optimal error exponent for quantum antidistinguishability from multivariate quantum Chernoff divergences}\label{sec:bounds_optimal_error_exponent_xi_max_min}

In this section, we begin by introducing the general concept of multivariate quantum Chernoff divergences, and after that, we employ this concept in order to obtain bounds on the optimal error exponent for quantum antidistinguishability. The reasoning used here is inspired by similar reasoning used for distinguishability problems between two states \cite{Mat13,matsumoto2014maximization,mosonyi2015quantum,HT2016,hiaimosonyi2017,Matsumoto2018}.

\subsection{Multivariate quantum Chernoff divergences}

\begin{boxed}{white}
    \begin{definition}
    \label{def:quantum-CH-div-def}
        Let $r \geq 2$ be an integer. We call a function $\boldsymbol\xi:\mathcal{D}^r \to [0, \infty]$ a multivariate quantum Chernoff divergence if it satisfies the following properties:
\begin{enumerate}
    \item Data processing: for states $\rho_1,\ldots, \rho_r$ and a channel $\mathcal{N}$,
    \begin{align}
         \boldsymbol\xi(\rho_1,\ldots, \rho_r) \geq \boldsymbol\xi(\mathcal{N}(\rho_1),\ldots, \mathcal{N}(\rho_r)) ,
    \end{align}
    \item Reduction to the   multivariate classical Chernoff divergence for commuting states: if the states $\rho_1,\ldots, \rho_r$ commute, then
    \begin{align}
        \boldsymbol\xi(\rho_1,\ldots, \rho_r)= \xi_{\operatorname{cl}}(P_1,\ldots, P_r),
    \end{align}
    where $\xi_{\operatorname{cl}}$ is defined in \eqref{eq:classical-CH-divergence}, $P_1,\ldots, P_r$ are probability measures on $[\dim(\mathcal{H})]$,
    \begin{align}
        P_{\ell}(X) \coloneqq \sum_{i \in X} \lambda_{\ell, i}, \qquad \text{for }X \subseteq [\dim(\mathcal{H})],\label{eq:comm-states-prob-meaures}
    \end{align}
     given by a spectral decomposition of the states in a common eigenbasis
     \begin{align}
         \rho_{\ell} = \sum_{i \in [\dim(\mathcal{H})]} \lambda_{\ell, i}|i \rangle\!\langle i|, \qquad \text{for }  \ell \in [r].
     \end{align}
\end{enumerate}
    \end{definition}
\end{boxed}
As stated above, all  multivariate quantum Chernoff divergences agree on commuting states and are equal to the   multivariate classical Chernoff divergence of the corresponding probability measures induced by the states in a common eigenbasis. If $\rho_1,\ldots, \rho_r$ are commuting states, then we denote their divergence by $\xi_{\operatorname{cl}}(\rho_1,\ldots, \rho_r)$. In this case, it is easy to verify that
\begin{align}
    \xi_{\operatorname{cl}}(\rho_1,\ldots, \rho_r) = -\ln \inf_{\mathbf{s} \in \mathbb{S}_r} \sum_{i \in [\dim(\mathcal{H})]} \left(\prod_{\ell \in [r]} \lambda_{\ell, i}^{s_\ell}\right).
\end{align}

As a first starting point, let us explicitly note that the optimal error exponent in \eqref{eq:qu-asy-error-rate-def} is itself a  multivariate quantum Chernoff divergence.
\begin{boxed}{white}
    \begin{proposition}\label{prop:optimal-error-exponent-ch-divergence}
        The optimal error exponent $\operatorname{E}:\mathcal{D}^r \to [0,\infty]$ defined by \eqref{eq:qu-asy-error-rate-def} is a  multivariate quantum Chernoff divergence.
    \end{proposition}
\end{boxed}
\begin{proof}
    See Appendix~\ref{app:proof-optimal-error-exponent-ch-divergence}.
\end{proof}

Let us note that other  multivariate quantum Chernoff divergences can be constructed from the multivariate log-Euclidean divergence, as discussed in Remark~\ref{rem:log-euc-multivariate} below, as well as by means of the multivariate quantum R\'enyi divergences proposed in \cite{furuya2023monotonic,mosonyi2022geometric}. In what follows, we discuss some other constructions of  multivariate quantum Chernoff divergences.

We say that a multivariate quantum Chernoff divergence $\xi_{\operatorname{min}}$ is minimal if it is a lower bound to any other multivariate quantum Chernoff divergence; i.e.,  for any  multivariate quantum Chernoff divergence $\boldsymbol\xi$, we have
\begin{align}
    \xi_{\operatorname{min}}(\rho_1,\ldots,\rho_r) \leq \boldsymbol\xi(\rho_1,\ldots,\rho_r), \qquad \text{for } (\rho_1,\ldots,\rho_r) \in \mathcal{D}^r. \label{eq:minimal-ch-divergence}
\end{align}
A minimal  multivariate quantum Chernoff divergence is unique by definition, and it can be obtained as an optimization over \textit{quantum-to-classical} or \textit{measurement} channels as presented in Proposition~\ref{prop:minimal-ch-divergence} below.

 Let $\mathcal{K}$ be a complex Hilbert space of dimension $t$ with an orthonormal basis $\{|1 \rangle,\ldots, |t \rangle\}$. Associated with a POVM $\{M_1,\ldots, M_t\}$ acting on the Hilbert space $\mathcal{H}$ is a channel $\mathcal{M}$, called measurement channel, which has the following action on an input state $\rho \in \mathcal{D}(\mathcal{H})$:
\begin{align}
    \mathcal{M}(\rho)=\sum_{\omega \in [t]} \operatorname{Tr}[M_\omega \rho] |\omega \rangle\!\langle \omega|.\label{eq:measurement-channel}
\end{align}
 The action of the measurement channel on any given states $\rho_1,\ldots, \rho_r$ produces commuting states $\mathcal{M}(\rho_1),\ldots, \mathcal{M}(\rho_r)$. This induces probability measures $P^\mathcal{M}_1,\ldots, P^{\mathcal{M}}_r$ on the discrete space $\Omega=[t]$, defined by
\begin{align}
    P^{\mathcal{M}}_i (X) 
    &\coloneqq \sum_{x \in X} \operatorname{Tr}[M_x \rho_i], \qquad \text{for } X \subseteq \Omega.\label{eq:probability-measure-measurement-channel}
\end{align}
It can be easily verified that the optimal error probability of antidistinguishing the commuting states $\mathcal{M}(\rho_1),\ldots, \mathcal{M}(\rho_r)$ is equal to that of antidistinguishing the corresponding probability measures $P^\mathcal{M}_1,\ldots, P^{\mathcal{M}}_r$. See  \eqref{eq:commuting-states-representation}--\eqref{eq:optimal-error-equality-classical-quantum} in Appendix~\ref{app:proof-optimal-error-exponent-ch-divergence}. 
\begin{boxed}{white}
    \begin{proposition} \label{prop:minimal-ch-divergence}
        The minimal  multivariate quantum Chernoff divergence is given by
        \begin{align}
             \xi_{\operatorname{min}}(\rho_1,\ldots,\rho_r) =\sup_{\mathcal{M}} \xi_{\operatorname{cl}}(P^{\mathcal{M}}_1,\ldots, P^{\mathcal{M}}_r),
             \label{eq:min-CH-optim-expr}
        \end{align}
        where the supremum is taken over all measurement channels $\mathcal{M}$ with a $t$-dimensional classical output space for all $t \in \mathbb{N}$ and each probability measure $P^{\mathcal{M}}_i$ is defined in \eqref{eq:probability-measure-measurement-channel}.
    \end{proposition}
\end{boxed}
\begin{proof}
    See Appendix~\ref{app:proof-prop:minimal-ch-divergence}.
\end{proof}

Similar to the definition of minimal multivariate quantum Chernoff divergence, we can define the maximal multivariate quantum Chernoff divergence.
We say that a multivariate quantum Chernoff divergence $\xi_{\operatorname{max}}$ is maximal if it is an upper bound to any other multivariate quantum Chernoff divergence; i.e.,  for any  multivariate quantum Chernoff divergence $\boldsymbol\xi$, we have
\begin{align}
    \xi_{\operatorname{max}}(\rho_1,\ldots,\rho_r) \geq \boldsymbol\xi(\rho_1,\ldots,\rho_r), \qquad \text{for } (\rho_1,\ldots,\rho_r) \in \mathcal{D}^r. \label{eq:maximal-ch-divergence}
\end{align}
A maximal  multivariate quantum Chernoff divergence is unique by definition, and it can be obtained as an optimization over \textit{classical-to-quantum} or \textit{preparation} channels as given in Proposition~\ref{prop:maximal-ch-divergence} below.

We can view any probability measure
 $P$ on the discrete space $\Omega=[t]$ as a quantum state in $\mathcal{K}$ with the fixed eigenbasis $\{|1 \rangle\!\langle 1|,\ldots, |t \rangle\!\langle t|\}$, i.e.,
 \begin{align}
     P \equiv \sum_{\omega \in \Omega} P(\{\omega\}) |\omega\rangle\!\langle \omega |.
     \label{eq:prob-meas-into-q-state}
 \end{align}
A quantum channel $\mathcal{P}:\mathcal{L}(\mathcal{K}) \to \mathcal{L}(\mathcal{H})$ is said to prepare a state $\rho \in \mathcal{D}(\mathcal{H})$ from a probability measure $P$ if it satisfies $\mathcal{P}(P)=\rho$ and is called a preparation channel or classical--to--quantum channel (see \cite[Section~4.6.5]{wilde2017quantum} for a review of classical--to--quantum channels).
\begin{boxed}{white}
    \begin{proposition}\mbox{}\label{prop:maximal-ch-divergence}
        The maximal  multivariate quantum Chernoff divergence is given by
        \begin{align}
             \xi_{\operatorname{max}}(\rho_1,\ldots,\rho_r) =\inf_{\substack{(\mathcal{P}, \{P_i\}_{i\in [r]})  }} \left\{\xi_{\operatorname{cl}}(P_1,\ldots,P_r) :\mathcal{P}(P_i)=\rho_i \quad \text{for all } i \in [r] \right\},
             \label{eq:max-CH-div}
        \end{align}
        where the infimum involves  preparation channels $\mathcal{P}$ with $t$-dimensional classical input system,  for all $t \in \mathbb{N}$,  as well as probability measures $\{P_1,\ldots, P_r\}$ of the form in \eqref{eq:prob-meas-into-q-state}. 
    \end{proposition}
\end{boxed}
\begin{proof}
    See Appendix~\ref{app:proof-prop:maximal-ch-divergence}.
\end{proof}

\subsection{Bounds on the optimal error exponent for quantum antidistinguishability}
\label{sec:quantum-antidist-bnds}

The optimal error exponent for quantum antidistinguishability can be bounded from above and  below by the minimal and the maximal  multivariate quantum Chernoff divergences, respectively, as stated in the following theorem. 
\begin{boxed}{white}
    \begin{theorem}\label{thm:min-max-bounds-optimal-error-exponent}
        Let $\mathcal{E}=\{(\eta_i, \rho_i): i \in [r]\}$ be a quantum ensemble. We have
        \begin{align}
    \xi_{\operatorname{min}}(\rho_1,\ldots, \rho_r) \leq \operatorname{E}(\rho_1,\ldots, \rho_r) \leq \xi_{\operatorname{max}}(\rho_1,\ldots, \rho_r),\label{eq:bounds-optimal-error-exponent}
\end{align}
    where $\xi_{\operatorname{min}}$ and $\xi_{\operatorname{max}}$ are given by \eqref{eq:min-CH-optim-expr} and \eqref{eq:max-CH-div}, respectively.
    Additionally, the bounds in \eqref{eq:bounds-optimal-error-exponent} can be strengthened through regularization as
    \begin{align}
    \sup_{\ell \in \mathbb{N}}\dfrac{1}{\ell}\xi_{\operatorname{min}}(\rho_1^{\otimes \ell},\ldots, \rho_r^{\otimes \ell}) \leq \operatorname{E}(\rho_1,\ldots, \rho_r) \leq \inf_{\ell \in \mathbb{N}} \dfrac{1}{\ell}\xi_{\operatorname{max}}(\rho_1^{\otimes \ell},\ldots, \rho_r^{\otimes \ell}).\label{eq:bounds-regularized-error-exponent}
\end{align}
    \end{theorem}
\end{boxed}

\begin{proof}
    We know from Proposition~\ref{prop:optimal-error-exponent-ch-divergence} that the optimal error exponent is a  multivariate quantum Chernoff divergence, which, along with \eqref{eq:minimal-ch-divergence} and \eqref{eq:maximal-ch-divergence}, justifies the inequalities in \eqref{eq:bounds-optimal-error-exponent}. 

We know from Lemma~\ref{lem:additivity-exp} in  Appendix~\ref{app:additivity-optimal-error-exponent} that
\begin{align}
    \operatorname{E}(\rho_1,\ldots, \rho_r) = \dfrac{1}{\ell}\operatorname{E}(\rho_1^{\otimes \ell},\ldots, \rho_r^{\otimes \ell}) \qquad \text{for all } \ell \in \mathbb{N}.
\end{align}
Substituting the above equality into \eqref{eq:bounds-optimal-error-exponent} gives
\begin{align}
    \dfrac{1}{\ell}\xi_{\operatorname{min}}(\rho_1^{\otimes \ell},\ldots, \rho_r^{\otimes \ell}) \leq \operatorname{E}(\rho_1,\ldots, \rho_r) \leq \dfrac{1}{\ell}\xi_{\operatorname{max}}(\rho_1^{\otimes \ell},\ldots, \rho_r^{\otimes \ell}) \qquad \text{for all } \ell \in \mathbb{N},
\end{align}
which implies the inequalities \eqref{eq:bounds-regularized-error-exponent}.
\end{proof}

\bigskip
We note that in the upper bound in \eqref{eq:bounds-regularized-error-exponent}, the infimum over $\ell \in \mathbb{N}$ can be replaced with the limit $\ell \to \infty$:
\begin{equation}
    \inf_{\ell \in \mathbb{N}} \dfrac{1}{\ell}\xi_{\operatorname{max}}(\rho_1^{\otimes \ell},\ldots, \rho_r^{\otimes \ell}) = \lim_{\ell \to \infty} \dfrac{1}{\ell}\xi_{\operatorname{max}}(\rho_1^{\otimes \ell},\ldots, \rho_r^{\otimes \ell}).\label{eq:limit-maximal-regularized-ch-divergence}
\end{equation}
See Appendix~\ref{app:limits-regularized-ch-divergences}. 
It is open to determine whether the supremum over $\ell \in \mathbb{N}$ in the lower bound in \eqref{eq:bounds-regularized-error-exponent} can be replaced with the limit $\ell \to \infty$, if the limit exists.

It is known from \cite[Corollary~III.8]{mosonyi2015quantum} and  \cite[Corollary~4]{HT2016} (see also \cite[Section~9.3]{matsumoto2014maximization}) that when $r=2$, the following equality holds
\begin{equation}
\sup_{\ell\in \mathbb{N}} \frac{1}{\ell}\xi_{\min}(\rho_{1}^{\otimes\ell
},\rho_{2}^{\otimes\ell})=\widetilde{\xi}(\rho_{1},\rho_{2})\coloneqq \sup
_{s\in\left(  0,1\right)  }\left[-\ln\widetilde{Q}_{s}(\rho_{1},\rho_{2})\right],
\end{equation}
where
\begin{equation}
\widetilde{Q}_{s}(\rho_{1},\rho_{2})\coloneqq \begin{cases}
    \operatorname{Tr}\!\left[  \left(
\rho_{2}^{\left(  1-s\right)  /2s}\rho_{1}\rho_{2}^{\left(  1-s\right)
/2s}\right)  ^{s}\right]  & : s \in [1/2,1)\\
\operatorname{Tr}\!\left[  \left(
\rho_{1}^{  s  /2(1-s)}\rho_{2}\rho_{1}^{  s  /2(1-s)}\right)  ^{1-s}\right] & : s \in (0,1/2)
\end{cases}.
\end{equation}
Since the optimal error exponent is known in this case to be $\xi(\rho
_{1},\rho_{2})$, which is defined in \eqref{eq:quantum_chernoff_divergence}, and it is also known from \cite[Lemma~3]{datta2014limit} that
\begin{equation}
\xi(\rho_{1},\rho_{2})\geq\widetilde{\xi}(\rho_{1},\rho_{2}),
\end{equation}
where the inequality is strict if $\rho_{1}$ and $\rho_{2}$ are invertible and do not commute (see \cite[Theorem~2.1]{hiai1994}), it follows that the lower bound in \eqref{eq:bounds-regularized-error-exponent} cannot be optimal in general.

It is also known from \cite{Mat13,Matsumoto2018}, that when $r=2$, we have
\begin{equation}
\inf_{\ell\in \mathbb{N}} \dfrac{1}{\ell}\xi_{\max}(\rho_{1}^{\otimes\ell},\rho
_{2}^{\otimes\ell}) \geq \widehat{\xi}(\rho_{1},\rho_{2})\coloneqq \sup_{s\in\left(
0,1\right)  }-\ln\widehat{Q}_{s}(\rho_{1},\rho_{2}),
\end{equation}
where%
\begin{equation}
\widehat{Q}_{s}(\rho_{1},\rho_{2})\coloneqq \operatorname{Tr}\!\left[  \rho
_{2}\left(  \rho_{2}^{-1/2}\tilde{\rho}_{1}\rho_{2}^{-1/2}\right)  ^{s}\right].
\end{equation}
Here $\tilde{\rho}_1$ is the absolutely continuous part of $\rho_1$ with respect to $\rho_2$ \cite{ando1976lebesgue}, and the negative power of $\rho_2$ is taken in on its support.
Since the optimal error exponent is known in this case to be $\xi(\rho
_{1},\rho_{2})$ given in \eqref{eq:quantum_chernoff_divergence}, and it is also known from \cite{Mat13,Matsumoto2018}
that
\begin{equation}
\xi(\rho_{1},\rho_{2})\leq\widehat{\xi}(\rho_{1},\rho_{2}),
\end{equation}
where the inequality is strict if $\rho_{1}$ and $\rho_{2}$ are invertible and do not commute (see \cite[Theorem~4.3]{hiaimosonyi2017}), it follows that the upper bound in \eqref{eq:bounds-regularized-error-exponent} cannot be the tightest possible upper bound in general.

\section{Single-letter semi-definite programming upper bound on the optimal error exponent for antidistinguishability}

\label{sec:one-letter-upper-bound-optimal-error-exponent}

    In this section, we derive a single-letter semi-definite programming upper bound on the optimal error exponent.
    Let us begin by recalling that the minimum error probability of
    antidistinguishability of an ensemble $\mathcal{E}\coloneqq \left\{  \left(  \eta_{i},\rho_{i}\right)  :i\in\left[  r\right]  \right\}  $ can also be expressed
in terms of the following primal and dual semi-definite programs \cite[Section~II]{bandyopadhyay2014conclusive} 
(see also \cite[(III.15)]{yuen_kennedy_lax1975}):
\begin{align}
\operatorname{Err}(\mathcal{E}) &  = \inf_{\left\{  M_{i}\right\}_{i\in [r]}%
}\left\{  \sum_{i\in\left[  r\right]  }\eta_{i}\operatorname{Tr}[M_{i}\rho
_{i}]:M_{i}\geq0\quad\text{for all } i\in\left[  r\right]  ,\sum_{i\in\left[
r\right]  }M_{i}=\mathbb{I}\right\}  \label{eq:primal-SDP-antidist}\\
&  =\sup_{Y\in\text{Herm}}\left\{  \operatorname{Tr}[Y]:Y\leq\eta_{i}%
\rho_{i}\quad\text{for all } i\in\left[  r\right]  \right\}
,\label{eq:dual-SDP-antidist}%
\end{align}
where Herm denotes the set of Hermitian operators. The equality holds as a consequence of Slater's condition; indeed we see this by noting that $M_i = \mathbb{I}/r$ is strictly feasible for the primal and $Y = 0$ is feasible for the dual. Defining $\eta_{\min} \coloneq \min_{i\in [r]} \eta_i$,
then it follows that
\begin{align}
    \operatorname{Err}(\mathcal{E}) & =\sup_{Y\in\text{Herm}}\left\{  \operatorname{Tr}
[Y]:Y\leq\eta_{i}\rho_{i}\ \ \forall i\in\left[  r\right]  \right\} \label{eq:err-lower-bound-kappa-1} \\
& \geq \sup_{Y\in\text{Herm}}\left\{  \operatorname{Tr}
[Y]:Y\leq\eta_{\min}\rho_{i}\ \ \forall i\in\left[  r\right]  \right\} \\
& = \sup_{Y\in\text{Herm}}\left\{  \operatorname{Tr}
[\eta_{\min}Y]:\eta_{\min}Y\leq\eta_{\min}\rho_{i}\ \ \forall i\in\left[  r\right]  \right\} \\
& = \eta_{\min} \cdot\sup_{Y\in\text{Herm}}\left\{  \operatorname{Tr}
[Y]:Y\leq\rho_{i}\ \ \forall i\in\left[  r\right]  \right\} \\
& \geq \eta_{\min} \kappa(\rho_1,\ldots, \rho_r)
\label{eq:err-lower-bound-kappa-last} 
\end{align}
where 
\begin{equation} \label{eq:definition-kappa-quantity}
    \kappa(\rho_1,\ldots, \rho_r) \coloneqq \sup_{Y\in\text{Herm}}\left\{  \operatorname{Tr}[Y]:-\rho_{i}\leq Y\leq
\rho_{i}\ \ \forall i\in\left[  r\right]  \right\}.
\end{equation}
The first inequality follows because \begin{equation}
    Y\leq\eta_{\min}\rho_{i}\ \ \forall i\in\left[  r\right] \qquad \Rightarrow \qquad Y\leq\eta_{i}\rho_{i}\ \ \forall i\in\left[  r\right].
\end{equation}
The second equality follows because optimizing over all Hermitian $Y$ is equivalent to optimizing over $\eta_{\min} Y$ since $\eta_{\min} >0$. The third equality follows because $\eta_{\min}Y\leq\eta_{\min}\rho_{i} \Leftrightarrow Y\leq\rho_{i}$ and by factoring $\eta_{\min}$ out of the optimization. The final inequality follows because the optimization in the definition of $\kappa(\rho_1,\ldots, \rho_r)$ adds extra constraints.

The main advantage of the $\kappa$ quantity over the antidistinguishability error probability itself is that it is supermultiplicative, as stated below. For this reason, we can use it to bound the error exponent.

\begin{boxed}{white}
\begin{lemma}\mbox{}
\label{lem:supermult-kappa}
For the tuples of states, $(\rho_1,\ldots, \rho_r)$ and
$(\sigma_1,\ldots, \sigma_r)$, the following supermultiplicativity
inequality holds
\begin{equation}
    \kappa(\rho_1\otimes \sigma_1,\ldots, \rho_r \otimes \sigma_r)
        \geq
            \kappa(\rho_1,\ldots, \rho_r)\cdot\kappa(\sigma_1,\ldots, \sigma_r).\label{eq:super-mult}%
\end{equation}

\end{lemma}
\end{boxed}

\begin{proof}
Let $Y_{\rho}, Y_{\sigma} \in \operatorname{Herm}$ satisfy $-\rho_{i}\leq Y_{\rho}\leq\rho_{i}$ and  $-\sigma_{i}\leq
Y_{\sigma}\leq\sigma_{i}$ for all $i\in\left[  r\right]  $. 
Now invoking
Lemma~12.35 of \cite{KW20}, we conclude that, for all $i\in\left[
r\right]  $,%
\begin{equation}
-\rho_{i}\otimes\sigma_{i}\leq Y_{\rho}\otimes Y_{\sigma}\leq\rho_{i}%
\otimes\sigma_{i}.
\end{equation}
It then follows that%
\begin{align}
\operatorname{Tr}[Y_{\rho}]\cdot\operatorname{Tr}[Y_{\sigma}]  &
=\operatorname{Tr}[Y_{\rho}\otimes Y_{\sigma}]\\
& \leq\sup_{Y\in\text{Herm}}\left\{  \operatorname{Tr}[Y]:-\rho_{i}%
\otimes\sigma_{i}\leq Y\leq\rho_{i}\otimes\sigma_{i}\ \ \forall i\in\left[
r\right]  \right\}  \\
& =\kappa (\rho_1\otimes \sigma_1,\ldots, \rho_r \otimes \sigma_r).
\end{align}
Since the inequality holds for all $Y_{\rho}$ and $Y_{\sigma}$ satisfying the
aforementioned constraints, we conclude~\eqref{eq:super-mult}.
\end{proof}

By applying the supermultiplicativity result inductively, combined with the development in \eqref{eq:err-lower-bound-kappa-1}--\eqref{eq:err-lower-bound-kappa-last}, we conclude the following:

\begin{boxed}{white}
\begin{theorem}\mbox{} \label{thm:improved-one-letter-upper-bound}
For states $\rho_1, \ldots, \rho_r$, the following upper bound holds for the asymptotic error exponent of quantum
antidistinguishability:%
\begin{equation}\label{eq:improved-one-letter-upper-bound}
\operatorname{E}(\rho_1,\ldots, \rho_r) \leq -\ln \kappa(\rho_1,\ldots, \rho_r).
\end{equation}
\end{theorem}
\end{boxed}

\begin{proof}
    Consider that
    \begin{align}
        \operatorname{E}(\rho_1,\ldots, \rho_r) & = \liminf_{n \to \infty} -\frac{1}{n} \ln \operatorname{Err}(\mathcal{E}^n) \\
        & \leq \liminf_{n \to \infty} -\frac{1}{n} \ln \left(\eta_{\min} \kappa(\rho_1^{\otimes n}, \ldots,\rho_r^{\otimes n})\right) \\
        & \leq \liminf_{n \to \infty} -\frac{1}{n} \ln \left( \kappa(\rho_1, \ldots,\rho_r)^n\right) \\
        & = -\ln \kappa(\rho_1, \ldots,\rho_r).
    \end{align}
    The first inequality follows from \eqref{eq:err-lower-bound-kappa-1}--\eqref{eq:err-lower-bound-kappa-last}. The second inequality follows from $\liminf -\frac{1}{n} \ln \eta_{\min}=0$ and Lemma~\ref{lem:supermult-kappa} applied inductively.
\end{proof}

The upper bound in \eqref{eq:improved-one-letter-upper-bound} can be bounded from above by a quantity expressed in terms of the extended max-relative entropy, defined for a Hermitian operator $X \neq 0$ and a positive semi-definite operator $\sigma \neq 0$ as \cite[Eqs.~(14)--(16)]{ww2020}:
\begin{equation}
D_{\max}(X\Vert \sigma)   \coloneqq \ln\inf_{\lambda\geq0}\left\{  \lambda:-\lambda \sigma \leq X \leq\lambda \sigma\right\}.
\end{equation}
We have that $D_{\max}(X\Vert \sigma)=+\infty$ if the support of $X$ is not contained in the support of $\sigma$. 
Also, whenever the support of $X$ is contained in
the support of $\sigma$, we have $D_{\max}(X\Vert \sigma) < +\infty$ and in this case,
\begin{equation}
   D_{\max}(X\Vert \sigma) =\ln\left\Vert \sigma^{-\frac{1}{2}}X\sigma^{-\frac{1}{2}}\right\Vert _{\infty},
\end{equation}
where the inverse is understood to be taken on the support of $\sigma$. In Appendix~\ref{app:properties_extended_dmax}, we derive several fundamental properties of the extended max-relative entropy, including data processing, joint quasi-convexity, lower semi-continuity, non-negativity, faithfulness, monotonicity, and additivity.

\begin{boxed}{white}
\begin{theorem}\mbox{}\label{thm:extended-dmax-err-exp-bnd}
For  quantum states $\rho_1,\ldots, \rho_r$, the quantity $\kappa(\rho_1,\ldots, \rho_r)$ is bounded from below in terms of the extended max-relative entropy, as follows:
\begin{equation}\label{eq:dual-err-double-prime-dmax}
\kappa(\rho_1,\ldots, \rho_r) \geq \exp\left(  -\inf_{\omega\in\mathcal{D}^{\prime}%
    }\max_{i\in\left[  r\right]  }D_{\max}(\omega\Vert\rho_{i})\right),
\end{equation}
where $\mathcal{D}^{\prime} \coloneqq \left\{ \omega : \omega = \omega^\dag, \ \operatorname{Tr}[\omega]=1\right\}$ is the set of all Hermitian operators with trace one. 
Consequently, we have
\begin{align}\label{eq:optimal-error-dmax-upper-bound}
    \operatorname{E}(\rho_1,\ldots, \rho_r) 
        &\leq \inf_{\omega\in\mathcal{D}^{\prime}%
    }\max_{i\in\left[  r\right]  }D_{\max}(\omega\Vert\rho_{i}) \\
    &= \max_{\left\{  s_{i}\right\}  _{i\in\left[  r\right]  }}\inf_{\omega
\in\mathcal{D}^{\prime}}\sum_{i\in\left[  r\right]  }s_{i}D_{\max}(\omega\Vert\rho_{i}),
\end{align}
where $\left\{  s_{i}\right\}  _{i\in\left[  r\right]  }$ is a probability
distribution.
\end{theorem}
\end{boxed}

\begin{proof}
    By the definition \eqref{eq:definition-kappa-quantity} and the fact that $Y=0$ is always feasible for $\kappa(\rho_1,\ldots, \rho_r) $, we conclude that
\begin{align}
\kappa(\rho_1,\ldots, \rho_r) 
& = \sup_{Y\in\text{Herm}}\left\{  \operatorname{Tr}[Y] \geq 0:-\rho_{i}\leq Y\leq
\rho_{i}\ \ \forall i\in\left[  r\right]  \right\} \label{eq:ADerror-lowerbound-trace-positive} \\
& = \sup_{\substack{Y\in\text{Herm} :  \operatorname{Tr}[Y]\geq 0}}\left\{  \operatorname{Tr}[Y]:-\rho_{i}\leq Y\leq
\rho_{i}\ \ \forall i\in\left[  r\right]  \right\} \label{eq:ADerror-lowerbound-trace-positive-1} \\
&  =\sup_{\lambda\geq0,\omega\in\mathcal{D}^\prime}\left\{  \operatorname{Tr}%
[\lambda\omega]:-\rho_{i}\leq \lambda\omega \leq
\rho_{i},\quad\forall  i\in\left[
r\right]  \right\} \label{eq:ADerror-lowerbound-1}  \\
&  =\sup_{\lambda\geq0,\omega\in\mathcal{D}^\prime}\left\{  \lambda:-\rho_{i}\leq \lambda\omega \leq
\rho_{i},\quad\forall  i\in\left[  r\right]  \right\}  \\
&  \geq \sup_{\lambda>0,\omega\in\mathcal{D}^\prime}\left\{  \lambda:-\rho_{i}\leq \lambda\omega \leq
\rho_{i},\quad\forall  i\in\left[  r\right]  \right\}  \\
&  =\sup_{\lambda>0,\omega\in\mathcal{D}^\prime}\left\{  \lambda:-\frac{1}{\lambda}\rho_{i}\leq \omega \leq
\frac{1}{\lambda}\rho_{i},\quad\forall  i\in\left[  r\right]  \right\}  \\
&  =\sup_{\lambda'>0,\omega\in\mathcal{D}^\prime}\left\{  \frac{1}{\lambda'}:-\lambda^\prime \rho_{i}\leq \omega \leq \lambda^\prime\rho_{i},\quad\forall  i\in\left[  r\right]  \right\} \label{eq:ADerror-lowerbound-2}  \\
&  =\left[  \inf_{\lambda'>0,\omega\in\mathcal{D}^\prime}\left\{  \lambda':-\lambda^\prime \rho_{i}\leq \omega \leq \lambda^\prime\rho_{i},\quad\forall  i\in\left[  r\right]  \right\}  \right]  ^{-1}\\
&  =\left[  \inf_{\omega\in\mathcal{D}^\prime}\exp\left(  \max_{i\in\left[  r\right]
}D_{\max}\left(  \omega\Vert\rho_{i}\right)  \right)  \right]  ^{-1}\\
&  =\left[  \exp\left(  \inf_{\omega\in\mathcal{D}^\prime}\max_{i\in\left[  r\right]
}D_{\max}(\omega\Vert\rho_{i})\right)  \right]  ^{-1}\\
&  =\exp\left(  -\inf_{\omega\in\mathcal{D}^\prime}\max_{i\in\left[  r\right]
}D_{\max}(\omega\Vert\rho_{i})\right)  .
\end{align}
The equality \eqref{eq:ADerror-lowerbound-trace-positive} follows because $Y=0$ is feasible in \eqref{eq:definition-kappa-quantity}.
The equality \eqref{eq:ADerror-lowerbound-1} follows because for any Hermitian operator $Y$ with positive trace, we can choose $\lambda = \operatorname{Tr}[Y]$ and $\omega = Y/\operatorname{Tr}[Y] \in \mathcal{D}^\prime$ so that $Y=\lambda \omega$;  and if $Y=0$ then we can choose $\lambda = 0$ and $\omega=\mathbb{I}/\dim(\mathcal{H}) \in \mathcal{D}^\prime$ so that $Y=\lambda \omega$. 
The equality \eqref{eq:ADerror-lowerbound-2} follows from the substitution
$\lambda=\frac{1}{\lambda'}$. 

The desired inequality \eqref{eq:optimal-error-dmax-upper-bound} is a direct consequence of \eqref{eq:improved-one-letter-upper-bound} and \eqref{eq:dual-err-double-prime-dmax}.
Also, we have
\begin{align}
\inf_{\omega\in\mathcal{D}^{\prime}}\max_{i\in\left[  r\right]  }D_{\max}(\omega
\Vert\rho_{i})  & =\inf_{\omega\in\mathcal{D}^{\prime}}\max_{\left\{
s_{i}\right\}  _{i\in\left[  r\right]  }}\sum_{i\in\left[  r\right]  }%
s_{i}D_{\max}(\omega\Vert\rho_{i})\\
& =\max_{\left\{  s_{i}\right\}  _{i\in\left[  r\right]  }}\inf_{\omega
\in\mathcal{D}^{\prime}}\sum_{i\in\left[  r\right]  }s_{i}D_{\max}(\omega\Vert\rho_{i}).
\end{align}
 The first equality follows because the maximum over a finite set
can be replaced with a maximum of the expected value of the elements of the
set, with the maximum taken over all possible distributions. The second
equality follows from an application of the Sion's minimax theorem \cite{sion_1958}:
indeed, the objective function $\sum_{i\in\left[  r\right]  }s_{i}D_{\max
}(\omega\Vert\rho_{i})$ is linear and continuous in the probability
distribution $\left\{  s_{i}\right\}  _{i\in\left[  r\right]  }$, and it is
lower semi-continuous and quasi-convex in $\omega\in\mathcal{D}^{\prime}$.
\end{proof}

\begin{remark}
\label{rem:log-euc-multivariate}

By replacing the set $\mathcal{D}^{\prime}$ with $\mathcal{D}$ (the set of density operators) in Theorem~\ref{thm:extended-dmax-err-exp-bnd}, we get an interesting (although weaker) upper bound on the optimal error exponent:
\begin{align}\label{eq:opt-error-weaker-upp-bound-dmax}
    \operatorname{E}(\rho_1,\ldots, \rho_r) \leq \max_{\left\{  s_{i}\right\}  _{i\in\left[  r\right]  }}\inf_{\omega
\in\mathcal{D}}\sum_{i\in\left[  r\right]  }s_{i}D_{\max}(\omega\Vert\rho_{i}).
\end{align}
This upper bound has a resemblance to
the following divergence:%
\begin{equation}
\max_{\left\{  s_{i}\right\}  _{i\in\left[  r\right]  }}\inf_{\omega
\in\mathcal{D}}\sum_{i\in\left[  r\right]  }s_{i}D(\omega\Vert\rho_{i}%
)=\max_{\left\{  s_{i}\right\}  _{i\in\left[  r\right]  }}\left(
-\ln\operatorname{Tr}\!\left[  \exp\left(  \sum_{i\in\left[  r\right]  }s_{i}%
\ln\rho_{i}\right)  \right]  \right)  ,\label{eq:d-convex-comb-log-euclid}%
\end{equation}
where the equality follows whenever each $\rho_{i}$ is positive definite.
Indeed, the only difference between
\eqref{eq:opt-error-weaker-upp-bound-dmax} and
\eqref{eq:d-convex-comb-log-euclid} is the substitution $D_{\max}(\rho\Vert \sigma)\rightarrow
D(\rho\Vert \sigma) \coloneqq \operatorname{Tr}[\rho (\ln \rho - \ln \sigma)]$, where the latter denotes the standard quantum relative entropy \cite{umegaki_1962}. The
equality in \eqref{eq:d-convex-comb-log-euclid} was established in Eq.~(V.121) and Example~V.25 of \cite{mosonyi2022geometric}. See Appendix~\ref{app:max-rel-entropy-bounds-similarity} for a review of the proof of \eqref{eq:d-convex-comb-log-euclid}. Finally, note that \eqref{eq:d-convex-comb-log-euclid} reduces to the multivariate classical Chernoff divergence when the states in the set $\{\rho_i\}_{i \in [r]}$ commute (have a common eigenbasis). As such, this quantity is a  multivariate quantum Chernoff divergence according to Definition~\ref{def:quantum-CH-div-def}.
\end{remark}

We end the section by deriving another alternative form for the $\kappa$ quantity.

\begin{boxed}{white}
\begin{proposition}\mbox{}
The quantity $\kappa(\rho_1,\ldots, \rho_r)$ can alternatively be written as%
\begin{equation}\label{eq:dual-err-double-prime}
\kappa(\rho_1,\ldots, \rho_r) =\inf_{\substack{Z_{1,i},Z_{2,i}%
\geq0\\\forall i\in\left[  r\right]  }}\left\{  \sum_{i\in [r]}%
\operatorname{Tr}[\left(  Z_{1,i}+Z_{2,i}\right)  \rho_{i}]:\mathbb{I}=\sum_{i\in [r]}%
Z_{2,i}-Z_{1,i}\right\}.
\end{equation}

\end{proposition}
\end{boxed}
\begin{proof}
We prove this by showing that the expression in the right-hand side of \eqref{eq:dual-err-double-prime} is the dual SDP\ of $\kappa(\rho_1,\ldots, \rho_r)$ and that the strong duality holds. 
We derive it as follows:%
\begin{align}
& \sup_{Y\in\text{Herm}}\left\{  \operatorname{Tr}[Y]:-\rho_{i}\leq Y\leq
\rho_{i}\ \ \forall i\in\left[  r\right]  \right\}  \nonumber\\
& =\sup_{Y\in\text{Herm}}\left\{  \operatorname{Tr}[Y]+\inf_{Z_{1,i}%
,Z_{2,i}\geq0}\left\{  \sum_{i \in [r]}\left(  \operatorname{Tr}[Z_{1,i}\left(
Y+\rho_{i}\right)  ]+\operatorname{Tr}[Z_{2,i}\left(  \rho_{i}-Y\right)
]\right)  \right\}  \right\}  \\
& =\sup_{Y\in\text{Herm}}\inf_{Z_{1,i},Z_{2,i}\geq0}\left\{  \operatorname{Tr}%
[Y]+\sum_{i \in [r]}\left(  \operatorname{Tr}[Z_{1,i}\left(  Y+\rho_{i}\right)
]+\operatorname{Tr}[Z_{2,i}\left(  \rho_{i}-Y\right)  ]\right)  \right\}  \\
& =\sup_{Y\in\text{Herm}}\inf_{Z_{1,i},Z_{2,i}\geq0}\left\{  \operatorname{Tr}%
\left[  Y\left(  \mathbb{I}+\sum_{i \in [r]}\left(  Z_{1,i}-Z_{2,i}\right)  \right)
\right]  +\sum_{i \in [r]}\operatorname{Tr}[\left(  Z_{1,i}+Z_{2,i}\right)
\rho_{i}]\right\}  \\
& \leq\inf_{Z_{1,i},Z_{2,i}\geq0}\sup_{Y\in\text{Herm}}\left\{
\operatorname{Tr}\left[  Y\left(  \mathbb{I}+\sum_{i \in [r]}\left(  Z_{1,i}%
-Z_{2,i}\right)  \right)  \right]  +\sum_{i \in [r]}\operatorname{Tr}[\left(
Z_{1,i}+Z_{2,i}\right)  \rho_{i}]\right\}  \\
& =\inf_{Z_{1,i},Z_{2,i}\geq0}\left\{  \sum_{i \in [r]}\operatorname{Tr}[\left(
Z_{1,i}+Z_{2,i}\right)  \rho_{i}]:\mathbb{I}=\sum_{i \in [r]}Z_{2,i}-Z_{1,i}\right\}  .
\end{align}
Strong duality holds here by picking $Z_{2,i}=2\mathbb{I}/r$ and $Z_{1,i}=\mathbb{I}/r$ for all
$i\in\left[  r\right]  $ in the dual and by picking $Y=0$ for the primal.
\end{proof}


\section{Conclusion}

\textbf{Summary}. We have solved the classical antidistinguishability problem of finding the optimal error exponent, which we proved to be equal to the   multivariate classical Chernoff divergence of the given probability measures. To the best of our knowledge, this result constitutes the first operational interpretation of the divergence involving three or more states.
We have also given various upper and lower bounds on the optimal error exponent in the quantum case, while it still remains an open problem to compute its exact expression. 
In analogy with the classical case, we believe that the quantity that gives the exact error exponent in the quantum case should be called {\it the}  multivariate quantum Chernoff divergence.

\textbf{Future directions}.
Recall from \cite{bandyopadhyay2014conclusive} that quantum $m$-state exclusion can be thought of as antidistinguishability of a set of  states related to the original set. We leave it as an intriguing open question to determine the optimal asymptotic error exponent for quantum $m$-state exclusion.

Analogous to the task of antidistinguishing quantum states, one may consider the problem of antidistinguishing an ensemble of quantum channels. In this problem, a quantum channel is chosen randomly from a finite set of quantum channels, with known {\it a priori} probability distribution. The antidistinguisher is allowed to pass one share of a bipartite quantum state through the channel, after which both the reference system and the channel output system are measured. Based on the measurement outcome, the antidistinguisher's goal is to rule out a quantum channel other than the selected one. It would be an interesting future work to study the asymptotics of the error rates for antidistinguishing an ensemble of quantum channels.


\section*{Acknowledgments}

We are especially grateful to Mil{\'a}n Mosonyi for several clarifying discussions about quantum hypothesis testing, as well as to Kaiyuan Ji, Felix Leditzky, Vishal Singh, and Aaron Wagner for insightful discussions. We also thank the anonymous referee and the editor for many helpful comments that improved the manuscript.
HKM and MMW acknowledge support from the National Science Foundation under Grant No.~2304816.

\section*{Declaration of interests}
On behalf of all authors, the corresponding author states that there is no conflict of interest.

\section*{Data availability statement}
Data sharing not applicable to this article as no datasets were generated or analysed during the current study.

\bibliographystyle{alpha}
\bibliography{Ref}

\newcommand{\etalchar}[1]{$^{#1}$}
\begin{thebibliography}{ACMT{\etalchar{+}}07}

\bibitem[ACMT{\etalchar{+}}07]{audenaert2007discriminating}
Koenraad M.~R. Audenaert, John Calsamiglia, Ram{\'o}n Munoz-Tapia, Emilio
  Bagan, Ll. Masanes, Antonio Acin, and Frank Verstraete.
\newblock Discriminating states: The quantum {C}hernoff bound.
\newblock {\em Physical Review Letters}, 98(16):160501, April 2007.
\newblock
  \href{https://arxiv.org/abs/quant-ph/0610027}{arXiv:quant-ph/0610027}.

\bibitem[And76]{ando1976lebesgue}
T.~Ando.
\newblock Lebesgue-type decomposition of positive operators.
\newblock {\em Acta Sci. Math.(Szeged)}, 38(3-4):253--260, January 1976.

\bibitem[BC09]{Barnett09}
Stephen~M. Barnett and Sarah Croke.
\newblock Quantum state discrimination.
\newblock {\em Advances in Optics and Photonics}, 1(2):238--278, April 2009.
\newblock \href{https://arxiv.org/abs/0810.1970}{arXiv:0810.1970}.

\bibitem[BCDvD06]{bacon2006}
Dave Bacon, Andrew~M. Childs, and Wim Dam~van Dam.
\newblock Optimal measurements for the dihedral hidden subgroup problem.
\newblock {\em Chicago Journal of Theoretical Computer Science}, 2006:2,
  October 2006.
\newblock
  \href{https://arxiv.org/abs/quant-ph/0501044}{arXiv:quant-ph/0501044}.

\bibitem[BCLM14]{barrett2014no}
Jonathan Barrett, Eric~G. Cavalcanti, Raymond Lal, and Owen J.~E. Maroney.
\newblock No $\psi$-epistemic model can fully explain the indistinguishability
  of quantum states.
\newblock {\em Physical Review Letters}, 112(25):250403, June 2014.
\newblock \href{https://arxiv.org/abs/1310.8302}{arXiv:1310.8302}.

\bibitem[Bil95]{billingsley1995probability}
Patrick Billingsley.
\newblock {\em Probability and Measure}.
\newblock Wiley Series in Probability and Statistics. Wiley, 1995.

\bibitem[BJOP14]{bandyopadhyay2014conclusive}
Somshubhro Bandyopadhyay, Rahul Jain, Jonathan Oppenheim, and Christopher
  Perry.
\newblock Conclusive exclusion of quantum states.
\newblock {\em Physical Review A}, 89(2):022336, February 2014.
\newblock \href{https://arxiv.org/abs/1306.4683}{arXiv:1306.4683}.

\bibitem[BK15]{bae2015quantum}
Joonwoo Bae and Leong-Chuan Kwek.
\newblock Quantum state discrimination and its applications.
\newblock {\em Journal of Physics A: Mathematical and Theoretical},
  48(8):083001, January 2015.
\newblock \href{https://arxiv.org/abs/1707.02571}{arXiv:1707.02571}.

\bibitem[BL06]{Borwein2006}
J.~Borwein and A.~Lewis.
\newblock {\em Convex Analysis}.
\newblock Springer, New York, 2006.

\bibitem[Bor26]{born1926quantum}
Max Born.
\newblock Quantum mechanics of collision processes.
\newblock {\em Uspekhi Fizich}, June 1926.

\bibitem[CDD{\etalchar{+}}14]{Collins2014}
Robert~J. Collins, Ross~J. Donaldson, Vedran Dunjko, Petros Wallden, Patrick~J.
  Clarke, Erika Andersson, John Jeffers, and Gerald~S. Buller.
\newblock Realization of quantum digital signatures without the requirement of
  quantum memory.
\newblock {\em Phys. Rev. Lett.}, 113:040502, July 2014.
\newblock \href{https://arxiv.org/abs/1311.5760}{arXiv:1311.5760}.

\bibitem[CFS02]{caves2002conditions}
Carlton~M. Caves, Christopher~A. Fuchs, and R{\"u}diger Schack.
\newblock Conditions for compatibility of quantum-state assignments.
\newblock {\em Physical Review A}, 66(6):062111, December 2002.
\newblock
  \href{https://arxiv.org/abs/quant-ph/0206110}{arXiv:quant-ph/0206110}.

\bibitem[Che52]{chernoff1952}
Herman Chernoff.
\newblock A measure of asymptotic efficiency for tests of a hypothesis based on
  the sum of observations.
\newblock {\em The Annals of Mathematical Statistics}, 23(4):493--507, December
  1952.

\bibitem[Das11]{dasgupta2011probability}
Anirban DasGupta.
\newblock {\em Probability for Statistics and Machine Learning: Fundamentals
  and Advanced Topics}.
\newblock Springer, 2011.

\bibitem[DL14]{datta2014limit}
Nilanjana Datta and Felix Leditzky.
\newblock A limit of the quantum {R}{\'e}nyi divergence.
\newblock {\em Journal of Physics A: Mathematical and Theoretical},
  47(4):045304, 2014.

\bibitem[Fek23]{fekete1923distribution}
Michael Fekete.
\newblock \"uber die verteilung der wurzeln bei gewissen algebraischen
  gleichungen mit ganzzahligen koeffizienten.
\newblock {\em Mathematische Zeitschrift}, 17(1):228--249, December 1923.

\bibitem[FL96]{fazekas1996some}
I.~Fazekas and F.~Liese.
\newblock Some properties of the {H}ellinger transform and its application in
  classification problems.
\newblock {\em Computers \& Mathematics with Applications}, 31(8):107--116,
  April 1996.

\bibitem[FLO23]{furuya2023monotonic}
Keiichiro Furuya, Nima Lashkari, and Shoy Ouseph.
\newblock Monotonic multi-state quantum f-divergences.
\newblock {\em Journal of Mathematical Physics}, 64(4):042203, April 2023.

\bibitem[Gri93]{Grigelionis1993}
B.~Grigelionis.
\newblock {\em On Hellinger transforms for solutions of martingale problems},
  pages 107--116.
\newblock Springer New York, 1993.

\bibitem[HB20]{havlivcek2020simple}
Vojt{\v{e}}ch Havl{\'\i}{\v{c}}ek and Jonathan Barrett.
\newblock Simple communication complexity separation from quantum state
  antidistinguishability.
\newblock {\em Physical Review Research}, 2(1):013326, March 2020.
\newblock \href{https://arxiv.org/abs/1911.01927}{arXiv:1911.01927}.

\bibitem[Hel09]{hellinger1909neue}
Ernst Hellinger.
\newblock Neue begr{\"u}ndung der theorie quadratischer formen von
  unendlichvielen ver{\"a}nderlichen.
\newblock {\em Journal f{\"u}r die reine und angewandte Mathematik},
  1909(136):210--271, July 1909.

\bibitem[Hel69]{helstrom1969quantum}
Carl~W. Helstrom.
\newblock Quantum detection and estimation theory.
\newblock {\em Journal of Statistical Physics}, 1:231--252, June 1969.

\bibitem[Hia94]{hiai1994}
Fumio Hiai.
\newblock Equality cases in matrix norm inequalities of {G}olden-{T}hompson
  type.
\newblock {\em Linear and Multilinear Algebra}, 36(4):239--249, April 1994.

\bibitem[HK18]{heinosaari2018antidistinguishability}
Teiko Heinosaari and Oskari Kerppo.
\newblock Antidistinguishability of pure quantum states.
\newblock {\em Journal of Physics A: Mathematical and Theoretical},
  51(36):365303, July 2018.
\newblock \href{https://arxiv.org/abs/1804.10457}{arXiv:1804.10457}.

\bibitem[HM17]{hiaimosonyi2017}
Fumio Hiai and Mil\'{a}n Mosonyi.
\newblock Different quantum f-divergences and the reversibility of quantum
  operations.
\newblock {\em Reviews in Mathematical Physics}, 29(07):1750023, August 2017.

\bibitem[Hol72]{Hol72}
Alexander~S. Holevo.
\newblock An analogue of statistical decision theory and noncommutative
  probability theory.
\newblock {\em Trudy Moskovskogo Matematicheskogo Obshchestva}, 26:133--149,
  1972.

\bibitem[HSR03]{horodecki2003entanglement}
Michael Horodecki, Peter~W. Shor, and Mary~Beth Ruskai.
\newblock Entanglement breaking channels.
\newblock {\em Reviews in Mathematical Physics}, 15(06):629--641, 2003.

\bibitem[HT16]{HT2016}
Masahito Hayashi and Marco Tomamichel.
\newblock Correlation detection and an operational interpretation of the
  {R}\'enyi mutual information.
\newblock {\em Journal of Mathematical Physics}, 57(10):102201, October 2016.

\bibitem[Jac89]{JACOD19893}
Jean Jacod.
\newblock Filtered statistical models and {H}ellinger processes.
\newblock {\em Stochastic Processes and their Applications}, 32(1):3--45, June
  1989.

\bibitem[KW20]{KW20}
Sumeet Khatri and Mark~M. Wilde.
\newblock {\em Principles of Quantum Communication Theory: A Modern Approach}.
\newblock {arXiv}, 2020.
\newblock arXiv:2011.04672v1.

\bibitem[LD20]{leifer2020noncontextuality}
Matthew Leifer and Cristhiano Duarte.
\newblock Noncontextuality inequalities from antidistinguishability.
\newblock {\em Physical Review A}, 101(6):062113, June 2020.
\newblock \href{https://arxiv.org/abs/2001.11485}{arXiv:2001.11485}.

\bibitem[LeC70]{lecam1970}
L.~LeCam.
\newblock On the assumptions used to prove asymptotic normality of maximum
  likelihood estimates.
\newblock {\em The Annals of Mathematical Statistics}, 41(3):802--828, June
  1970.

\bibitem[Lei14]{leifer2014quantum}
Matthew~Saul Leifer.
\newblock Is the quantum state real? {A}n extended review of $\psi$-ontology
  theorems.
\newblock {\em Quanta}, 3:67--155, November 2014.
\newblock \href{https://arxiv.org/abs/1409.1570}{arXiv:1409.1570}.

\bibitem[Li16]{keli}
Ke~Li.
\newblock Discriminating quantum states: The multiple {C}hernoff distance.
\newblock {\em The Annals of Statistics}, 44(4):1661--1679, August 2016.
\newblock \href{https://arxiv.org/abs/1508.06624}{arXiv:1508.06624}.

\bibitem[LJ97]{Leang1997}
C.~C. Leang and D.~H. Johnson.
\newblock On the asymptotics of $m$-hypothesis {B}ayesian detection.
\newblock {\em IEEE Transactions on Information Theory}, 43(1):280--282,
  January 1997.

\bibitem[LM10]{liese2010statistical}
F.~Liese and Klaus-J. Miescke.
\newblock {\em Statistical Decision Theory}.
\newblock Springer New York, NY, 2010.

\bibitem[LR73]{LR73}
Elliott~H. Lieb and Mary~Beth Ruskai.
\newblock Proof of the strong subadditivity of quantum-mechanical entropy.
\newblock {\em Journal of Mathematical Physics}, 14(12):1938--1941, December
  1973.

\bibitem[Mat13]{Mat13}
Keiji Matsumoto.
\newblock A new quantum version of $f$-divergence, 2013.
\newblock arXiv:1311.4722.

\bibitem[Mat14]{matsumoto2014maximization}
Keiji Matsumoto.
\newblock On maximization of measured $f$-divergence between a given pair of
  quantum states, 2014.
\newblock arXiv:1412.3676.

\bibitem[Mat18]{Matsumoto2018}
Keiji Matsumoto.
\newblock A new quantum version of $f$-divergence.
\newblock In Masanao Ozawa, Jeremy Butterfield, Hans Halvorson, Mikl{\'o}s
  R{\'e}dei, Yuichiro Kitajima, and Francesco Buscemi, editors, {\em Reality
  and Measurement in Algebraic Quantum Theory}, volume 261, pages 229--273,
  Singapore, 2018. Springer Singapore.
\newblock Series Title: Springer Proceedings in Mathematics \& Statistics.

\bibitem[MBV22]{mosonyi2022geometric}
Mil{\'a}n Mosonyi, Gergely Bunth, and P{\'e}ter Vrana.
\newblock Geometric relative entropies and barycentric {R}{\'e}nyi divergences.
\newblock \href{https://arxiv.org/abs/2207.14282v2}{arXiv:2207.14282v2}, July
  2022.

\bibitem[MO15]{mosonyi2015quantum}
Mil{\'a}n Mosonyi and Tomohiro Ogawa.
\newblock Quantum hypothesis testing and the operational interpretation of the
  quantum {R}{\'e}nyi relative entropies.
\newblock {\em Communications in Mathematical Physics}, 334:1617--1648, 2015.

\bibitem[NS09]{ns_chernoff}
Michael Nussbaum and Arleta Szkoła.
\newblock The {C}hernoff lower bound for symmetric quantum hypothesis testing.
\newblock {\em The Annals of Statistics}, 37(2):1040--1057, April 2009.
\newblock
  \href{https://arxiv.org/abs/quant-ph/0607216}{arXiv:quant-ph/0607216}.

\bibitem[PBR12]{pusey2012reality}
Matthew~F. Pusey, Jonathan Barrett, and Terry Rudolph.
\newblock On the reality of the quantum state.
\newblock {\em Nature Physics}, 8(6):475--478, June 2012.
\newblock \href{https://arxiv.org/abs/1111.3328}{arXiv:1111.3328}.

\bibitem[RS23]{russo2023}
Vincent Russo and Jamie Sikora.
\newblock Inner products of pure states and their antidistinguishability.
\newblock {\em Physical Review A}, 107(3):L030202, March 2023.
\newblock \href{https://arxiv.org/abs/2206.08313}{arXiv:2206.08313}.

\bibitem[Rus02]{ruskai2002inequalities}
Mary~Beth Ruskai.
\newblock Inequalities for quantum entropy: A review with conditions for
  equality.
\newblock {\em Journal of Mathematical Physics}, 43(9):4358--4375, 2002.

\bibitem[Sal73]{salikhov1973asymptotic}
N.~P. Salikhov.
\newblock Asymptotic properties of error probabilities of tests for
  distinguishing between several multinomial testing schemes.
\newblock {\em Doklady Akademii Nauk}, 209(1):54--57, 1973.

\bibitem[Sal99]{salikhov1999one}
N.~P. Salikhov.
\newblock On one generalization of {C}hernov's distance.
\newblock {\em Theory of Probability \& Its Applications}, 43(2):239--255,
  1999.

\bibitem[Sal03]{salikhov2003optimal}
N.~P. Salikhov.
\newblock Optimal sequences of tests for several polynomial schemes of trials.
\newblock {\em Theory of Probability \& Its Applications}, 47(2):286--298,
  2003.

\bibitem[Sch12]{schervish2012theory}
Mark~J. Schervish.
\newblock {\em Theory of Statistics}.
\newblock Springer Science \& Business Media, 2012.

\bibitem[Shi16]{shiryaev2016probability}
Albert~N. Shiryaev.
\newblock {\em Probability-1}, volume~95.
\newblock Springer, 2016.

\bibitem[Sio58]{sion_1958}
Maurice Sion.
\newblock On general minimax theorems.
\newblock {\em Pacific Journal of Mathematics}, 8:171--176, 1958.

\bibitem[Str11]{strasser2011mathematical}
Helmut Strasser.
\newblock {\em Mathematical Theory of Statistics: Statistical Experiments and
  Asymptotic Decision Theory}, volume~7.
\newblock Walter de Gruyter, 2011.

\bibitem[Tor81]{torgersen1981measures}
Erik~N. Torgersen.
\newblock Measures of information based on comparison with total information
  and with total ignorance.
\newblock {\em The Annals of Statistics}, 9(3):638--657, 1981.

\bibitem[Tor91]{torgersen1991comparison}
Erik~N. Torgersen.
\newblock {\em Comparison of Statistical Experiments}, volume~36.
\newblock Cambridge University Press, 1991.

\bibitem[Tou74]{toussaint1974some}
Godfried~T. Toussaint.
\newblock Some properties of {M}atusita’s measure of affinity of several
  distributions.
\newblock {\em Annals of the Institute of Statistical Mathematics},
  26(3):389--394, 1974.

\bibitem[Ume62]{umegaki_1962}
Hisaharu Umegaki.
\newblock Conditional expectation in an operator algebra, {{IV}} ({{Entropy}}
  and information).
\newblock {\em Kodai Mathematical Seminar Reports}, 14:59--85, 1962.

\bibitem[Wil17]{wilde2017quantum}
Mark~M. Wilde.
\newblock {\em Quantum Information Theory}.
\newblock Cambridge University Press, second edition, 2017.

\bibitem[WW20]{ww2020}
Xin Wang and Mark~M. Wilde.
\newblock $\ensuremath{\alpha}$-{L}ogarithmic negativity.
\newblock {\em Physical Review A}, 102:032416, September 2020.
\newblock \href{https://arxiv.org/abs/1904.10437}{arXiv:1904.10437}.

\bibitem[YKL75]{yuen_kennedy_lax1975}
H.~Yuen, R.~Kennedy, and M.~Lax.
\newblock Optimum testing of multiple hypotheses in quantum detection theory.
\newblock {\em IEEE Transactions on Information Theory}, 21(2):125--134, March
  1975.

\end{thebibliography}

\appendix

\section{Expectation values at non-corner points}\label{app:explicit-form-gamma-i-non-corner-points}
We begin by stating a known property of convex functions in the lemma below. We include a proof of the statement for the sake of completeness.
\begin{lemma}\label{lem-right-deriv}
Let $a>0$ be arbitrary. Let $f:[0,a] \to \mathbb{R}$ be a convex and continuous function on $[0,a]$, and suppose $f$ is differentiable on $\left(
0,a\right) $. Then the one-sided derivative%
\begin{align}\label{eq:convex-one-sided-der-at-zero}
f_{+}^{\prime}\left(  0\right)  \coloneqq \lim_{t\searrow0}\frac{f\!\left(  t\right)
-f\!\left(  0\right)  }{t}%
\end{align}
exists and fulfills
\begin{equation}
f_{+}^{\prime}\left(  0\right)  =\lim_{t\searrow0}f^{\prime}\left(  t\right)
. \label{2nd-claim}%
\end{equation}
Here $f_{+}^{\prime}\left(  0\right)$ is either finite or takes the value $-\infty$; if $f$ takes its minimum value at $0$, then $f_{+}^{\prime}\left(  0\right)$ is finite and $f_{+}^{\prime}\left(  0\right) \geq 0$. 
\end{lemma}

\begin{proof}
The map $t \mapsto (f(t)-f(0))/t$ defined on $(0,a)$ is non-decreasing. See \cite[Section~2.1, Exercise~7]{Borwein2006}). Also, the limit in \eqref{eq:convex-one-sided-der-at-zero} exists in $\mathbb{R}\cup \{-\infty\}$ \cite[Proposition~3.1.2]{Borwein2006}.
By the Lagrange mean-value theorem, for any $t \in (0,a)$ there exists $u_t \in (0,t)$ such that
\begin{align}
    \frac{f\!\left(  t\right)
-f\!\left(  0\right)  }{t} = f^\prime (u_t).
\end{align}
This implies
\begin{align}\label{eq:convex-one-sided-der-at-zero-2}
    f_+^\prime(0)=\lim_{t \searrow 0} \frac{f\!\left(  t\right)
-f\!\left(  0\right)  }{t} = \lim_{t \searrow 0} f^\prime (u_t).
\end{align}
The derivative of $f$ is a non-decreasing function on $(0,a)$. Indeed, let $t_1, t_2 \in (0,a)$ such that $t_1 < t_2$. Let $h \in (t_1, t_2)$ be arbitrary.
We have
\begin{align}
    h &= \left(\dfrac{h-t_1}{t_2-t_1} \right) t_2 + \left(1-\dfrac{h-t_1}{t_2-t_1} \right) t_1.
\end{align}
By convexity of $f$, we get
\begin{align}
    f(h) & \leq \left(\dfrac{h-t_1}{t_2-t_1} \right) f(t_2) + \left(1-\dfrac{h-t_1}{t_2-t_1} \right) f(t_1),
\end{align}
which implies
\begin{align}
    \dfrac{f(h)-f(t_1)}{h-t_1} \leq \dfrac{f(t_2)-f(t_1)}{t_2-t_1}.
\end{align}
By taking the limit $h \searrow t_1$ in the above inequality, we get
\begin{align}\label{eq:f-prime-left}
    f^{\prime}(t_1) \leq \dfrac{f(t_2)-f(t_1)}{t_2-t_1}.
\end{align}
A similar argument shows that
\begin{align}\label{eq:f-prime-right}
    \dfrac{f(t_2)-f(t_1)}{t_2-t_1} \leq f^{\prime}(t_2).
\end{align}
By \eqref{eq:f-prime-left} and \eqref{eq:f-prime-right}, we thus get $f^\prime(t_1) \leq f^\prime(t_2)$ and hence $f^\prime$ is non-decreasing on $(0,a)$. From \eqref{eq:convex-one-sided-der-at-zero-2} and the fact that $f^\prime$ is non-decreasing, we thus get
\begin{align}
    f_+^\prime(0)=\lim_{t \searrow 0} f^\prime(t),
\end{align}
with a possible value $-\infty$.
If $f$ is minimized at $0$, then we have $f(t)-f(0) \geq 0$ for all $t\in (0,a)$. It then directly follows from the definition \eqref{eq:convex-one-sided-der-at-zero} that $f_+^\prime (0) \geq 0$.
\end{proof}

Recall that $\mathbb{T}_r^1$ is the set of non-corner points of $\mathbb{T}_r$ given by \eqref{eq:corner-point-def}.
Let $\mathbf{t} \in \mathbb{T}_r^1$. Define a set
\begin{align}
    B_{\mathbf{t}} \coloneqq \left\{  i\in\left[  r-1\right]  : t_{i}>0\right\},
\end{align}
and let $B_{\mathbf{t}}^c \coloneqq [r-1] \backslash B_{\mathbf{t}}$. Let $\beta$ denote the cardinality of the set $B_{\mathbf{t}}$ so that $1 \leq \beta \leq r-1$.
We emphasize that $\mathbf{t}$ corresponds to an interior point of $\mathbb{T}_{\beta+1}$, which is the $\beta$-vector obtained by discarding the zero entries of $\mathbf{t}$. This allows us to use the properties of exponential family of densities given in \eqref{eq:density-pt-expression}.
 So, by the similar arguments as given for \eqref{neweqn3}, it follows that for any $i \in B_{\mathbf{t}}$, the expectation value $\mathbb{E}_{\mathbf{t}} [q_i]$ exists, and it satisfies $\partial_i\! \operatorname{K} (\mathbf{t})=\mathbb{E}_{\mathbf{t}} [q_i]$.
It remains to show for $i \in B_{\mathbf{t}}^c$ that $\mathbb{E}_{\mathbf{t}} [q_i]$ exists, and it is equal to $\partial_i^+\! \operatorname{K} (\mathbf{t})$. 
Let us fix an arbitrary index $i \in B_{\mathbf{t}}^c$. Choose a small number $\varepsilon > 0$ such that $\mathbf{t}+h \mathbf{e}_i \in \mathbb{T}_r^1$ for all $h \in [0,\varepsilon]$. The function $h \mapsto \operatorname{K}(\mathbf{t}+h \mathbf{e}_i)$ is continuous, convex on $[0, \varepsilon]$, and it is differentiable on $(0,\varepsilon)$. Lemma~\ref{lem-right-deriv} thus implies that 
\begin{align}\label{eq:existence-exp-non-corner-point}
    \partial_i^+\!\operatorname{K}(\mathbf{t}) = \lim_{h \searrow 0} \partial_i\!\operatorname{K}(\mathbf{t}+h \mathbf{e}_i)=\lim_{h \searrow 0} \mathbb{E}_{\mathbf{t}+h\mathbf{e}_i}[q_i].
\end{align}
Here we used the relation $\partial_i\!\operatorname{K}(\mathbf{t}+h \mathbf{e}_i)= \mathbb{E}_{\mathbf{t}+h\mathbf{e}_i}[q_i]$ proved earlier. 
We now claim that $\mathbb{E}_{\mathbf{t}}[q_i]$ exists and satisfies
\begin{equation}
\lim_{h \searrow 0} \mathbb{E}_{\mathbf{t}+h\mathbf{e}_i}[q_i]= \mathbb{E}_{\mathbf{t}}[q_i]\label{claim-crucial}%
\end{equation}
with a possible value $-\infty$.
Indeed, we have
\begin{align}
\mathbb{E}_{\mathbf{t}+h\mathbf{e}_i}[q_i]  & =\frac{1}{\operatorname{H}\left(
\mathbf{t}+h\mathbf{e}_{i}\right)  }\int\!\!\d\mu\ q_{i}p_{r}\exp\left(
\sum_{j\in\left[  r-1\right]  }t_{j}q_{j}+hq_{i}\right) .
\end{align}
By continuity of $\operatorname{H}$, we have $\operatorname{H}\left(
\mathbf{t}+h\mathbf{e}_{i}\right)  \rightarrow \operatorname{H}\left(  \mathbf{t}\right) $ as $h\searrow0$. Thus, for \eqref{claim-crucial} to hold, it suffices to prove
\begin{equation}
\lim_{h\searrow0}\int\!\!\d\mu\ q_{i}p_{r}\exp\left(
\sum_{j\in \left[r-1\right] }t_{j}q_{j}+hq_{i}\right)=\int\!\!\d\mu\ q_{i}p_{r}\exp\left(
\sum_{j\in \left[r-1\right] }t_{j}q_{j}\right).\label{cc-2}%
\end{equation}
Let $q_i = q_{i}^+-q_{i}^-$, where $q_i^+$ and $q_i^-$ are non-negative functions with mutually disjoint supports. This gives
\begin{multline}\label{eq:pos-neg-part-qi-exp}
    \int\!\!\d\mu\ q_{i}p_{r}\exp\left(
\sum_{j\in \left[r-1\right] }t_{j}q_{j}+hq_{i}\right)  \\
    = \int\!\!\d\mu\ q_{i}^+p_{r}\exp\left(
\sum_{j\in \left[r-1\right] }t_{j}q_{j}+hq_{i}^+\right) - \int\!\!\d\mu\ q_{i}^-p_{r}\exp\left(
\sum_{j\in \left[r-1\right] }t_{j}q_{j}-hq_{i}^-\right).
\end{multline}
The first integral term in the right-hand side of \eqref{eq:pos-neg-part-qi-exp} is finite. Indeed,
 since $h \in (0,\varepsilon)$,
we have
\begin{align}
    q_{i}^+p_{r}\exp\left(
\sum_{j\in \left[r-1\right] }t_{j}q_{j}+hq_{i}^+\right) \leq
q_{i}^+p_{r}\exp\left(
\sum_{j\in \left[r-1\right] }t_{j}q_{j}+\varepsilon q_{i}^+\right).\label{eq:upper-bd-qi+}
\end{align}
The function in the right-hand side of  \eqref{eq:upper-bd-qi+} is $\mu$-integrable,
because $\mathbf{t}+\varepsilon \mathbf{e}_i$ corresponds to an interior point of
$\mathbb{T}_{r-\beta+1}$ so that the properties of an exponential family of densities apply.
 Moreover, we have pointwise convergence
\begin{align}
q_{i}^+p_{r}\exp\left(
\sum_{j\in \left[r-1\right] }t_{j}q_{j}+hq_{i}^+\right)  \rightarrow q_{i}^+p_{r}\exp\left(
\sum_{j\in \left[r-1\right] }t_{j}q_{j}\right) \qquad \text{as }h\searrow0.
\end{align}
By the Lebesgue dominated convergence theorem, we have
\begin{align}
\lim_{h \searrow 0} \int\!\!\d\mu\ q_{i}^+p_{r}\exp\left(
\sum_{j\in \left[r-1\right] }t_{j}q_{j}+hq_{i}^+\right) = \int\!\!\d\mu\ q_{i}^+p_{r}\exp\left(
\sum_{j\in \left[r-1\right] }t_{j}q_{j}\right) < \infty. \label{eq:expectation_qi+}
\end{align}
We now consider the second integral term in the right-hand side of \eqref{eq:pos-neg-part-qi-exp}.
We have the pointwise monotone convergence
\begin{align}
q_{i}^-p_{r}\exp\left(
\sum_{j\in \left[r-1\right] }t_{j}q_{j}-hq_{i}^-\right) \nearrow q_{i}^-p_{r}\exp\left(
\sum_{j\in \left[r-1\right] }t_{j}q_{j}\right), \qquad \text{as } h \searrow 0.
\end{align}
By the monotone convergence theorem, we get
\begin{align}
\lim_{h \searrow 0}\int\!\!\d\mu\ q_{i}^-p_{r}\exp\left(
\sum_{j\in \left[r-1\right] }t_{j}q_{j}-hq_{i}^-\right) = \int\!\!\d\mu\ q_{i}^-p_{r}\exp\left(
\sum_{j\in \left[r-1\right] }t_{j}q_{j}\right) \label{eq:expectation_qi-}
\end{align}
regardless of whether the right-hand integral in \eqref{eq:expectation_qi-} is finite or infinite. 
The latter
point is explicitly stressed in Theorem~16.2 of \cite{billingsley1995probability}. 
By taking the limit $h \searrow 0$ in \eqref{eq:pos-neg-part-qi-exp}, and then using \eqref{eq:existence-exp-non-corner-point}, \eqref{eq:expectation_qi+} and \eqref{eq:expectation_qi-}, we get
\begin{align}
    \partial_i^+\!K(\mathbf{t}) = \mathbb{E}_{\mathbf{t}}[q_i^+]-\mathbb{E}_{\mathbf{t}}[q_i^-]=\mathbb{E}_{\mathbf{t}}[q_i].
\end{align}
Since $\mathbb{E}_{\mathbf{t}}[q_i^+]$ is a real number,  $\mathbb{E}_{\mathbf{t}}[q_i]$ takes a value in $\mathbb{R} \cup \{-\infty\}$. If $\mathbf{t}$ is a minimizer of $\operatorname{K}$ then by Lemma~\ref{lem-right-deriv} we have $\partial_i^+\!K(\mathbf{t}) \geq 0$, and hence $\mathbb{E}_{\mathbf{t}}[q_i]$ is finite.
We have thus accomplished that if  $\mathbf{t}\in\mathbb{T}_{r}^{1}$ is
a minimizer of $\operatorname{K}$ and $i\in\lbrack r-1]$, then the expectation value
$\mathbb{E}_{\mathbf{t}}[q_{i}]$ exists, is finite,  and satisfies
$\partial_{i}^{+}\!\operatorname{K}(\mathbf{t})=\mathbb{E}_{\mathbf{t}}[q_{i}]$.

\section{Proof of Theorem~\ref{thm:optimal-classical-general}}

\label{app:general-classical-AD-error-exponent}
One of the major differences between the special case of mutually absolutely continuous probability measures and the general case of probability measures with unequal supports is that, the Hellinger transform is continuous on the entire unit simplex in the special case, but only in the interior of the unit simplex in the general case.
As it turns out, in the definition \eqref{eq:classical-CH-divergence}, restricting the infimum to be over the interior of the unit simplex does not change the value of the   multivariate classical Chernoff divergence.

Define
\begin{equation}\label{eq:def-multi-chernoff-interior}
\xi_{\operatorname{cl}}^{\circ}\left(  P_{1},\ldots,P_{r}\right)  \coloneqq  -\ln \inf_{\mathbf{s}%
\in\mathbb{S}_{r}^{\circ}}\mathbb{H}_{\mathbf{s}}\left(  P_{1},\ldots
,P_{r}\right),
\end{equation}
where $\mathbb{S}_{r}^{\circ}$ denotes the interior of the unit simplex $\mathbb{S}_{r}$.
 We will
show that $\xi_{\operatorname{cl}}^{\circ}\left(  P_{1},\ldots,P_{r}\right)$ coincides with the   multivariate classical Chernoff divergence $\xi_{\operatorname{cl}}\left(  P_{1},\ldots,P_{r}\right)$ defined in \eqref{eq:classical-CH-divergence}, and also that it is the optimal error exponent for antidistinguishing the given probability measures.

As a first step towards showing $\xi_{\operatorname{cl}}^{\circ}\left(  P_{1},\ldots,P_{r}\right)=\xi_{\operatorname{cl}}\left(  P_{1},\ldots,P_{r}\right)$, we discuss a continuous extension of the Hellinger transform defined in the interior of the unit simplex to its boundary.
Recall that $\mu$ is the dominating measure given by \eqref{eq:dominating-measure} and $p_1,\ldots, p_r$ are the induced densities defined in \eqref{eq:induced-densities}. 
 Define
 \begin{align}
     D &  \coloneqq \{\omega \in \Omega: p_i(\omega)>0, \ \forall i \in [r]\}.
 \end{align}
For any $\mathbf{s}\coloneqq(s_1,\ldots, s_r) \in \mathbb{S}_r$, define
\begin{align}
    \mathbb{H}_{\mathbf{s}}^-(P_1,\ldots, P_r) &\coloneqq \int_D\!\!\d\mu\
p_{1} ^{s_{1}}\cdots p_{r} ^{s_{r}}.
\end{align}
We note that $\mathbb{H}_{\mathbf{s}}^-(P_1,\ldots, P_r)=\mathbb{H}_{\mathbf{s}}(P_1,\ldots, P_r)$ for all $\mathbf{s} \in \mathbb{S}_r^\circ$, which implies that the map $\mathbf{s} \mapsto \mathbb{H}_{\mathbf{s}}^-(P_1,\ldots, P_r)$ in continuous in the interior of the unit simplex $\mathbb{S}_r$. Also, its continuity on the boundary of the unit simplex follows by the Lebesgue dominated convergence theorem.
It is easy to see that $\mathbb{H}_{\mathbf{s}}^-(P_1,\ldots, P_r) \leq \mathbb{H}_{\mathbf{s}}(P_1,\ldots, P_r)$ for all $\mathbf{s} \in \mathbb{S}_r$. Indeed, let 
$\mathbf{s}\coloneqq(s_1,\ldots, s_r)\in \mathbb{S}_r$ be arbitrary, set $J \coloneqq \{j \in [r]: s_j> 0\}$ and $D_J\coloneqq \{\omega \in \Omega: p_j(\omega)>0, \ \forall j \in J\}$. Since $D \subset D_J$, we get
\begin{align}
    \mathbb{H}_{\mathbf{s}}^-(P_1,\ldots, P_r) 
    &= \int_D\!\!\d\mu\
    {\displaystyle\prod\limits_{j\in J}}
    p_{j} ^{s_{j}} \\
    &\leq \int_{D_J}\!\!\d\mu\
    {\displaystyle\prod\limits_{j\in J}}
    p_{j} ^{s_{j}} \\
    &= \mathbb{H}_{\mathbf{s}}(P_1,\ldots, P_r).
\end{align}
It should be noted that the last equality follows from the convention $0^0=1$ in the definition \eqref{eq:hellinger-trans-def} of the Hellinger transform. We thus get from \eqref{eq:def-multi-chernoff-interior} that
\begin{align}
\xi_{\operatorname{cl}}^\circ\left(  P_{1},\ldots,P_{r}\right)   
    & =-\ln \inf_{\mathbf{s}%
\in\mathbb{S}_{r}^{\circ}}\mathbb{H}_{\mathbf{s}}^-\left(  P_{1},\ldots
,P_{r}\right) \\  
&=-\ln \inf_{\mathbf{s}%
\in\mathbb{S}_{r}}\mathbb{H}_{\mathbf{s}}^-\left(  P_{1},\ldots
,P_{r}\right) \\
&=-\ln \inf_{\mathbf{s}\in\mathbb{S}_{r}}\mathbb{H}_{\mathbf{s}%
}\left(  P_{1},\ldots,P_{r}\right)  \\
 &=\xi_{\operatorname{cl}}\left(  P_{1},\ldots,P_{r}\right),
 \label{eq:xi-0-equals-xi}
\end{align}
as was claimed earlier.

For an achievable bound on the optimal error exponent, we follow the same arguments as in the development \eqref{eq:upper-bound-optimal-error-n-case}--\eqref{eq:upper-bound-optimal-error-n-case-last} to get
\begin{align}
\mathrm{Err}_{\mathrm{cl}}\left(\mathcal{E}_{\operatorname{cl}}^n\right) \leq\mathbb{H}_{\mathbf{s}}\left(  P_{1},\ldots,P_{r}\right)  ^{n}, \qquad \text{for all } \  \mathbf{s} \in \mathbb{S}_r^\circ.
\end{align}
This gives
\begin{align}
     \liminf_{n \to \infty} -\dfrac{1}{n} \ln \operatorname{Err}_{\operatorname{cl}} (\mathcal{E}_{\operatorname{cl}}^{n}) 
     & \geq -\ln \inf_{\mathbf{s}\in\mathbb{S}_r^\circ}  \mathbb{H}_{\mathbf{s}}(P_1,\ldots,P_r)=\xi_{\operatorname{cl}}^\circ(P_1,\ldots,P_r)=\xi_{\operatorname{cl}}(P_1,\ldots,P_r).\label{eq:achievability-error-exponent-general}
\end{align}

We now present an optimality bound on the optimal error exponent. The case $\mu(D)=0$ is trivial, since the antidistinguishability error probability is equal to zero in this case. So, we assume $\mu\left(  D\right)  >0$.
Define
\begin{align}\label{eq:def-alpha-i}
    \alpha_i \coloneqq P_i(D)>0, \qquad i \in [r]
\end{align}
and {\it conditional probability densities} on the set $D$ as
\begin{align}\label{eq:def-p-i-tilde}
    \tilde{p}_i \coloneqq \dfrac{p_i}{\alpha_i}, \qquad i \in [r].
\end{align}
Observe that the conditional probability densities $\tilde{p}_1,\ldots, \tilde{p}_r$ correspond to mutually absolutely continuous probability measures $P_1/\alpha_1,\ldots, P_r/\alpha_r$ on the set $D$. 
Let us define for $\mathbf{t}\coloneqq (t_1,\ldots, t_{r-1})\in \mathbb{T}_r$ and $\mathbf{s}_{\mathbf{t}}=(t_1,\ldots, t_{r-1}, 1-\sum_{i\in[r-1]}t_i) \in \mathbb{S}_r$,
\begin{align}
    \operatorname{H}^-(\mathbf{t}) &\coloneqq \mathbb{H}^-_{\mathbf{s}_{\mathbf{t}}}(P_1,\ldots,P_r), \label{eq:def-H-minus} \\
    \operatorname{K}^-(\mathbf{t}) &\coloneqq \ln \operatorname{H}^-(\mathbf{t}).\label{eq:def-K-minus}
\end{align}
From \eqref{eq:def-p-i-tilde} we get,
\begin{align}
    \operatorname{H}^-(\mathbf{t}) 
        &=  \int_D\!\!\d\mu\ p_1^{t_1}\cdots p_{r-1}^{t_{r-1}} p_r^{1-\sum_{i\in[r-1]}t_i} \\
        &= \alpha_1^{t_1}\cdots \alpha_{r-1}^{t_{r-1}} \alpha_r^{1-\sum_{i\in[r-1]}t_i}  \int_D\!\!\d\mu\ \tilde{p}_1^{t_1}\cdots \tilde{p}_{r-1}^{t_{r-1}} \tilde{p}_r^{1-\sum_{i\in[r-1]}t_i}. \label{eq:K-minus-expansion}
\end{align}
Define
\begin{align}
    \tilde{\operatorname{H}}(\mathbf{t}) 
        &\coloneqq  \int_D\!\!\d\mu\ \tilde{p}_1^{t_1}\cdots \tilde{p}_{r-1}^{t_{r-1}} \tilde{p}_r^{1-\sum_{i\in[r-1]}t_i}, \label{eq:def-H-tilde} \\
    \tilde{\operatorname{K}}(\mathbf{t}) 
        &\coloneqq \ln \tilde{\operatorname{H}}(\mathbf{t}). \label{eq:def-K-tilde}
\end{align}
From the above development \eqref{eq:def-H-minus}--\eqref{eq:def-K-tilde}, we get
\begin{align}\label{eq:K(minus)-K(tilde)-relation}
    \operatorname{K}^-(\mathbf{t}) = \tilde{\operatorname{K}}(\mathbf{t}) + \operatorname{L}(\mathbf{t}), \qquad \text{for all } \mathbf{t} \in \mathbb{T}_r,
\end{align}
where $\operatorname{L}$ is a linear map defined on $\mathbb{T}_r$ as
\begin{align}
    \operatorname{L}(\mathbf{t}) \coloneqq \ln \alpha_r+\sum_{i\in [r-1]} t_i \ln \left(\dfrac{\alpha_i}{\alpha_r}\right).
\end{align}
Analogous to \eqref{eq:density-pt-expression}, we define an exponential family of densities on $D$ as
\begin{align}
    \tilde{p}_{\mathbf{t}} \coloneqq \dfrac{1}{\tilde{\operatorname{H}}(\mathbf{t})} \tilde{p}_1^{t_1}\cdots \tilde{p}_{r-1}^{t_{r-1}} \tilde{p}_r^{1-\sum_{i\in[r-1]}t_i}.
\end{align}
Recall from the definition of the antidistinguishability error probability \eqref{eq:optimal-AD-error-probability} that we only need to work with the probability densities $p_1,\ldots, p_r$ restricted to the set $D$; i.e., we have for all positive integers $n$,
\begin{align}
\operatorname{Err}_{\operatorname{cl}}(\mathcal{E}_{\operatorname{cl}}^{n}) =  \int_D\!\!\d\mu^{\otimes n}\ \left(\eta_1 p_{1}^{\otimes n}\land \cdots \land \eta_r p_{r}^{\otimes n}\right).
\end{align}
Define 
\begin{equation}
\tilde{\mathbb{T}}_{r,f}^{1}\coloneqq \left\{  \mathbf{t}\in\mathbb{T}_{r}^{1}:\partial
_{i}^{+}\tilde{\operatorname{K}}(\mathbf{t})\neq-\infty, \ \forall i\in\left[  r-1\right]  \right\}
\label{Trf-tilde-definition}%
\end{equation}
analogous to \eqref{Trf-definition}.
Let $\omega^{n}\coloneqq (\omega_1,\ldots,\omega_n) \in D^{n}$ and $\mathbf{t} \in \tilde{\mathbb{T}}_{r,f}^{1}$ be arbitrary. We have that
\begin{align}
    p_i^{\otimes n} (\omega^n)&= \exp\! \left(n \tilde{G}^{(i)}_{\mathbf{t},n}(\omega^n)\right)\tilde{p}_{\mathbf{t}}^{\otimes n} (\omega^n),\label{eq:tensored-tilde-density-expression}
\end{align}
where
\begin{align}
\tilde{G}^{(i)}_{\mathbf{t},n}(\omega^n) &\coloneqq \dfrac{1}{n}  \sum_{j\in[n]} \ln \dfrac{\tilde{p}_i}{\tilde{p}_{\mathbf{t}}} (\omega_j) + \ln \alpha_i \qquad \text{for } i\in[r].\label{eq:def-G-i-tilde-2}
\end{align}
By following similar arguments as given in the development \eqref{eq:tensored-density-expression}--\eqref{eq:classical-error-exp-upp-bound-gamma-i}, we get
\begin{align}
    \limsup_{n \to \infty} - \dfrac{1}{n} \ln \operatorname{Err}_{\operatorname{cl}}(\mathcal{E}_{\operatorname{cl}}^n) \leq  -\min_{1\leq i \leq r} \left(\tilde{\gamma}_{i}(\mathbf{t})+\ln \alpha_i \right),\label{eq:classical-error-exp-upp-bound-gamma-i-tilde}
\end{align}
such that, similar to \eqref{eq:gamma-i-compact-form},
\begin{align}\label{eq:gamma-tilde-i-compact-form}
    \tilde{\gamma}_i (\mathbf{t}) =
    \begin{cases}
        \partial_i^+\tilde{\operatorname{K}}(\mathbf{t}) - \mathbf{t}^T \nabla^+ \tilde{\operatorname{K}}(\mathbf{t}) + \tilde{\operatorname{K}}(\mathbf{t}), & i\in [r-1], \\
        - \mathbf{t}^T \nabla^+ \tilde{\operatorname{K}}(\mathbf{t}) + \tilde{\operatorname{K}}(\mathbf{t}), & i=r,
    \end{cases} \qquad \text{for }  \mathbf{t} \in \tilde{\mathbb{T}}_{r,f}^{1}.
\end{align}
From \eqref{eq:K(minus)-K(tilde)-relation}, we have for $\mathbf{t}\in\mathbb{T}_{r}^{1}$,
\begin{align}
    \partial_{i}^{+}\tilde{\operatorname{K}}\!\left(  \mathbf{t}\right)  =\partial_{i}^{+}%
\operatorname{K}^{-}\!\left(  \mathbf{t}\right)  -\ln\!\left(  \frac{\alpha_{i}}{\alpha_{r}%
}\right)\label{eq:Kminus-Ktilde-der-relation}
\end{align}
and hence
\begin{align}
\mathbf{t}^{T}\nabla^{+}\tilde{\operatorname{K}}\!\left(  \mathbf{t}\right)  =\mathbf{t}%
^{T}\nabla^{+} \operatorname{K}^{-}\!\left(  \mathbf{t}\right)  -\operatorname{L}\!\left(  \mathbf{t}\right)
+\ln\alpha_{r}.\label{eq:gradient-Ktilde-Kminus-relation}
\end{align}
Substituting the relations \eqref{eq:Kminus-Ktilde-der-relation} and \eqref{eq:gradient-Ktilde-Kminus-relation} in \eqref{eq:gamma-tilde-i-compact-form}, we obtain
\begin{align}\label{eq:gamma-tilde-i-compact-form-2}
    \tilde{\gamma}_{i}\left(  \mathbf{t}\right)  +\ln\alpha_{i} =
    \begin{cases}
       \partial_{i}^{+}\operatorname{K}^{-}\!\left(  \mathbf{t}\right)  -\mathbf{t}^{T}\nabla
^{+}\operatorname{K}^{-}\!\left(  \mathbf{t}\right)  +\operatorname{K}^{-}\!\left(  \mathbf{t}\right) , & i\in [r-1], \\
     -\mathbf{t}^{T}\nabla^{+}\!\operatorname{K}^{-}\!\left(  \mathbf{t}\right)  +\operatorname{K}^{-}\!\left(
\mathbf{t}\right)   , & i=r,
    \end{cases} \qquad \text{for }  \mathbf{t} \in \tilde{\mathbb{T}}_{r,f}^{1}.
\end{align}
As the function $\operatorname{L}$ in \eqref{eq:K(minus)-K(tilde)-relation} is linear and $\tilde{\operatorname{K}}$ is convex on
$\mathbb{T}_{r}$, we conclude that $\operatorname{K}^{-}$ is convex as well, and it inherits
the smoothness properties of $\tilde{\operatorname{K}}$.
One can now essentially follow the steps described in the development \eqref{eqn4}--\eqref{eq:reparametrization-minimizer} to show that there is no loss in assuming that 
there exists a minimizer $\mathbf{t}^* \in \tilde{\mathbb{T}}_{r,f}^{1}$ of $\operatorname{K}^-$, and that it satisfies
\begin{align}
     -\min_{1\leq i \leq r} \left(\tilde{\gamma}_{i}(\mathbf{t}^*)+\ln \alpha_i \right) \leq -\operatorname{K}^-(\mathbf{t}^*).\label{eq:gamma-tilde-K-inequality}
\end{align}
Combining \eqref{eq:classical-error-exp-upp-bound-gamma-i-tilde} and \eqref{eq:gamma-tilde-K-inequality}, we thus get
\begin{align}
    \limsup_{n \to \infty} - \dfrac{1}{n} \ln \operatorname{Err}_{\operatorname{cl}}(\mathcal{E}_{\operatorname{cl}}^n) \leq \sup_{\mathbf{t} \in \mathbb{T}_r} -\operatorname{K}^-(\mathbf{t}) = \xi_{\operatorname{cl}}^\circ(P_1,\ldots,P_r)=\xi_{\operatorname{cl}}(P_1,\ldots,P_r). \label{eq:optimality-error-exponent-general}
\end{align}
\sloppy
This completes the optimality part. It follows from \eqref{eq:achievability-error-exponent-general} and \eqref{eq:optimality-error-exponent-general} that the limit $\lim_{n \to \infty} -\frac{1}{n}\ln \operatorname{Err}_{\operatorname{cl}}(\mathcal{E}_{\operatorname{cl}}^n)$ exists and the desired equality \eqref{eq:classical-asymptotic-error-rate-general} holds.

\section{Proof of Equation~\eqref{eq:pure-state-identity-TD}}

\label{app:pure-state-identity-TD}

\begin{proposition}
\label{prop:pure-state-identity-TD}
For arbitrary (not necessarily normalized)\ vectors $|\varphi\rangle
,|\zeta\rangle\in\mathcal{H}$, the following equality holds:%
\begin{equation}\label{eq:trace-norm-vectors-formula}
\left\Vert |\varphi\rangle\!\langle\varphi|-|\zeta\rangle\!\langle\zeta
|\right\Vert _{1}^{2}=\left(  \langle\varphi|\varphi\rangle+\langle\zeta
|\zeta\rangle\right)  ^{2}-4\left\vert \langle\zeta|\varphi\rangle\right\vert
^{2}.
\end{equation}

\end{proposition}

\begin{proof}
The equality \eqref{eq:trace-norm-vectors-formula} trivially holds if one of the vectors is zero. So, we assume that both $|\varphi\rangle$ and $|\zeta\rangle$ are non-zero vectors.
Define%
\begin{equation}
|\varphi^{\prime}\rangle\coloneqq \frac{|\varphi\rangle}{\left\Vert |\varphi
\rangle\right\Vert},\qquad|\zeta^{\prime}\rangle\coloneqq \frac{|\zeta\rangle
}{\left\Vert |\zeta\rangle\right\Vert}.
\end{equation}
Then the desired equality is equivalent to%
\begin{equation}
\left\Vert c|\varphi^{\prime}\rangle\!\langle\varphi^{\prime}|-d|\zeta^{\prime
}\rangle\!\langle\zeta^{\prime}|\right\Vert _{1}^{2}=\left(  c+d\right)
^{2}-4cd\left\vert \langle\zeta^{\prime}|\varphi^{\prime}\rangle\right\vert
^{2},
\end{equation}
where%
\begin{equation}
c\coloneqq \left\Vert |\varphi\rangle\right\Vert^{2},\qquad d\coloneqq \left\Vert
|\zeta\rangle\right\Vert^{2}.
\end{equation}
Defining $|\varphi^{\perp}\rangle$ to be the unit vector orthogonal to
$|\varphi^{\prime}\rangle$ in $\operatorname{span}\left\{  |\varphi^{\prime
}\rangle,|\zeta^{\prime}\rangle\right\}  $, we find that%
\begin{equation}
|\zeta^{\prime}\rangle=\cos(  \theta)  |\varphi^{\prime}%
\rangle+\sin(\theta)|\varphi^{\perp}\rangle,
\end{equation}
where%
\begin{equation}
\cos(  \theta)  =\langle\varphi^{\prime}|\zeta^{\prime}\rangle.
\end{equation}
Then it follows that%
\begin{align}
& c|\varphi^{\prime}\rangle\!\langle\varphi^{\prime}|-d|\zeta^{\prime}%
\rangle\!\langle\zeta^{\prime}|\nonumber\\
& =c|\varphi^{\prime}\rangle\!\langle\varphi^{\prime}|-d\left(  \cos(
\theta)  |\varphi^{\prime}\rangle+\sin(\theta)|\varphi^{\perp}%
\rangle\right)  \left(  \cos(  \theta)  \langle\varphi^{\prime
}|+\sin(\theta)\langle\varphi^{\perp}|\right)  \\
& =\left[  c-d\cos^{2}(  \theta)  \right]  |\varphi^{\prime}%
\rangle\!\langle\varphi^{\prime}|-d\sin(\theta)\cos(  \theta)
|\varphi^{\perp}\rangle\!\langle\varphi^{\prime}|\nonumber\\
& \qquad -d\sin(\theta)\cos(  \theta)  |\varphi^{\prime}\rangle
\langle\varphi^{\perp}|-d\sin^{2}(\theta)|\varphi^{\perp}\rangle\!\langle
\varphi^{\perp}|.
\end{align}
As a matrix with respect to the basis $\left\{
|\varphi^{\prime}\rangle,|\varphi^{\perp}\rangle\right\}$, the last line has
the following form:%
\begin{equation}%
\begin{bmatrix}
c-d\cos^{2}(  \theta)   & -d\sin(\theta)\cos(  \theta)
\\
-d\sin(\theta)\cos(  \theta)   & -d\sin^{2}(\theta)
\end{bmatrix}
,
\end{equation}
and this matrix has the following eigenvalues:%
\begin{align}
\lambda_{1}  & =\frac{1}{2}\left(  c-d+\sqrt{\left(  c+d\right)  ^{2}%
-4cd\cos^{2}(  \theta)  }\right)  ,\\
\lambda_{2}  & =\frac{1}{2}\left(  c-d-\sqrt{\left(  c+d\right)  ^{2}%
-4cd\cos^{2}(  \theta)  }\right)  .
\end{align}
Note that $c\geq0$ and $d\geq0$. Without loss of generality, suppose that
$c\geq d$. Then%
\begin{align}
0  & \leq4cd\sin^{2}(\theta)\\
& =4cd\left(  1-\cos^{2}(\theta)\right)  \\
\Rightarrow\qquad-2cd  & \leq2cd-4cd\cos^{2}(\theta)\\
\Rightarrow\qquad c^{2}-2cd+d^{2}  & \leq c^{2}+2cd+d^{2}-4cd\cos^{2}%
(\theta)\\
\Rightarrow\qquad\left(  c-d\right)  ^{2}  & \leq\left(  c+d\right)
^{2}-4cd\cos^{2}(\theta)\\
\Rightarrow\qquad c-d  & \leq\sqrt{\left(  c+d\right)  ^{2}-4cd\cos^{2}%
(\theta)}.
\end{align}
Then it follows that the square of the trace norm of $c|\varphi^{\prime
}\rangle\!\langle\varphi^{\prime}|-d|\zeta^{\prime}\rangle\!\langle\zeta^{\prime
}|$ is given by
\begin{align}
& \left\Vert c|\varphi^{\prime}\rangle\!\langle\varphi^{\prime}|-d|\zeta
^{\prime}\rangle\!\langle\zeta^{\prime}|\right\Vert _{1}^{2}\nonumber\\
& =\left(  \left\vert \lambda_{1}\right\vert +\left\vert \lambda
_{2}\right\vert \right)  ^{2}\\
& =\left(  \frac{1}{2}\left(  c-d+\sqrt{\left(  c+d\right)  ^{2}-4cd\cos
^{2}(  \theta)  }\right)  -\frac{1}{2}\left(  c-d-\sqrt{\left(
c+d\right)  ^{2}-4cd\cos^{2}(  \theta)  }\right)  \right)  ^{2}\\
& =\left(  c+d\right)  ^{2}-4cd\cos^{2}(  \theta)  ,
\end{align}
concluding the proof.
\end{proof}

\section{Proof of Proposition~\ref{prop:optimal-error-exponent-ch-divergence}}\label{app:proof-optimal-error-exponent-ch-divergence}

To prove the data-processing inequality, let $\mathcal{N}$ be any quantum channel. We denote by $\mathcal{N}(\mathcal{E})$ the ensemble $\{(\eta_i, \mathcal{N}(\rho_i)): i \in [r]\}$, which results from applying the channel $\mathcal{N}$ to each state in $\mathcal{E}$.
The optimal antidistinguishability error probability for the ensemble $\operatorname{Err}(\mathcal{E})$ is not more than that for the ensemble $\mathcal{N}(\mathcal{E})$. To see this, let $\mathscr{M}=\{M_1,\ldots, M_r\}$ be any POVM. We have
\begin{align}
     \operatorname{Err}(\mathscr{M};\mathcal{N}(\mathcal{E})) 
    &= \sum_{i \in [r]} \eta_i \Tr[M_i \mathcal{N}(\rho_i)] \\
    &= \sum_{i \in [r]} \eta_i \Tr[\mathcal{N}^{\dagger}(M_i) \rho_i] \\
    &\geq \operatorname{Err}(\mathcal{E}).\label{eq:measurement-channel-quantum-error}
\end{align}
The inequality \eqref{eq:measurement-channel-quantum-error} follows because $\{\mathcal{N}^{\dagger}(M_1),\ldots, \mathcal{N}^{\dagger}(M_r)\}$ is a POVM.
Since \eqref{eq:measurement-channel-quantum-error} holds for every POVM $\mathscr{M}$, we have
\begin{align}
    \operatorname{Err}(\mathcal{E}) \leq \operatorname{Err}(\mathcal{N}(\mathcal{E})).
    \label{eq:optimal-error-post-processing}
\end{align}
Therefore, for all $n\in \mathbb{N}$, we get
\begin{align}
    -\dfrac{1}{n} \ln \operatorname{Err}(\mathcal{E}^n) \geq -\dfrac{1}{n} \ln \operatorname{Err}(\mathcal{N}(\mathcal{E})^n),
\end{align}
which implies
\begin{align}
    \operatorname{E}(\rho_1,\ldots, \rho_r) \geq \operatorname{E}(\mathcal{N}(\rho_1),\ldots, \mathcal{N}(\rho_r)).
\end{align}

Now, assume that the states in the given ensemble commute with each other. The following arguments show that the optimal error of antidistinguishing the given states is equal to that of the induced probability measures. Let $P_1,\ldots, P_r$ be the probability measures on the discrete space $[\dim(\mathcal{H})]$ induced by the states in a common eigenbasis as defined in \eqref{eq:comm-states-prob-meaures}, and let $\mathcal{E}_{\operatorname{cl}}$ be the classical ensemble $\{(\eta_i, P_i): i \in [r]\}$.
Suppose $p_1,\ldots, p_r$ are the corresponding densities of the probability measures with respect to the counting measure $\mu$. 
This gives the following representation of each state
\begin{align}
    \rho_i = \int_{[\dim(\mathcal{H})]}\!\!\d\mu(\omega)\ p_i(\omega) |\omega \rangle\!\langle \omega|, \qquad i \in [r].\label{eq:commuting-states-representation}
\end{align}
We have
\begin{align}
    \operatorname{Err}(\mathscr{M}; \mathcal{E})
    &= \sum_{i \in [r]} \eta_i \Tr[M_i \rho_i]\label{eq:quantum-classical-error-equality-1} \\
    &= \sum_{i \in [r]} \eta_i \Tr\!\left[M_i \left(\int_{[\dim(\mathcal{H})]}\!\!\d\mu(\omega)\ p_i(\omega) |\omega \rangle\!\langle \omega|\right)\right] \\
    &=\int_{[\dim(\mathcal{H})]}\!\!\d\mu(\omega)\ \sum_{i \in [r]} \langle \omega|M_i|\omega \rangle \eta_i p_i(\omega)\\
    &= \operatorname{Err}_{\operatorname{cl}}(\delta; \mathcal{E}_{\operatorname{cl}}), \label{eq:quantum-classical-error-equality-2}
\end{align}
where $\delta$ is the decision rule given by $\delta(\omega)\coloneqq (\langle \omega | M_1|\omega \rangle, \ldots, \langle \omega | M_r|\omega \rangle)$. We note here that for any POVM $\mathscr{M}$, there corresponds a decision rule $\delta$ that satisfies \eqref{eq:quantum-classical-error-equality-1}--\eqref{eq:quantum-classical-error-equality-2}. Conversely, given any decision rule $\delta$ for antidistinguishing the classical ensemble $\mathcal{E}_{\operatorname{cl}}$ there corresponds a POVM $\mathscr{M}=\{M_1,\ldots, M_r\}$, given by
\begin{align}
    M_i \coloneqq \int_{[\dim(\mathcal{H})]}\!\!\d\mu(\omega)\ \delta_i(\omega) |\omega \rangle\!\langle \omega|,
\end{align}
that satisfies  \eqref{eq:quantum-classical-error-equality-1}--\eqref{eq:quantum-classical-error-equality-2}. 
This then implies
\begin{align}
    \inf_{\mathscr{M}} \operatorname{Err}(\mathscr{M}; \mathcal{E}) = \inf_{\delta} \operatorname{Err}(\delta; \mathcal{E}_{\operatorname{cl}}),
\end{align}
where the infima are taken over all POVMs $\mathscr{M}$ and decision rules $\delta$ corresponding to the given quantum and classical ensembles, respectively.
We have thus proved that
\begin{align}
     \operatorname{Err}(\mathcal{E}) =  \operatorname{Err}_{\operatorname{cl}}(\mathcal{E}_{\operatorname{cl}}),\label{eq:optimal-error-equality-classical-quantum}
\end{align}
which directly implies
\begin{align}
    \operatorname{E}(\rho_1,\ldots, \rho_r) = \operatorname{E}_{\operatorname{cl}}(P_1,\ldots, P_r).\label{eq:equality-quantum-classical-error-exponent}
\end{align}

\section{Proof of Proposition~\ref{prop:minimal-ch-divergence}}

\label{app:proof-prop:minimal-ch-divergence}
Define a map $\xi^{\prime}: \mathcal{D}^r \to [0,\infty]$ by
\begin{align}
    \xi^{\prime}(\rho_1,\ldots, \rho_r)\coloneqq \sup_{\mathcal{M}} \xi_{\operatorname{cl}}(P^{\mathcal{M}}_1,\ldots, P^{\mathcal{M}}_r)
\end{align}
as given in the right-hand side of \eqref{eq:min-CH-optim-expr}.
We first show that $\xi^{\prime}$ is a lower bound on any multivariate Chernoff divergence. Let $\xi: \mathcal{D}^r \to [0,\infty]$ be any  multivariate quantum Chernoff divergence and $\rho_1,\ldots,\rho_r$ be arbitrary quantum states. For any measurement channel $\mathcal{M}$, we have
\begin{align}
    \xi(\rho_1,\ldots, \rho_r) 
    &\geq \xi(\mathcal{M}(\rho_1),\ldots, \mathcal{M}(\rho_r)) = \xi_{\operatorname{cl}}(P_1^{\mathcal{M}},\ldots, P_r^{\mathcal{M}}).\label{eq:minimal-hellinger-inequality}
\end{align}
Here we used the assumptions that $\xi$ satisfies the data-processing inequality and reduces to the multivariate classical Chernoff divergence for commuting states. Since the inequality \eqref{eq:minimal-hellinger-inequality} holds for an arbitrary measurement channel $\mathcal{M}$,  taking the supremum over $\mathcal{M}$ gives
\begin{align}
    \xi(\rho_1,\ldots, \rho_r) \geq  \xi^{\prime}(\rho_1,\ldots,\rho_r).
\end{align}

We now show that $\xi^{\prime}$ is a multivariate quantum Chernoff divergence; i.e., it satisfies the data processing inequality and reduces to the multivariate classical Chernoff divergence for commuting states.
Consider a quantum channel $\mathcal{N}$ and any measurement channel $\mathcal{M}$ corresponding to a POVM $\{M_1,\ldots, M_t\}$ on the output Hilbert space of the channel $\mathcal{N}$. Let $\mathcal{M}_{\mathcal{N}}$ be the measurement channel corresponding to the POVM $\{\mathcal{N}^{\dagger}(M_1),\ldots, \mathcal{N}^{\dagger}(M_t)\}$.
Let $P_1^{\mathcal{M}_{\mathcal{N}}},\ldots, P_r^{\mathcal{M}_{\mathcal{N}}}$ denote the probability measures induced by $\mathcal{M}_{\mathcal{N}}$ corresponding to the states $\rho_1,\ldots, \rho_r$ as given in the development \eqref{eq:measurement-channel}--\eqref{eq:probability-measure-measurement-channel}.
Similarly, let $Q_1^{\mathcal{M}},\ldots, Q_r^{\mathcal{M}}$ denote the probability measures induced by $\mathcal{M}$ corresponding to the states $\mathcal{N}(\rho_1),\ldots, \mathcal{N}(\rho_r)$.
Since $\operatorname{Tr}[M_j \mathcal{N}(\rho_i)] = \operatorname{Tr}[\mathcal{N}^{\dag}(M_j) (\rho_i)]$ for all $i,j$, it follows that $Q_i^{\mathcal{M}}=P_i^{\mathcal{M}_{\mathcal{N}}}$ for $i \in [r]$.
This implies
\begin{align}
    \xi^{\prime}(\mathcal{N}(\rho_1),\ldots,\mathcal{N}(\rho_r))
    &= \sup_{\mathcal{M}} \xi_{\operatorname{cl}}(Q_1^{\mathcal{M}},\ldots, Q_r^{\mathcal{M}}) \\
    &= \sup_{\mathcal{M}} \xi_{\operatorname{cl}}(P_1^{\mathcal{M}_{\mathcal{N}}},\ldots, P_r^{\mathcal{M}_{\mathcal{N}}}) \\
    &\leq \xi^{\prime}(\rho_1,\ldots,\rho_r),
\end{align}
which means that $\xi^{\prime}$ satisfies the data-processing inequality. In the case when the states $\rho_1,\ldots, \rho_r$ commute, Theorem~\ref{thm:optimal-classical-general} and Proposition~\ref{prop:optimal-error-exponent-ch-divergence} give the following classical data-processing inequality
\begin{align}
    \xi_{\operatorname{cl}}(\rho_1,\ldots, \rho_r) 
    &\geq \xi_{\operatorname{cl}}(P_1^{\mathcal{M}},\ldots, P_r^{\mathcal{M}}). \label{eq:classical-data-processing-minimal-xi}
\end{align}
Also, the inequality in \eqref{eq:classical-data-processing-minimal-xi} is saturated for the measurement channel corresponding to a common eigenbasis of the commuting states. Therefore, we get
\begin{align}
    \xi^{\prime}(\rho_1,\ldots,\rho_r)= \xi_{\operatorname{cl}}(\rho_1,\ldots,\rho_r).
\end{align}
We thus conclude that $\xi^{\prime}$ is the minimal multivariate quantum Chernoff divergence.


\section{Proof of Proposition~\ref{prop:maximal-ch-divergence}}

\label{app:proof-prop:maximal-ch-divergence}

Define a map $\xi^{\prime \prime}: \mathcal{D}^r \to [0,\infty]$ by
\begin{align}
    \xi^{\prime \prime}(\rho_1,\ldots, \rho_r)\coloneqq \inf_{\substack{(\mathcal{P}, \{P_i\}_{i\in [r]})  }} \left\{\xi_{\operatorname{cl}}(P_1,\ldots,P_r) :\mathcal{P}(P_i)=\rho_i \quad \text{for all } i \in [r] \right\},
\end{align}
as given in the right-hand side of \eqref{eq:max-CH-div}.
We first show that $\xi^{\prime \prime}$ is an upper bound on any multivariate Chernoff divergence. Let $\xi: \mathcal{D}^r \to [0,\infty]$ be any  multivariate quantum Chernoff divergence and $\rho_1,\ldots,\rho_r$ be arbitrary quantum states. 
 Given a preparation channel $\mathcal{P}$ and probability measures $P_1,\ldots, P_r$ satisfying
\begin{align}
    \mathcal{P}(P_i)=\rho_i,\qquad \text{for }i\in [r],
     \label{eq:prepare-ensemble}
\end{align}
we have
\begin{align}
    \xi(\rho_1,\ldots,\rho_r) = \xi(\mathcal{P}(P_1),\ldots,\mathcal{P}(P_r)) \leq \xi_{\operatorname{cl}}(P_1,\ldots,P_r).\label{eq:maximal-hellinger-inequality}
\end{align}
In \eqref{eq:maximal-hellinger-inequality},
we used the assumptions that $\xi$ satisfies the data-processing inequality and reduces to the multivariate classical Chernoff divergence for commuting states. By taking the infimum in \eqref{eq:maximal-hellinger-inequality} over preparation channels and probability measures satisfying \eqref{eq:prepare-ensemble}, we thus get
\begin{align}
    \xi(\rho_1,\ldots,\rho_r) \leq \xi^{\prime \prime}(\rho_1,\ldots, \rho_r).
\end{align}

We now show that $\xi^{\prime \prime}$ is a multivariate quantum Chernoff divergence; i.e., it satisfies the data processing inequality and reduces to the multivariate classical Chernoff divergence for commuting states.
 Let $\mathcal{N}$ be any quantum channel.  We have
\begin{align}
    \xi^{\prime \prime}(\mathcal{N}(\rho_1),\ldots, \mathcal{N}(\rho_r)) 
    &= \inf_{\substack{(\mathcal{P}, \{P_i\}_{i\in [r]}) \\ \mathcal{P}(P_i)=\mathcal{N}(\rho_i)}} \xi_{\operatorname{cl}}(P_1,\ldots,P_r) \\
    &\leq \inf_{\substack{(\mathcal{P}, \{P_i\}_{i\in [r]}) \\ \mathcal{P}(P_i)=\rho_i}} \xi_{\operatorname{cl}}(P_1,\ldots,P_r) \\
    &= \xi^{\prime \prime}(\rho_1,\ldots, \rho_r),
\end{align}
where the inequality follows because for every preparation channel $\mathcal{P}$ satisfying $\mathcal{P}(P_i)=\rho_i$, its concatenation with $\mathcal{N}$ gives another preparation channel $\mathcal{N} \circ \mathcal{P}$ that satisfies $(\mathcal{N} \circ \mathcal{P})(P_i)=\mathcal{N}(\mathcal{P}(P_i))= \mathcal{N}(\rho_i)$.
If the states $\rho_1,\ldots,\rho_r$ commute, then by the classical data-processing inequality, for any preparation channel $\mathcal{P}$ and probability measures $P_1,\ldots, P_r$ satisfying \eqref{eq:prepare-ensemble}, we get
\begin{align}
     \xi_{\operatorname{cl}}(\rho_1,\ldots, \rho_r)= \xi_{\operatorname{cl}}(\mathcal{P}(P_1),\ldots, \mathcal{P}(P_r))
    \leq \xi_{\operatorname{cl}}(P_1,\ldots, P_r).
\end{align}
Also, the last inequality is equality for probability distributions prepared from a spectral decomposition of the commuting states in a common orthonormal basis. Therefore, we get
\begin{align}
    \xi^{\prime \prime}(\rho_1,\ldots, \rho_r) = \xi_{\operatorname{cl}}(\rho_1,\ldots, \rho_r).
\end{align}
We thus conclude that $\xi^{\prime \prime}$ is the maximal multivariate quantum Chernoff divergence.


\section{Additivity of the optimal error exponent}\label{app:additivity-optimal-error-exponent}

\begin{lemma}
\label{lem:additivity-exp}
Let $\mathcal{E}=\{(\eta_i, \rho_i): i\in [r]\}$ be an ensemble of states. 
The following equality holds%
\begin{equation}
\operatorname{E}(\rho_1,\ldots, \rho_r)=\frac{1}{\ell}\operatorname{E}(\rho_1^{\otimes \ell},\ldots, \rho_r^{\otimes \ell}) \qquad\text{for all }\ell\in\mathbb{N},
\end{equation}
where $\operatorname{E}(\rho_1,\ldots, \rho_r)$ is the optimal error exponent defined in \eqref{eq:qu-asy-error-rate-def}.

\end{lemma}

\begin{proof}
First, we have that
\begin{equation}
\operatorname{E}(\rho_1,\ldots, \rho_r) \leq \frac{1}{\ell}\operatorname{E}(\rho_1^{\otimes \ell},\ldots, \rho_r^{\otimes \ell}) \qquad\text{for all }\ell
\in\mathbb{N}, \label{eq:additivity-optimal-error-exponent-one-side}
\end{equation}
because $\left\{
-\frac{1}{n \ell} \ln\operatorname{Err}(\mathcal{E}^{n\ell}) \right\}
_{n\in\mathbb{N}}$ is a subsequence of $\left\{
-\frac{1}{n} \ln\operatorname{Err}(\mathcal{E}^{n}) \right\}
_{n\in\mathbb{N}}$. 
We now prove the converse inequality to \eqref{eq:additivity-optimal-error-exponent-one-side}. Let $\{M_{k, \ell}(1),\ldots, M_{k, \ell}(r)\}$ be a POVM attaining $\operatorname{Err}(\mathcal{E}^{k\ell})$ for all $k,\ell \in \mathbb{N}$. Then for any $n \in \mathbb{N}$ with $n \geq \ell$, we have
\begin{align}
    \operatorname{Err}(\mathcal{E}^n) 
        &\leq \sum_{i \in [r]} \eta_i \operatorname{Tr}\!\left[ \rho_i^{\otimes n} \left(M_{\lfloor \frac{n}{\ell} \rfloor, \ell}(i) \otimes \mathbb{I}^{\otimes (n-\lfloor \frac{n}{\ell} \rfloor)} \right) \right] \\
        &= \sum_{i \in [r]} \eta_i \operatorname{Tr}\!\left[ \rho_i^{\otimes \lfloor \frac{n}{\ell} \rfloor \ell} M_{\lfloor \frac{n}{\ell} \rfloor, \ell}(i) \right] \\
        &= \operatorname{Err}(\mathcal{E}^{\lfloor \frac{n}{\ell} \rfloor\ell}).
\end{align}
This implies
\begin{align}
    \liminf_{n \to \infty} -\dfrac{1}{n}\ln \operatorname{Err}(\mathcal{E}^{n}) \geq \liminf_{n \to \infty} -\dfrac{1}{\lfloor \frac{n}{\ell} \rfloor\ell}\ln \operatorname{Err}(\mathcal{E}^{\lfloor \frac{n}{\ell} \rfloor\ell}) = \dfrac{1}{\ell} \liminf_{k \to \infty} -\dfrac{1}{k}\ln \operatorname{Err}(\mathcal{E}^{k\ell}).
\end{align}
This completes the proof.
\end{proof}


\section{Limit of the regularized maximal multivariate quantum Chernoff divergence}\label{app:limits-regularized-ch-divergences}
Here we provide a proof of equation~\eqref{eq:limit-maximal-regularized-ch-divergence}. We first observe that the multivariate classical Chernoff divergence is subadditive; i.e.,
\begin{align}
    \xi_{\operatorname{cl}}(P_1 \otimes Q_1,\ldots, P_r \otimes Q_r) \leq \xi_{\operatorname{cl}}(P_1,\ldots, P_r)+ \xi_{\operatorname{cl}}( Q_1,\ldots, Q_r)
\end{align}
for any sets of probability densities $\{P_1,\ldots, P_r\}$ and $\{Q_1,\ldots, Q_r\}$ on a measureable space $(\Omega, \mathcal{A})$. This follows easily from the definitions of the Hellinger transform \eqref{eq:hellinger-trans-def} and multivariate Chernoff divergence \eqref{eq:classical-CH-divergence}.
So, from the definition \eqref{eq:max-CH-div}, we have for $\ell,m \in \mathbb{N}$ that
\begin{align}
    & \xi_{\operatorname{max}}(\rho_1^{\otimes (\ell+m)}, \ldots, \rho_r^{\otimes (\ell+m)}) \notag \\
    &= \inf_{\substack{(\mathcal{P}^{(\ell+m)}, \{P_i^{(\ell+m)}\}_{i\in [r]}) \\ \mathcal{P}^{(\ell+m)}(P_i^{(\ell+m)})=\rho_i^{\otimes \ell} \otimes \rho_i^{\otimes m}}}  \xi_{\operatorname{cl}}(P_1^{(\ell+m)},\ldots, P_r^{(\ell+m)}) \\
    &\leq \inf_{\substack{(\mathcal{P}^{(\ell)} \otimes \mathcal{P}^{(m)}, \{P_i^{(\ell)} \otimes P_i^{(m)}\}_{i\in [r]}) \\ \mathcal{P}^{(\ell)}(P_i^{(\ell)})=\rho_i^{\otimes \ell}, \mathcal{P}^{(m)}(P_i^{(m)})=\rho_i^{\otimes m}}}  \xi_{\operatorname{cl}}(P_1^{(\ell)}\otimes P_1^{(m)} ,\ldots, P_r^{(\ell)}\otimes P_r^{(m)}) \\
    &\leq \inf_{\substack{(\mathcal{P}^{(\ell)} \otimes \mathcal{P}^{(m)}, \{P_i^{(\ell)} \otimes P_i^{(m)}\}_{i\in [r]}) \\ \mathcal{P}^{(\ell)}(P_i^{(\ell)})=\rho_i^{\otimes \ell}, \mathcal{P}^{(m)}(P_i^{(m)})=\rho_i^{\otimes m}}}\left(  \xi_{\operatorname{cl}}(P_1^{(\ell)} ,\ldots, P_r^{(\ell)}) + \xi_{\operatorname{cl}}( P_1^{(m)} ,\ldots, P_r^{(m)}) \right)\\
    &= \inf_{\substack{(\mathcal{P}^{(\ell)}, \{P_i^{(\ell)}\}_{i\in [r]}) \\ \mathcal{P}^{(\ell)}(P_i^{(\ell)})=\rho_i^{\otimes \ell}}}\left(  \xi_{\operatorname{cl}}(P_1^{(\ell)} ,\ldots, P_r^{(\ell)}) \right)+ \inf_{\substack{(\mathcal{P}^{(m)}, \{P_i^{(m)}\}_{i\in [r]}) \\ \mathcal{P}^{(m)}(P_i^{(m)})=\rho_i^{\otimes m}}}\left(  \xi_{\operatorname{cl}}(P_1^{(m)} ,\ldots, P_r^{(m)}) \right)\\
    &= \xi_{\operatorname{max}}(\rho_1^{\otimes \ell}, \ldots, \rho_r^{\otimes \ell}) + \xi_{\operatorname{max}}(\rho_1^{\otimes m}, \ldots, \rho_r^{\otimes m}) .
\end{align}
We have thus proved that the sequence $\left( \xi_{\operatorname{max}}(\rho_1^{\otimes \ell}, \ldots, \rho_r^{\otimes \ell}) \right)_{\ell \in \mathbb{N}}$ is subadditive. 
It then follows from Fekete's subadditive lemma \cite{fekete1923distribution} that
the limit $\lim_{ \ell \to \infty} \xi_{\operatorname{max}}(\rho_1^{\otimes \ell}, \ldots, \rho_r^{\otimes \ell})/\ell$ exists and is given by
\begin{align}
    \lim_{\ell \to \infty} \dfrac{1}{\ell}\xi_{\operatorname{max}}(\rho_1^{\otimes \ell}, \ldots, \rho_r^{\otimes \ell}) = \inf_{\ell \in \mathbb{N}} \dfrac{1}{\ell} \xi_{\operatorname{max}}(\rho_1^{\otimes \ell}, \ldots, \rho_r^{\otimes \ell}) .
\end{align}

\section{Properties of the extended max-relative entropy} \label{app:properties_extended_dmax}

\begin{proposition}[Data-processing inequality]
    Let $X\neq 0$ be a Hermitian operator and $\sigma$ a positive semi-definite operator.
Let $\mathcal{N}$ be a positive map (a special case of which is a quantum channel, i.e., a completely positive and trace preserving map). Then
\begin{equation}
     D_{\max}(X \| \sigma)  \geq D_{\max}(\mathcal{N}(X) \| \mathcal{N}(\sigma)).
\end{equation}
\end{proposition}

\begin{proof}
    A special case of this inequality follows from  \cite[Lemma~2]{ww2020} by taking the limit $\alpha \to \infty$. Here we prove it for all positive maps.
Suppose that $\lambda \geq 0$ is such that $-\lambda \sigma \leq X \leq \lambda \sigma$ (for this inequality to be satisfied, we see that it is required for $\lambda \geq 0$ because $\lambda\sigma \geq -\lambda\sigma$ implies $\lambda \geq 0$). Then the following inequality holds $-\lambda \mathcal{N}(\sigma) \leq \mathcal{N}(X) \leq \lambda \mathcal{N}(\sigma)$, from the assumption that $\mathcal{N}$ is a positive map.
Consequently, we get 
\begin{align}
    D_{\max}(X \| \sigma) 
        &= \ln \inf_{\lambda\geq 0}\{\lambda : -\lambda \sigma \leq X \leq \lambda \sigma\} \\
        &\geq \ln \inf_{\lambda\geq 0} \{\lambda : -\lambda \mathcal{N}(\sigma) \leq \mathcal{N}(X) \leq \lambda \mathcal{N}(\sigma)\} \\
        &= D_{\max}(\mathcal{N}(X) \| \mathcal{N}(\sigma)),
\end{align}
concluding the proof.
\end{proof}

\begin{proposition}[Joint quasi-convexity]
Let $\mathscr{X}$ be a finite alphabet and $p$  a probability distribution on $\mathscr{X}$.
Let $X^x\neq 0$ and $\sigma^x$ be Hermitian and positive semi-definite operators, respectively, for all $x \in \mathscr{X}$. Then
\begin{equation}
    \max_{x \in \mathscr{X}} D_{\max}(X^x \| \sigma^x) \geq D_{\max}\left(\sum_{x \in \mathscr{X}} p(x) X^x \bigg \| \sum_{x \in \mathscr{X}} p(x) \sigma^x \right).
\end{equation}
\end{proposition}
\begin{proof}
    
If $\lambda \geq 0$ satisfies $-\lambda \sigma^x \leq X^x \leq \lambda \sigma^x$ for all $x \in \mathscr{X}$, then we also have $-\lambda \sum_{x \in \mathscr{X}} p(x) \sigma^x \leq \sum_{x \in \mathscr{X}} p(x) X^x \leq \lambda \sum_{x \in \mathscr{X}} p(x) \sigma^x$.
This gives
\begin{align}
    D_{\max}\left(\sum_{x \in \mathscr{X}} p(x) X^x \bigg \| \sum_{x \in \mathscr{X}} p(x) \sigma^x \right)
        &= \ln \inf_{\lambda \geq 0} \bigg \{\lambda: -\lambda \sum_{x \in \mathscr{X}} p(x) \sigma^x \leq \sum_{x \in \mathscr{X}} p(x) X^x \leq \lambda \sum_{x \in \mathscr{X}} p(x) \sigma^x  \bigg\} \\
        &\leq  \ln \inf_{\lambda \geq 0} \big \{\lambda: -\lambda  \sigma^x \leq  X^x \leq \lambda  \sigma^x, \forall x \in \mathscr{X} \big\}  \\
        &= \max_{x \in \mathscr{X}}\ln \inf_{\lambda \geq 0} \big \{\lambda: -\lambda  \sigma^x \leq  X^x \leq \lambda  \sigma^x \big\} \\
        &= \max_{x \in \mathscr{X}} D_{\max}(X^x \| \sigma^x),
\end{align}
concluding the proof.
\end{proof}

\begin{proposition}[Lower semi-continuity]
Let $X$ be a Hermitian operator and $\sigma$ a positive semi-definite operator.
    For every sequence $(X_n)_{n \in \mathbb{N}}$ of Hermitian operators converging to $X$ in the operator norm,
    \begin{equation}
        \liminf_{n \to \infty} D_{\max}(X_n \| \sigma) \geq D_{\max}(X \| \sigma).
    \end{equation}
\end{proposition}

\begin{proof}
    Let $(X_n)_{n \in \mathbb{N}}$ be an arbitrary sequence of Hermitian operators
that converges to $X$ in the operator norm $\left\|\cdot\right\|_\infty$.
Note that the convergence of the sequence is independent of the choice of the norm due to the finite dimensionality of the underlying Hilbert space.
If $\operatorname{supp}(X_n) \nsubseteq \operatorname{supp}(\sigma)$ for infinitely many $n$ then we have
$\liminf_{n \to \infty} D_{\max}(X_n \| \sigma)=\infty$, and hence $\liminf_{n \to \infty} D_{\max}(X_n \| \sigma) \geq D_{\max}(X\| \sigma)$ is trivially satisfied in this case.
Thus, we assume without loss of generality that $\operatorname{supp}(X_n) \subseteq \operatorname{supp}(\sigma)$ for large $n$.
For all $|v\rangle \in \mathcal{H}\backslash \operatorname{supp}(\sigma)$, we have
$X_n |v\rangle=0$ for large $n$, which implies $X |v\rangle=\lim_{n \to \infty} X_n |v\rangle=0$.
Consequently, we get $\operatorname{supp}(X) \subseteq \operatorname{supp}(\sigma)$.
Thus, we have
\begin{align}
    \liminf_{n \to \infty} D_{\max}(X_n \| \sigma)
        &= \lim_{n \to \infty} \ln \left\|\sigma^{-1/2} X_n \sigma^{-1/2}\right\|_\infty \\
        &= \ln \left\|\sigma^{-1/2} X \sigma^{-1/2}\right \|_\infty \\
        &= D_{\max}(X \| \sigma).
\end{align}
We have thus shown that, for every sequence $(X_n)_{n \in \mathbb{N}}$ of Hermitian operators converging to $X$, we have $\liminf_{n \to \infty} D_{\max}(X_n \| \sigma) \geq D_{\max}(X \| \sigma)$.
This proves lower semi-continuity of the extended max-relative entropy $D_{\max}(X \| \sigma)$ in $X$.
\end{proof}

\begin{proposition}[Non-negativity and faithfulness]
    Let $X$ be a Hermitian operator of unit trace, and let $\sigma$ be a quantum state. Then $D_{\max}(X \| \sigma) \geq 0$. Also, under the same conditions, $D_{\max}(X \| \sigma) = 0$ if and only if $X = \sigma$.
\end{proposition}

\begin{proof}
    
For every $\lambda \geq 0$ satisfying $-\lambda \sigma \leq X \leq \lambda \sigma$, we have that 
    $\lambda =  \operatorname{Tr}[\lambda \sigma] \geq \operatorname{Tr}X =1,$
implying $\ln \lambda \geq 0$.
By definition, we then get $D_{\max}(X \| \sigma) \geq 0$.

If $X=\sigma$, then it trivially follows by definition that $D_{\operatorname{max}}(X \| \sigma)=0$.
Conversely, suppose that $D_{\operatorname{max}}(X \| \sigma)=0$.
This implies $-\sigma \leq X \leq \sigma$, and hence $\sigma - X \geq 0$.
By the Helstrom–Holevo Theorem \cite[Eq.~$(5.1.17)$]{KW20}, and the fact that $\operatorname{Tr}[\sigma - X]=0$, we get
\begin{align}
    \frac{1}{2} \|\sigma-X\|_1 
        &= \sup_{M \geq 0} \{\operatorname{Tr}[M(\sigma-X)]: M \leq \mathbb{I}\} \label{eq:primal_sigma_minus_x} \\
        &\leq \inf_{Y \geq 0} \{\operatorname{Tr}[Y]: Y \geq \sigma-X\}, \label{eq:dual_sigma_minus_x}
\end{align}
where the last inequality follows by the weak duality of the SDP given in \eqref{eq:primal_sigma_minus_x}.
A feasible point in $\eqref{eq:dual_sigma_minus_x}$ is given by $Y=\sigma - X$,
and we have $\operatorname{Tr}[Y]=\operatorname{Tr}[\sigma-X]=0$.
It thus follows from \eqref{eq:dual_sigma_minus_x} that $\|\sigma-X \|_1 \leq 0$, which implies $\|\sigma-X \|_1=0$.
We have thus shown that $\sigma=X$.

\end{proof}

\begin{proposition}[Monotonicity]
    Let $X$ be a Hermitian operator and $\sigma', \sigma$ be positive semi-definite operators such that $\sigma' \leq \sigma$.
    Then
    \begin{equation}
    D_{\max}(X \| \sigma) \leq D_{\max}(X \| \sigma').
    \end{equation}
\end{proposition}

\begin{proof}
    Given an arbitrary $\lambda \geq 0$ that satisfies $-\lambda \sigma' \leq X \leq \lambda \sigma'$, this $\lambda$ also satisfies
    $-\lambda \sigma \leq X \leq \lambda \sigma$.
Consequently,
\begin{align}
    D_{\max}(X \| \sigma) 
        & = \ln \inf_{\lambda \geq 0} \{\lambda: -\lambda \sigma \leq X \leq \lambda \sigma\} \\
        & \leq \ln \inf_{\lambda \geq 0} \{\lambda: -\lambda \sigma' \leq X \leq \lambda \sigma' \} \\
        & = D_{\max}(X \| \rho),
\end{align}
concluding the proof.
\end{proof}

    If $A, B$ are Hermitian operators on a Hilbert space $\mathcal{H}$ then it is easy to prove that the kernel of their tensor product is given by $\operatorname{ker}(A \otimes B) = \operatorname{ker}(A) \otimes \mathcal{H} + \mathcal{H} \otimes \operatorname{ker}(B)$.
    We use this observation in the proof of the next property.
\begin{proposition}[Additivity]
    Let $X_1, X_2$ be non-zero Hermitian operators, and let $\sigma_1, \sigma_2$ be non-zero positive semi-definite operators.
    Then
    \begin{align}
        D_{\max} (X_1 \otimes X_2 \| \sigma_1 \otimes \sigma_2) = D_{\max} (X_1  \| \sigma_1)+ D_{\max} ( X_2 \| \sigma_2).
    \end{align}
\end{proposition}

\begin{proof}
    First, suppose that $\operatorname{supp}(X_1) \nsubseteq \operatorname{supp}(\sigma_1)$.
    This implies that $\operatorname{supp}(X_1 \otimes X_2) \nsubseteq \operatorname{supp}(\sigma_1 \otimes \sigma_2)$.
    Indeed, let $|x_1 \rangle \in \operatorname{supp}(X_1) \backslash \operatorname{supp}(\sigma_1)$.
    Also, $X_2 \neq 0$ implies that there exists a non-zero vector $|x_2 \rangle \in \operatorname{supp}(X_2)$.
    We thus have $(X_1 \otimes X_2)(|x_1 \rangle \otimes |x_2 \rangle) \neq 0$ and $(\sigma_1 \otimes \sigma_2)(|x_1 \rangle \otimes |x_2 \rangle) = 0$, implying that $\operatorname{supp}(X_1 \otimes X_2) \nsubseteq \operatorname{supp}(\sigma_1 \otimes \sigma_2)$.
    Also, the assumption that $X_2$ and $\sigma_2$ are non-zero implies that $D_{\max}(X_2 \|\sigma_2) > -\infty$.
    Therefore, in this case, both $D_{\max}(X_1 \otimes X_2 \|\sigma_1 \otimes \sigma_2)$ and $D_{\max}(X_1 \|\sigma_1)+ D_{\max}(X_2 \|\sigma_2)$ are equal to $\infty$.
    We also get by similar arguments for the case $\operatorname{supp}(X_2) \nsubseteq \operatorname{supp}(\sigma_2)$ that both $D_{\max}(X_1 \otimes X_2 \|\sigma_1 \otimes \sigma_2)$ and $D_{\max}(X_1 \|\sigma_1)+ D_{\max}(X_2 \|\sigma_2)$ are equal to $\infty$.

    To complete the proof, we now consider the case when $\operatorname{supp}(X_1) \subseteq \operatorname{supp}(\sigma_1)$ and $\operatorname{supp}(X_2) \subseteq \operatorname{supp}(\sigma_2)$.
    In this case, we have $\operatorname{supp}(X_1 \otimes X_2) \subseteq \operatorname{supp}(\sigma_1 \otimes \sigma_2)$.
    This is because we have $\operatorname{ker}(\sigma_1) \subseteq \operatorname{ker}(X_1)$ and $\operatorname{ker}(\sigma_2) \subseteq \operatorname{ker}(X_2)$, which gives
    \begin{align}
        \operatorname{ker}(\sigma_1 \otimes \sigma_2)
            &= \operatorname{ker}(\sigma_1) \otimes \mathcal{H} + \mathcal{H} \otimes \operatorname{ker}(\sigma_2) \\
            &\subseteq \operatorname{ker}(X_1) \otimes \mathcal{H} + \mathcal{H} \otimes \operatorname{ker}(X_2) \\
            &= \operatorname{ker}(X_1 \otimes X_2).
    \end{align}
    We thus have
    \begin{align}
        D_{\max} (X_1 \otimes X_2 \| \sigma_1 \otimes \sigma_2) 
            &= \ln \big\|(\sigma_1^{-1/2} \otimes \sigma_2^{-1/2}) (X_1 \otimes X_2) (\sigma_1^{-1/2} \otimes \sigma_2^{-1/2}) \big\|_\infty \\
            &= \ln \big\|\sigma_1^{-1/2} X_1 \sigma_1^{-1/2} \otimes  \sigma_2^{-1/2} X_2 \sigma_2^{-1/2} \big\|_\infty \\
            &= \ln \left(\big\|\sigma_1^{-1/2} X_1 \sigma_1^{-1/2}\|_\infty \cdot \|\sigma_2^{-1/2} X_2 \sigma_2^{-1/2} \big\|_\infty \right) \\
            &= \ln \big\|\sigma_1^{-1/2} X_1 \sigma_1^{-1/2} \big\|_\infty + \ln \big\|\sigma_2^{-1/2} X_2 \sigma_2^{-1/2} \big\|_\infty\\
            &= D_{\max}(X_1 \| \sigma_1) + D_{\max}(X_2 \| \sigma_2),
    \end{align}
    concluding the proof.
\end{proof}


\section{Proof of Equation~\eqref{eq:d-convex-comb-log-euclid}}

\label{app:max-rel-entropy-bounds-similarity}

Let $\omega \in \mathcal{D}$ be arbitrary and $(s_1,\ldots, s_r)\in \mathbb{R}^r$ be any probability vector. Since the quantum states $\rho_1,\ldots, \rho_r$ have full support, we have
\begin{align}
&  \sum_{i\in [r]}s_{i}D(\omega\Vert\rho_{i})\nonumber\\
&  =\sum_{i\in [r]}s_{i}\operatorname{Tr}[\omega(\ln\omega-\ln\rho_{i})]\\
&  =\operatorname{Tr}[\omega\ln\omega]-\operatorname{Tr}\!\left[  \omega\left(
\sum_{i\in [r]}s_i\ln\rho_{i}\right)  \right]  \\
&  =\operatorname{Tr}[\omega\ln\omega]-\operatorname{Tr}\!\left[  \omega\ln
\exp\left(  \sum_{i\in [r]}s_i\ln\rho_{i}\right)  \right]  \\
&  =\operatorname{Tr}[\omega\ln\omega]-\operatorname{Tr}\!\left[  \omega
\ln\left(  \frac{\exp\left(  \sum_{i\in [r]}s_i\ln\rho_{i}\right)  }%
{\Tr\!\left[\exp\left(  \sum_{i\in [r]}s_i\ln\rho_{i}\right) \right]}\cdot \Tr\!\left[\exp\left(  \sum_{i\in [r]}s_i\ln\rho_{i}\right) \right]\right)  \right]  \\
&  =\operatorname{Tr}[\omega\ln\omega]-\operatorname{Tr}\!\left[  \omega
\ln\left(  \frac{\exp\left(  \sum_{i\in [r]}s_i\ln\rho_{i}\right)  }%
{\Tr\!\left[\exp\left(  \sum_{i\in [r]}s_i\ln\rho_{i}\right) \right]}\right)  \right]  -\ln
\Tr\!\left[\exp\left(  \sum_{i\in [r]}s_i\ln\rho_{i}\right) \right]\\
&  =D\left(  \omega\middle\Vert\frac{\exp\left(  \sum_{i\in [r]}s_i\ln
\rho_{i}\right)  }{\Tr\!\left[\exp\left(  \sum_{i\in [r]}s_i\ln\rho_{i}\right) \right]}\right)
-\ln \Tr\!\left[\exp\left(  \sum_{i\in [r]}s_i\ln\rho_{i}\right) \right]\\
&  \geq-\ln \Tr\!\left[\exp\left(  \sum_{i\in [r]}s_i\ln\rho_{i}\right) \right],
\end{align}
where the inequality follows from the non-negativity of quantum relative
entropy for quantum states. The lower bound is achieved by picking
$\omega=\frac{\exp\left(  \sum_{i\in [r]}s_i\ln\rho_{i}\right)  }%
{\Tr\left[\exp\left(  \sum_{i\in [r]}s_i\ln\rho_{i}\right) \right]}$, so that%
\begin{align}
&  \inf_{\omega \in \mathcal{D}}\sum_{i\in [r]}
s_{i}D(\omega\Vert\rho_{i})\nonumber\\
&  =\inf_{\omega \in \mathcal{D}}D\left(
\omega\middle\Vert\frac{\exp\left(  \sum_{i\in [r]}s_i\ln\rho_{i}\right)
}{\Tr\!\left[\exp\left(  \sum_{i\in [r]}s_i\ln\rho_{i}\right) \right]}\right)  -\ln
\Tr\!\left[\exp\left(  \sum_{i\in [r]}s_i\ln\rho_{i}\right) \right]\\
&  =-\ln \Tr\!\left[\exp\left(  \sum_{i\in [r]}s_i\ln\rho_{i}\right) \right].
\end{align}
This directly gives \eqref{eq:d-convex-comb-log-euclid}.



\end{document}